\documentclass[acmsmall]{acmart}

\AtBeginDocument{%
  }

\usepackage{graphicx}
\usepackage{amsmath}
\usepackage{amsthm}
\usepackage{url}
\usepackage{xspace}
\usepackage{balance}
\usepackage{makecell}
\usepackage{multirow}
\usepackage{subfig}
\usepackage{setspace}
\usepackage{hyperref}
\usepackage{enumitem}
\usepackage{xcolor}
\usepackage{booktabs}
\usepackage[linesnumbered,ruled,noend]{algorithm2e}

\newcolumntype{M}[1]{>{\centering\arraybackslash}m{#1}}

\graphicspath{{images/}}

\SetKwProg{Fn}{Function}{}{}	
\SetKwFor{ForEach}{for each}{do}{endfch}
\SetKwInput{KwInput}{Input}
\SetKwInput{KwOutput}{Output}    
\DontPrintSemicolon

\setlist{nolistsep,leftmargin=*}

\newtheorem{thm}{\textbf{Theorem}}
\newtheorem{lem}{\textbf{Lemma}}
\newtheorem{prob}{\textbf{Problem}}

\newcommand{\figwidth}{0.12\textwidth}
\newcommand{\examplewidth}{0.22\textwidth}
\newcommand{\gc}{\textsc{ConfiGS}\xspace}
\newcommand{\cn}{\name}
\newcommand{\name}{CGS\xspace}
\newcommand{\iu}{\textsc{\name-E}\xspace}		% 0
\newcommand{\inter}{\textsc{\name-I}\xspace}	% -1
\newcommand{\uni}{\textsc{\name-U}\xspace}		% +1

\newcommand{\insertkey}{Insert-Key\xspace}
\newcommand{\exmax}{Extract-Max\xspace}
\newcommand{\deckey}{Decrease-Key\xspace}
\newcommand{\inckey}{Increase-Key\xspace}
\newcommand{\eps}{\ensuremath{\delta}\xspace}
\newcommand{\dist}{\ensuremath{d}\xspace}

\newcommand{\rl}{\ensuremath{nl}\xspace}

\newcommand{\nlt}{\ensuremath{\Delta}\xspace}

\newcommand{\maxd}{\ensuremath{{d^*}}\xspace}
\newcommand{\ic}{\ensuremath{ic}\xspace}
\newcommand{\uc}{\ensuremath{uc}\xspace}
\newcommand{\fp}{\ensuremath{\operatorname{\emph{var}}}\xspace}
\newcommand{\ls}{\fp}

\newcommand{\greedycn}{\cn\xspace}
\newcommand{\issafemerge}{IsSafeMerge\xspace}
\newcommand{\updatedegree}{UpdateDegree\xspace}
\newcommand{\updateheap}{UpdateCG\xspace}
\newcommand{\merge}{MergeNodes\xspace}
\newcommand{\decomp}{Reconstruction\xspace}

\newcommand{\snmap}{Anc\xspace}

\newcommand{\desc}{Desc\xspace}
\newcommand{\child}{Child\xspace}
\newcommand{\nq}{NeighborhoodQuery\xspace}
\newcommand{\computecg}{ComputeCG\xspace}
\newcommand{\buildheap}{DiscoverCN\xspace}
\newcommand{\nic}{ComputeNIC\xspace}

\newcommand{\true}{\emph{true}\xspace}
\newcommand{\false}{\emph{false}\xspace}

\newcommand{\cutfigabove}{\vspace*{-6mm}}
\newcommand{\cutfigbelow}{\vspace*{-4mm}}
\newcommand{\cuttabbelow}{\vspace*{-5mm}}

\setcopyright{cc}
\setcctype{by}
\copyrightyear{2025}
\acmYear{2025}
\acmDOI{10.1145/3786788}

\acmJournal{TKDD}
\acmVolume{0}
\acmNumber{0}
\acmArticle{0}
\acmMonth{0}

\begin{document}

\title[CGS: Configurable Graph Summarization with Bounded Neighborhood Loss and Query Support]{CGS: Configurable Graph Summarization with Bounded Neighborhood Loss and Query Support}

\author{Shubhadip Mitra}
\orcid{0000-0002-0444-9652}
\affiliation{%
  \institution{Blue Yonder India Pvt. Ltd.}
  \city{Bengaluru}
  \country{India}
}
\email{shubhadip.mitra@blueyonder.com}

\author{Sona Elza Simon}
\orcid{0000-0001-5136-7273}
\affiliation{%
  \institution{Centre for Machine Intelligence and Data Science, Indian Institute of Technology Bombay}
  \city{Mumbai}
  \country{India}
}
\email{sonasimonp@gmail.com}

\author{C Oswald}
\orcid{0000-0002-1251-1495}
\affiliation{%
  \institution{Dept. of Computer Science and Engineering, National Institute of Technology Tiruchirappalli}
  \city{Tiruchirapalli}
  \country{India}
}
\email{oswald@nitt.edu}

\author{Arnab Bhattacharya}
\orcid{0000-0001-7331-0788}
\affiliation{%
  \institution{Dept. of Computer Science and Engineering, Indian Institute of Technology Kanpur}
  \city{Kanpur}
  \country{India}
}
\email{arnabb@cse.iitk.ac.in}

\author{Arindam Pal}
\orcid{0000-0001-5710-7464}
\affiliation{%
  \institution{TechSoftX Corporation}
  \city{Sydney}
  \state{New South Wales}
  \country{Australia}
}
\email{arindamp@gmail.com}

\renewcommand{\shortauthors}{Shubhadip Mitra, Sona Elza Simon, C Oswald, Arnab Bhattacharya, and Arindam Pal}

\begin{abstract}
	Given a large graph, how to generate a compact summary graph that is
	configurable by the user and supports multiple graph queries with either
	\emph{no loss} or with \emph{high accuracy}? The ever growing size of graph
	datasets makes the above question on graph summarization very pertinent.
	Although, there are several approaches, there does not exist a
	\emph{configurable} graph summarization method that offers high compression
	along with support for multiple graph queries on the summary graph with
	high accuracy, and allows the user to configure the summarization based on:
	(1)~lossless or lossy summarization, (2)~amount of tolerable neighborhood
	loss, (3)~the type of loss it can tolerate, in terms of \emph{false
	positive edges} (i.e., extra edges), \emph{false negative edges} (i.e.,
	missing edges), or neither, in both the (a)~reconstructed graph and the
	(b)~query answers.  To overcome these limitations, we propose a novel graph
	summarization framework \emph{\cn (Configurable Graph Summarizer)} that
	builds upon the idea of aggregating nodes with \emph{common neighborhoods}.
	The \cn framework consists of three summarization variants, \iu, \inter and
	\uni. While \iu is a lossless scheme, \inter and \uni are lossy  schemes
	that allow reconstruction of the input graph with \emph{no false positive}
	edges and \emph{no false negative} edges, respectively. To bound the graph
	reconstruction loss, we introduce a user-specified parameter
	\emph{neighborhood loss tolerance threshold}, that limits the maximum loss
	allowed in the	neighborhood of each node. This allows graph reconstruction
	and neighborhood query evaluation with either \emph{no loss} or with
	\emph{bounded loss guarantees}. This, in turn, enables retrieval of
	multiple graph queries such as shortest path and reachability queries with
	either no loss or with fairly high accuracy. Empirical evaluation on
	several synthetic and real-world graphs shows that \cn offers superior
	summarization than the state-of-the-art methods, and can answer graph
	queries with fairly high accuracy and efficiency. The implementation code
	and the datasets are available at
	\url{https://github.com/sonaelzasimon/CGS_Configurable_Graph_Summarization}.
\end{abstract}




\received{7th February, 2025}
\received[Revised]{1st September, 2025}
\received[Accepted]{24th October, 2025}

\maketitle

\section {Introduction}
\label{sec:intro}

One of the key challenges in graph data mining is the ever increasing
size of the graphs that prohibits its scalability in terms of storage,
communication and analysis. To address this challenge, several graph
summarization schemes have been proposed
\cite{navlakha2008graph,khan2017summarizing,lee2022slugger,
lai2023optimized} that generate a compact summary of a graph that allows
reconstruction of the original graph with minimal or no loss. A summary
graph is typically produced by aggregation of nodes and edges of the
input graph. Due to its numerous applications \cite{liu2018graph}, graph
summarization has remained an area of active research for the past two
decades.

\subsection{Background and Motivation}

A graph summarization scheme is
defined as follows.  An input graph $G$ is \emph{summarized} to a graph $G_s$
that is later \emph{decompressed} or \emph{expanded} to produce a graph $G_r$, which is
referred to as the \emph{reconstructed graph}. A graph summarization scheme (along with
its corresponding reconstruction) is said to be \emph{lossless} if and only if
$G_r=G$; else, it is \emph{lossy}.

Though there are several graph summarization works, most of them lack
\emph{configurability} as required by the user with \emph{no} or
\emph{limited } support for queries. They are either lossless
\cite{lee2022slugger,khan2015set,fan2012query, lai2023optimized,
berberidis2022summarizing, ramezani2022graphguess}, or lossy
\cite{nejad2021graph,li2019graph,lee2020ssumm, anagnostopoulos2024general}. There are limited works that allow both
lossless as well as lossy summarization
\cite{shin2019sweg,navlakha2008graph}. Moreover, majority of the lossy
summarization schemes do not allow the user to control the loss
\cite{nejad2021graph,li2019graph,yong2021efficient, anagnostopoulos2024general, lai2023optimized}. Hence, these works fail
to provide bounded quality guarantees on queries.  Further, none of the
existing works allow the flexibility to choose the type of loss in terms
of either \emph{false positive} edges (i.e., extra edges), or
\emph{false negative edges} (i.e., missing edges) in the reconstructed
graph. Finally, there are many works that offer good summarization but
provide no support for queries
\cite{navlakha2008graph,lee2022slugger,yong2021efficient,lee2020ssumm, lai2023optimized, berberidis2022summarizing}. 

In particular, through an extensive review of the prior works (Sec.~\ref{sec:related}), we observe that there does not exist a general purpose \emph{configurable} graph summarization framework that offers high compression along with \emph{all} the following desirable capabilities, where the user is allowed to:
\begin{enumerate}[label=(\arabic*)]
	\item \emph{choose} between \emph{lossless} and \emph{lossy} summarization,
	\item \emph{control} the amount of \emph{neighborhood loss} that it can tolerate,
	\item \emph{run graph queries} on the summary graph (without decompressing it) including \emph{neighborhood} queries, \emph{reachability} queries and \emph{shortest path} queries, and
	\item \emph{specify} the \emph{type of loss} it can tolerate---\emph{false positive edges} (i.e., extra edges) or \emph{false negative edges} (i.e., missing edges)---in the reconstructed graph or query answers.
\end{enumerate} 

In this work, we introduce the \gc (Configurable Graph Summarization)
problem that achieves all the above desirable properties and
proposes a solution framework, \textbf{\cn} \emph{(Configurable Graph
Summarizer)}.  The framework offers three summarization variants:
\begin{enumerate}[label=(\arabic*)]
	\item \textbf{\iu} (\name-Exact): a \emph{lossless} scheme that
		allows neither false positive edges nor false negative edges in
		the reconstructed graph
	\item \textbf{\inter} (\name-Intersection): a lossy scheme that
		allows only false negative edges, but no false positive edges
	\item \textbf{\uni} (\name-Union): a lossy scheme that allows only
		false positive edges, but no false negative edges
\end{enumerate}

Thus, the two lossy schemes in our framework have only \emph{one-sided
errors}. Hence, an application can judiciously select whatever fits its
requirements the best. For example, a friend recommendation application
may tolerate false positive answers, but not false negative answers,
while a route navigation application may tolerate false negative edges,
but not false positive edges in its recommended route.

The \name framework leverages the fact that several nodes in a graph \emph{share
common neighbors}.  Hence, nodes that share common neighbors can be
aggregated into a single supernode.  The exact details of the schemes are
described in Sec.~\ref{sec:framework}. In contrast to many prior works \cite{navlakha2008graph,khan2015set,shin2019sweg,ko2020incremental,yong2021efficient}, this work does not use correction edges to control the reconstruction loss, which avoids the space overhead to store the correction edges. This in turn, leads to superior compression ratio, as shown in Sec.~\ref{sec:results}.  

\subsection{Major Contributions}

The key contributions of this work are as follows:

\begin{enumerate}[label=(\arabic*)]

		\item {\bf Configurable Graph Summarization Framework:} We introduce
			the problem of configurable graph summarization, called \gc,
			for summarizing large, undirected, unweighted graphs, and
			propose a solution framework called \cn,  based on the discovery
			of common neighbors.  The \cn framework offers three
			summarization variants, namely \iu, \inter and \uni, While \iu
			is a lossless summarization scheme, \inter and \uni are lossy
			summarization schemes that allow reconstruction with \emph{no
			false positive} edges and \emph{no false negative} edges, respectively.

		\item \textbf{Support for Multiple Graph Queries:} The \cn framework
			can handle multiple graph queries including neighborhood
			queries, shortest path queries and reachability queries. Without
			compromising on query accuracy, the \cn framework can answer
			queries on the summary graph without reconstructing the input
			graph. To this end, it uses the concept of local decompression,
			explained in Sec.~\ref{sec:query}.  
	
		\item \textbf{Bounded Neighborhood Loss:} The lossy variants of \cn,
			namely, \inter and \uni, offer parameterized control to bound
			the neighborhood loss. More specifically, we introduce a user
			specified parameter, the \emph{neighborhood loss tolerance
			threshold}, that quantifies the maximum loss allowed in the
			neighborhood of each node in the graph. This allows graph reconstruction
			and neighborhood query evaluation with either \emph{no loss} or
			with \emph{bounded loss guarantees}. This, in turn, enables
			retrieval of other graph queries such as shortest path and
			reachability queries with either no loss or with fairly
			high accuracy.
	
		\item {\bf Extensive Benchmarking:} Empirical evaluation on multiple
			synthetic and real-world graphs demonstrates that \cn offers superior
			summarization than the state-of-the-art methods. Moreover, it also showcases the ability of \cn to
			answer a wide variety of graph queries with high accuracy and
			efficiency.

\end{enumerate}

\subsection{Organization of the Article}

The rest of the article is organized as follows. Sec.~\ref{sec:related} discusses the related works. Sec.~\ref{sec:formulation} states the preliminary concepts and formally states the configurable graph summarization problem. The \cn framework is presented in Sec.~\ref{sec:framework} and Sec.~\ref{sec:greedycn}.  Sec.~\ref{sec:query} describes the query processing through the proposed framework. Sec.~\ref{sec:cgs_correctness} establishes the correctness of \cn.
In addition, the graph properties that are necessary for summarization through the \cn framework are stated in Sec.~\ref{sec:properties}.
The complexity of \cn is analyzed in Sec.~\ref{sec:complexity}.
The empirical findings are presented in Sec.~\ref{sec:results}. The conclusions and future works are discussed in Sec.~\ref{sec:conc}.

\section{Related Work}
\label{sec:related}

\begin{table*}[!thp]
\centering
	\resizebox{\textwidth}{!}{
  \begin{tabular}{ccccccc}
		\toprule
		\multirow{2}{*}{\textbf{Algorithm}} & \multirow{2}{*}{\textbf{Type of Graph}} & \textbf{Summarization} & \textbf{Bound on} & \textbf{Loss Type} & \textbf{Query Support} & \multirow{2}{*}{\textbf{Key Metrics}} \\ 
		& & \textbf{Type} & \textbf{Neighborhood Loss} & \textbf{(FP/FN)} & \textbf{(NQ/SPQ/RQ)} & \\
		\midrule
		\cn (this work) & Undirected, Unweighted & Both & Yes & Either FP or FN or none & NQ, SPQ, RQ & Compression ratio $cr = \frac{|G_s|}{|G|}$ \\
		\midrule	
		SWeG \cite{shin2019sweg} & Undirected, Unweighted & Both & Yes & Both & NQ & Size of output $= |P^*| + |C^+| + |C^-|$ \\
		APXMDL \cite{navlakha2008graph} & Undirected, Unweighted & Both & Yes & Not stated & None & Compression ratio $cr$ \\
		GraphZip \cite{rossi2018graphzip} & Undirected, Unweighted & Both & No & None & NQ, SPQ, RQ & Space savings $ = 1 - \frac{|G_s|}{|G|}$ \\
		GSQG \cite{riondato2017graph} & Undirected, Unweighted & Both & No & None & NQ & Compression ratio $cr$\\
		SLUGGER \cite{lee2022slugger} & Undirected, Unweighted & Lossless & Not applicable & Not applicable & None & Encoding cost $= |P^+| + |P^-| + |H|$ \\
		MoSSo \cite{ko2020incremental} & Undirected, Unweighted & Lossless & Not applicable & Not applicable & NQ & Size of output $= |P_t| + |C_t^+| + |C_t^-|$ \\
		SAGS \cite{khan2015set} & Undirected, Unweighted & Lossless & Not applicable & Not applicable & None & Compression ratio $cr$ \\
		compressR \cite{fan2012query} & Directed, Weighted & Lossless & Not applicable & Not applicable & RQ & Compression ratio $cr$ \\
		optGS \cite{lai2023optimized} & Undirected, Unweighted & Lossless & No & False Positive & None & Compactness = $\frac{|P|+|C^+|+|C^-|}{|E|}$ \\
		LM-GSUM \cite{berberidis2022summarizing} & Directed, Weighted & Lossless & Not applicable & Not applicable & None & Compression Ratio = $\frac{Bits\_Before - Bits\_After}{Bits\_Before}$ \\
		NQFCSN \cite{maserrat2010neighbor} &Directed, Unweighted & Lossless & Not applicable & Not applicable & NQ & bits per edge rate \\
		Traverse+K2T \cite{nejad2021graph} & Undirected, Unweighted & Lossy & No & Both & NQ, RQ & Compression ratio $cr$ \\
		StarZip \cite{li2019graph} & Undirected, Unweighted & Lossy & No & Not stated & SPQ &  $vcr = \frac{|V|}{|V_s|}$, $ecr = \frac{|E|}{|E_s|}$ \\
		GraSS \cite{lefevre2010grass} & Undirected, Unweighted & Lossy & No & Not stated & NQ & Reconstruction error \\
		LDME \cite{yong2021efficient} & Undirected, Unweighted & Lossy & No & Not stated & None & Compression $c = 1 - \frac{|P|+|C^+|+|C^-|}{|E|}$ \\
		SSumM \cite{lee2020ssumm} & Undirected, Weighted & Lossy & Yes & Not stated & None & Reconstruction error \\
		SSAG \cite{ali2024ssag} & Undirected, Unweighted & Lossy & No & Not stated & None & Reconstruction error \\
		GPQPS \cite{anagnostopoulos2024general} & Undirected, Weighted & Lossy & No & Not stated & SPQ & Relative error \\
		GRAPHGUESS \cite{ramezani2022graphguess} & Undirected, Unweighted & Not stated & Not applicable& Not applicable & SPQ & --- \\
		QMG \cite{nabti2017querying} & Undirected, Unweighted & Not stated & Not applicable & Not applicable & No specific query &  Compression Rate = $\frac{|E(C(G))|}{|E(G)|}$  \\
		\bottomrule
	\end{tabular}
	}
	\caption{Comparison of key related works in terms of their various properties (FP = False Positive, FN = False Negative, NQ = Neighborhood Query, SPQ =  Shortest Path Query, RQ = Reachability Query)}
	\label{tab:related}
\end{table*}

The related works can be broadly classified into two categories:
\begin{enumerate}
	\item \emph{Graph compression methods}: \cite{boldi2004webgraph,grabowski2010tight,liakos2014effect, besta2019slim, besta2018survey, li2019graph, maneth2015survey, nejad2021graph}
	\item \emph{Graph summarization methods}: \cite{beg2018scalable, li2019graph, kang2022personalized, ko2020incremental, navlakha2008graph, ali2024ssag, lai2023optimized, berberidis2022summarizing, maserrat2010neighbor}
\end{enumerate}

While
the former methods attempt to produce a compressed representation of the input graph by optimizing the minimum description length (MDL) using encoding techniques to compress frequently
occurring subgraphs in a graph, the latter produces a summary graph of the input graph while still retaining its key 
structural properties, either exactly or approximately. The former methods usually do not allow querying on the compressed graph without fully decompressing it \cite{boldi2004webgraph,grabowski2010tight,liakos2014effect}. On the other hand, the summary graphs produced by some of the latter methods may support queries \cite{riondato2017graph,shin2019sweg, anagnostopoulos2024general}. Given that graph compression methods are complementary to this work that falls in the category of graph summarization methods, the related work discussion primarily focuses on graph summarization.

\subsection{Comparison of Key Related Work}
\label{sec: comparison of related work}

Table~\ref{tab:related} presents a comparison of the key related works based on
	their configurable capabilities and query support. More specifically, the
	works have been categorized based on the type of input graph, type of
	summarization, i.e., lossy, lossless or both, whether the neighborhood loss
	is bounded, type of loss, i.e., false positive edges, false negative edges
	or both, support for queries on the summary graph such as neighborhood
	queries (NQ), reachability queries (RQ) and shortest path queries (SPQ),
	and the key metrics evaluated. In the table, $G = (V,E)$ refers to the
	input graph, $G_s = (V_s,E_s)$ refers to the summary graph, $P$ is the set
	of super edges of the summary graph $G_s$, $C^+$ is the set of edges to be
	inserted, $C^-$ is the set of edges to be deleted, and $H$ is the set of
	hierarchy edges between the supernodes. For further details on the key
	metrics stated in the table, please refer to the respective works.  The
	algorithms are ordered based on their similarity to \cn.

SWeG \cite{shin2019sweg} is a parallel graph summarization algorithm that is  designed for both shared-memory and MapReduce settings. SWeG repeats two steps: (1)~dividing the input graph into small subgraphs, and (2)~processing the subgraphs in parallel without having to load the entire graph in main memory. SWeG is closely related to \cn and it offers both lossy as well as lossless summarization variants, in addition to bounded loss of neighborhood and support for neighborhood queries. However, the reconstructed graph can contain both false positive as well as false negative edges. In particular, it does not allow the user to configure the type of loss in terms of its tolerance to false positive edges or false negative edges in the reconstructed graph and the neighborhood queries.

Navlakha et al. \cite{navlakha2008graph} considered the problem of computing the Minimum Description Length (MDL) representation of a graph. Although this method offers bounded neighborhood loss, it does not allow
any control over false positive and false negative edges in the
reconstructed graph. Moreover, it does not support query processing. 
GraphZip \cite{rossi2018graphzip} is based on decomposing a graph into a number of cliques of varying
sizes. It does not offer any bound on the neighborhood loss. Further, evaluating on just
two datasets, the authors claimed that their approach competes favorably with the
Layered Label Propagation (LLP) technique \cite{boldi2011layered}.
GSQG \cite{riondato2017graph} is based on partitioning the original set of
vertices into a small number of supernodes connected by super edges to form a
complete weighted graph. To quantify the dissimilarity between the original
graph and a summary, the authors optimized the reconstruction error and the
cut-norm error. They developed the first polynomial-time approximation
algorithms to compute the best possible summary of a certain size under both
measures. However, compared with \cn, this method neither offers direct control
on the neighborhood loss, nor it allows any choice of either false positive
or false negative edges in the reconstructed graph and neighborhood
queries. 

SLUGGER \cite{lee2022slugger} is a lossless hierarchical summarization technique that extends the work of \cite{navlakha2008graph}.  Although this method is effective in summarizing massive graphs, it \emph{does not} support query processing. 
Ko et al. \cite{ko2020incremental} presented MoSSo for lossless
graph summarization of fully dynamic graphs. With each edge insertion or
deletion, this method moves nodes among the supernodes in the summary graph.  
Experimental results (Sec.~\ref{sec:results}) show that the lossless scheme
of \cn, i.e., \iu offers considerably \emph{better summarization} than both, SLUGGER and  
MoSSo. 
SAGS \cite{khan2015set} is a set-based summarization approach that summarizes naturally occurring sets of similar nodes that are identified using locality sensitive hashing. This method, however, does not support queries.
CompressR \cite{fan2012query} is a lossless  compression method designed  for directed and weighted graphs. However, their approach is restricted to only two specific classes of queries, reachability and graph pattern. 

Traverse+K2T \cite{nejad2021graph} is based on the transitivity property often found in social networks and
web graphs.  However, it has neither any bound on the neighborhood loss nor any support for answering queries.
StarZip \cite{li2019graph} is a lossy compression scheme for streaming graphs. It repeatedly identifies dominating sets in a
graph, and represents each such dominating set by a supernode. However,
there is no control over false positive and false negative edges, or
any guarantee on the neighborhood loss. GraSS \cite{lefevre2010grass} is based on a random world model which does not offer any bound on the neighborhood loss. Further, this method seeks to minimize the reconstruction error rather than compression ratio.
LDME \cite{yong2021efficient} is a correction set based graph summarization algorithm that is based on weighted locality sensitive hashing. The algorithm allows trading compression for running time. However, there is no support to answer queries and there is no guarantee on the quality of the reconstructed graph.
SSumM \cite{lee2020ssumm} is a sparse summarization method based on the MDL principle. This method, however lacks support for queries and focuses on reducing the reconstruction error rather than the compression ratio.

The technique by Maserrat et al. \cite{maserrat2010neighbor} proposes an Eulerian data structure with multi-position linearizations to compress social networks while allowing sublinear-time in-neighbor and out-neighbor queries without decompression. Although it is interesting to see that the algorithm provides a high compression ratio along with supporting both neighbor query types and provides theoretical bounds with empirical validation, it addresses only lossless versions and focuses on neighbor queries alone. It may not generalize well to other query types, and involves complex preprocessing. On the other hand, our proposed approach \cn focuses on Neighborhood query, Reachability query and Shortest-path query. Nabti et al. \cite{nabti2017querying} uses modular decomposition to compress both query and data graphs, enabling subgraph isomorphism search directly on compressed graphs to reduce search space. It reduces storage and computation to some extent and is scalable to massive graphs, and avoids decompression. It is best suited with a restriction to labeled graphs with structural regularities only. Moreover, the compression quality affects query accuracy, and is not optimized for highly dynamic graphs.
The work by Sarwan Ali et al. SsAG, a scalable lossy summarization method for attributed graphs that merges nodes into supernodes while minimizing reconstruction error and maximizing attribute homogeneity, followed by sparsification to reduce storage \cite{ali2024ssag}. Though it handles large-scale graphs incorporating both topology and attributes, summaries may have slightly higher storage cost and minor quality degradation. Moreover, they consider only the lossy type with no support for query processing. optGS \cite{lai2023optimized} proposes an optimized correction set–based lossless graph summarization method that improves grouping, merging, and computation efficiency through computation-oriented supernode splitting. Although it is interesting to observe that optGS maintains exact reconstruction and scalable to large graphs, it is complex and may not be suitable for scenarios tolerating lossy compression and fails to address query support whereas \cn addresses various query processing mechanisms along with a bound on reconstruction loss. LMGSUM \cite{berberidis2022summarizing} introduces methods for summarizing labeled multi-graphs by grouping nodes and aggregating edges while preserving label and multiplicity information. The algorithm is good in supporting multiple edge labels and parallel edges but the summaries may still be large for highly heterogeneous graphs. Moreover, the complexity increases with the number of labels and may require domain-specific tuning for optimal grouping. The technique GPQPS proposed by Aris Anagnostopoulos et al. introduces a framework to process general-purpose queries directly on summary graphs by introducing a query translation mechanism that maps queries to the summarized structure without full decompression \cite{anagnostopoulos2024general}. In some terms, though, it reduces storage and computation, it handles only shortest-path queries for which the results are not encouraging and tackles only lossy version whereas \cn handles three different types of queries namely neighborhood query, shortest-path and reachability query and addresses lossless version as well. GRAPHGUESS is a runtime adaptive approximation model for graph processing that minimizes preprocessing, deactivates low-influence edges, and uses periodic super steps for adaptive correction. Although it has minimal accuracy loss, applicable to various algorithms and avoids heavy preprocessing, it is not clearly stated as to whether it is lossless/lossy and could handle only shortest-path queries  \cite{ramezani2022graphguess}. 

Koutra et al. \cite{koutra2014vog} constructed a vocabulary of subgraph types
that occur frequently in real graphs (e.g., stars, cliques, chains). It builds
a succinct description of a graph from a set of subgraphs using this vocabulary
and the MDL principle.  Utility driven graph
summarization \cite{kumar2018utility} aims to summarize graphs such that the
utility does not drop below a given threshold. Besta et al.
\cite{besta2019slim} proposed a programming framework for lossy graph
compression based on a novel abstraction of compression kernels.  Kang et al.
\cite{kang2022personalized} proposed personalized graph summarization that
focuses on connections closer to a given set of target nodes.  Berberidis et
al. \cite{berberidis2022summarizing} proposed a summarization framework that
can handle different graph characteristics like node labels, directed edges,
edge multiplicities, and self-loops. Liang et al. \cite{liang2020reachability}
proposed two summarization schemes to answer reachability queries on dynamic
graphs. Detailed surveys on various graph summarization and compression methods
are available in \cite{maneth2015survey,khan2017summarizing,besta2018survey,
awesome_graph_reduction}. 

\subsection{Difference with Existing Compression Algorithms based on Merging of Nodes}
\label{sec: difference with existing methods}

We next briefly review the merging approaches employed by the existing methods how they are different from the merging approached used in the \cn framework.

Navlakha et al. \cite{navlakha2008graph} used a \emph{greedy algorithm}
to compress a graph. The basic idea is that, any two nodes that share common
neighbors can give a cost reduction. The more the number of common neighbors,
the higher is the cost reduction. Suppose, $v$ is a supernode, and the cost
$c_v$ of supernode $v$ is the sum of the costs of all the super edges $(v, x)$ to
its neighbors $x \in N(v)$, i.e., $c_v = \sum_{x \in N(v)} c(v, x)$. Suppose,
two supernodes $u, v$ are merged into a new supernode $w$.  The authors proposed
a fractional cost reduction metric: $s(u,v) = \frac{c_u + c_v - c_w}{c_u +
c_v}$. The pair $(u, v)$ of supernodes for which $s(u,v)$ is maximized is chosen
for merging.

Kumar and Efstathopoulos \cite{kumar2018utility} tried to prioritize the
merging of two nodes by considering edge importance. Their goal is to pick an
edge $e = (u, v)$ with the lowest importance and merge the nodes $u$ and $v$ so
as to form a supernode $w$. At each step, they pick the node pair $(u, v)$ with
the lowest importance score $f(nodeIS[u])+ f(nodeIS[v])$. Here, $f()$ is a
square function, as it helps in further delaying the merging of important nodes
with relatively important nodes.

Lee et al. \cite{lee2022slugger} developed an algorithm called
Slugger, which greedily merges a root node pair among those sampled
within each candidate set, which is obtained in the previous step.
Simultaneously, Slugger updates $p$-edges and $n$-edges incident to the
merged nodes and/or their 1-level descendants by exploiting the hierarchy
between supernodes. Slugger accelerates this update through memoization.

Shin et al. \cite{shin2019sweg} developed an algorithm named SWeG, which
merges some supernodes within each group in a greedy manner. The merging step is
quite involved, and is described in detail in
the paper.

The merging strategy adopted in the \cn framework (detailed in Sec.~\ref{sec:framework}) is different from the above approaches in multiple ways. Firstly, it carefully discovers \emph{all} eligible pairs of nodes that offer positive compression gain. This compression gain is different for each of the three variants of the \cn framework (Sec.~\ref{sec:framework}). This compression gain not only factors the number of common neighbors but also the number of exclusive neighbors. Then among all the eligible pairs to be merged, it identifies the pair that offers the maximum compression gain. Then, it is checked whether the given pair is safe to be merged, i.e., whether the merging would violate the neighborhood loss constraint stated in Sec.~\ref{sec:formulation}. Once it is confirmed to be safe, the merge is executed. Subsequently, the merging is applied to the remaining pairs of nodes and super-nodes that are formed in the process.

\section{Configurable Graph Summarization, \gc}
\label{sec:formulation}

We consider an undirected, unweighted, simple graph $G=(V,E)$, where $V$ and
$E$ denote the vertex set and edge set, respectively.  Assuming the graph $G$ to
be stored using the adjacency list representation, the space required to
store $G$, denoted by $|G|$, is
\begin{align}
	\label{eq:sizeg}
	|G| = |V| + 2 \cdot |E|
\end{align}

\subsection{Graph Summarization}

Given an undirected, unweighted, simple graph $G=(V,E)$, the goal of graph summarization
is to produce a \emph{summary graph} $G_s=(V_s,E_s)$ (where $V_s$ and $E_s$
denote the vertex set and edge set of $G_s$ respectively) with two objectives:
(1)~$|G_s| \le |G|$, i.e., the summary graph requires lesser space than the
original graph, and (2)~it is possible to reconstruct the original graph $G$
from $G_s$ with no or limited loss of information. It is important to note that
$|G_s|$ must take into account the total space required to store $G_s$ along with the space overhead of additional data structures (if any) that are necessary to reconstruct $G$ from $G_s$. Thus, it may be the case that $|G_s| > |V_s|+2|E_s|$. This point will be elaborated later.

Some prior works \cite{li2019graph}, including those that offer summarization without graph reconstruction \cite{lee2020ssumm}, do not factor this point and compute compression ratios simply based on $|V_s|$ and/or $|E_s|$ (refer to the key metrics column in Table~\ref{tab:related}).

\subsubsection*{\bf Lossy and Lossless Summarization}
Let $\langle S,R \rangle$ denote a \emph{graph summarization scheme} where $S:G
\rightarrow G_s$ is a graph \emph{summarization algorithm} that summarizes an
undirected, unweighted graph $G=(V,E)$ to a \emph{summary graph}
$G_s=(V_s,E_s)$, and $R:G_s \rightarrow G_r$ is a \emph{reconstruction
algorithm} that takes the summary graph $G_s=(V_s,E_s)$ as input and returns
a \emph{reconstructed graph} $G_r=(V_r,E_r)$.
Here, reconstruction (also referred to as decompression) refers to the process of
reproducing the original graph $G$ from $G_s$. The overall goal is that $G_r$
should be a good approximation of $G$ with respect to certain graph properties.
If $R$ is able to retrieve the original graph $G$
from the summary graph $G_s$ exactly, i.e., $G_r=G$, then the summarization
scheme $\langle S, R \rangle$ is said to be \emph{lossless}. Otherwise, the
scheme is \emph{lossy}. In a lossy scheme, the set of edges $E_r \setminus E$ are referred to as \emph{false positive} edges, while $E \setminus E_r$ are referred to as \emph{false negative} edges.
Typically, the vertex set is reconstructed without loss, i.e., $V_r = V$.

\subsubsection*{\bf Summarization Objective}
Our graph summarization scheme seeks to produce a summary graph $G_s$
of a given graph $G$ that reduces the space overhead. In other words, the goal
is to minimize the \emph{compression ratio}, $cr=|G_s|/|G|$
\cite{nejad2021graph}. 

\subsection{Configuring the Summarization Variants}

Our proposed \gc framework allows three summarization variants, namely, \iu,
\inter, and \uni\footnote{As we will explain later, the variant-letters stand for \emph{exact}, \emph{intersection}, and \emph{union}, respectively.}, that are characterized by the presence or absence of false
positive and false negative edges in the edge set $E_r$ of the
reconstructed graph $G_r$.  If false positives need to be avoided, i.e., $E_r
\subseteq E$, then \inter can be used, while to have no false
negatives, i.e., $E_r \supseteq E$, then \uni can be used.  The variant \iu
avoids both false positives and false negatives, i.e., $E_r = E$ and, thus,
incurs no loss.
These variants are identified using a parameter $\ls$ that assumes three values, $E$, $I$ and $U$:
\begin{align} \label{eq:variants}
\fp = \begin{cases} 
	E \Rightarrow \iu: & E_r = E, \text{i.e., } \forall u \in G, N_{G_r}(u) = N_G(u) \\
	I \Rightarrow \inter: & E_r \subseteq E, \text{i.e., } \forall u \in G, N_{G_r}(u) \subseteq N_G(u) \\
	U \Rightarrow \uni: & E_r \supseteq E, \text{i.e., } \forall u \in G, N_{G_r}(u) \supseteq N_G(u) \\
\end{cases}
\end{align}
where $N_G(u)$ and $N_{G_r}(u)$ denote the 1-hop \emph{neighborhood} vertex set of
node $u$ in $G$ and $G_r$ respectively.

\subsection{Configuring the Neighborhood Loss}

The \emph{neighborhood loss} $\rl(u)$ of a node $u$ is the
loss of neighborhood information as a consequence of summarization and
reconstruction. It is measured as a fraction of the neighborhood set $N_G(u)$
that is lost after reconstruction \cite{navlakha2008graph}:
\begin{align}
	\small
	\rl(u) =  \frac{ |N_G(u) \setminus N_{G_r}(u)| + |N_{G_r}(u) \setminus N_{G}(u)| }{|N_G(u)|}
\end{align} 

The above equation is applicable for only those
summarization schemes that do not permit loss/gain of nodes, i.e., $V_r=V$. The
\cn framework satisfies this property and,
hence, uses this neighborhood loss model. For the different variants,
it follows that
\begin{align}
\label{eq:varrl}
\fp = \begin{cases} 
	E \Rightarrow & \rl(u) = 0 \\
	I \Rightarrow & \rl(u) = |N_G(u) \setminus N_{G_r}(u)| / |N_G(u)| \\
	U \Rightarrow & \rl(u) = |N_{G_r}(u) \setminus N_{G}(u)| / |N_G(u)| \\
\end{cases}
\end{align}

The neighborhood loss for \inter and \uni as defined in Eq.~\eqref{eq:varrl} is unbounded.
To bound it, we define a \emph{neighborhood loss tolerance threshold} parameter $\eps_u$ for each vertex $u$, where $0 \leq \eps_u \leq 1$.
The threshold parameter $\eps_u$ limits the maximum neighborhood loss:
\begin{align}
	\label{eq:loss-constraint}
	\forall u \in G, \ \rl(u) \le \eps_u
\end{align}

The \emph{neighborhood loss threshold set} parameter, $\nlt$, is defined as the set
of neighborhood loss tolerance threshold parameters for all vertices $u$ in $G$:
\begin{align}
	\label{eq:nlt}
	\nlt = \{ \eps_u| u \in G \}
\end{align}

The above definition allows the user to control the maximum neighborhood loss
for each node $u$ in an independent manner. 
This feature may be useful for applications that require nodes of different priorities to tolerate different levels of neighborhood loss.
Hence, this
flexibility to control the neighborhood loss makes this model powerful and
unique.
Setting $\eps_u$
to be the same value for all vertices $u$ sets
a uniform loss threshold.

\subsection{Query Processing}

When the graph size is large, retrieving many graph queries becomes a challenge \cite{nejad2021graph}. Hence, it is desirable to build a graph summarization scheme that can support graph queries on the summary graph without necessarily reconstructing the graph \cite{shin2019sweg}. The queries may be answered exactly or approximately, as preferred by the user. If the user is tolerant to loss in the query answers, it is desirable to answer the queries with bounded quality guarantees or with high accuracy.

\subsection{Problem Statement}

We next formally define the problem statement.
\begin{prob}[{\bf Configurable Graph Summarization, \gc}]
	\label{prob:gc}
	Given an undirected, unweighted, simple graph $G=(V,E)$, summarization
	variant parameter $\fp \in \{I,E,U\}$, and a neighborhood loss threshold
	$\nlt$, design a graph summarization scheme $\langle S, R \rangle$
	to return a summary graph $G_s$ that \emph{minimizes} the compression
	ratio $cr=|G_s|/|G|$ such that the reconstructed graph $G_r$ satisfies the
	neighborhood loss constraint, i.e., $\forall {u \in G}, \ \rl(u) \le
	\eps_u$, and $G_s$ supports multiple graph queries such as neighborhood queries, reachability queries and shortest path queries.
\end{prob}

\noindent
The \cn framework addresses the above \gc problem.
 
For evaluating any graph query on an undirected unweighted graph, it is necessary  to be able to answer the neighborhood queries. This is because one can always reconstruct (part or whole of) the input graph using neighborhood queries, and then retrieve the desired query. This is why \cn is designed to answer neighborhood queries with  either no loss or with bounded neighborhood loss, as desired by the user.  Consequently, this enables answering other graph queries such as reachability queries and shortest path queries exactly or with high
accuracy.

The \emph{neighborhood query} for a vertex $u$
in $G$ returns the 1-hop neighborhood $N_G(u)$, i.e., all the immediate
neighbors of $u$. 
 Given an ordered pair of nodes,  $(u,v)\in G$, the \emph{reachability query}, $R_G(u,v)$, determines the existence of a path from $u$ to $v$. 
We denote the set of nodes that are \emph{reachable} from $u$ via any path in
$G$ by $R_G(u)$. Thus, $R_G(u,v)=\true$ if and only if $v \in R_G(u)$. Given an ordered pair of nodes,  $(u,v)\in G$, the \emph{shortest path} query returns the shortest path from the source
node $u$ to the destination node $v$ along with its distance, denoted by $\dist_G(u,v)$, i.e., the number of edges on the shortest path from $u$ to $v$.  If $v \notin R_G(u)$, $\dist(u,v)=\infty$.  If
$G$ is undirected, these queries are symmetric.  In other words, $u \in
N_G(v)$ if and only if $v \in N_G(u)$,  $R_G(u,v)=R_G(v,u)$ and $\dist_G(u,v)
= \dist_G(v,u)$.

The query characteristics of the three summarization variants for these three
queries are summarized in Table~\ref{tab:query properties}.  These
characteristics are verified in Sec.~\ref{sec:query}.

\begin{table*}[t]
	\centering
	\resizebox{\textwidth}{!}
	{
  	\begin{tabular}{clcc}
		\toprule
		\bf \name Variant & \multicolumn{1}{c}{\bf Neighborhood query} & \multicolumn{1}{c}{\bf Reachability query} & \multicolumn{1}{c}{\bf Shortest Path query} \\
		\midrule
		\iu     & $N_{G_r}(u) = N_G(u)$    &  $R_{G_r}(u) = R_G(u)$    & $\dist_{G_r}(u,v) = \dist_G(u,v) $ \\
		\inter    & $N_{G_r}(u) \subseteq N_G(u)$, $(1-\eps_u)|N_G(u)| \le |N_{G_r}(u)| \le |N_G(u)|$     &   $R_{G_r}(u) \subseteq R_G(u)$ & $ \dist_{G_r}(u,v) \ge \dist_G(u,v)$ \\		
		\uni     & $N_{G_r}(u) \supseteq N_G(u)$, $|N_G(u)| \le |N_{G_r}(u)| \le (1+\eps_u)|N_G(u)|$    &  $R_{G_r}(u) = R_G(u)$ & $ \dist_{G_r}(u,v) \le \dist_G(u,v)$ \\
		\bottomrule  
	\end{tabular}
	}
	\caption{Query characteristics of the different summarization variants of \cn} \label{tab:query properties}
	\cuttabbelow
	\vspace*{1mm}
\end{table*}

\section {The \cn Framework}
\label{sec:framework}

\begin{figure}[t]
	\centering
	\includegraphics[width=0.80\columnwidth]{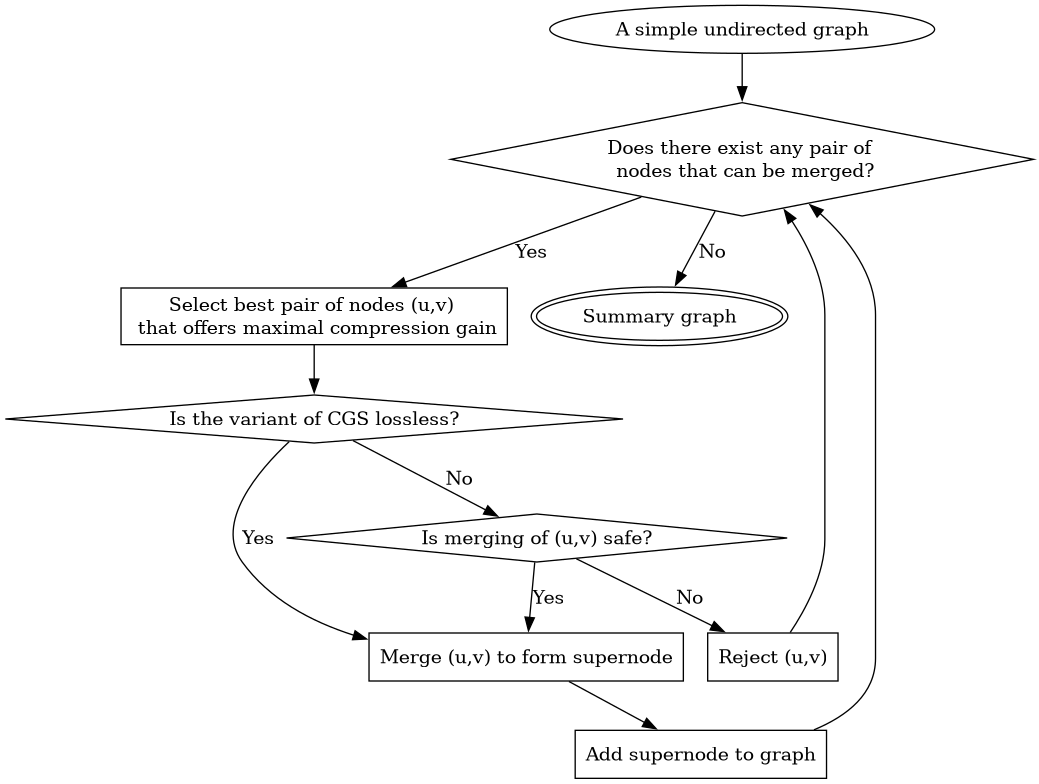}
	\Description[This figure demonstrates the overall algorithm of our proposed \cn framework.]{The algorithm proceeds in iterations and continues till there is a pair of nodes whose merge results in a compression and is safe according to the \cn variant.}
	\caption{Flowchart of the \cn framework. The algorithm proceeds in iterations and continues till there is a pair of nodes whose merge results in a compression and is safe according to the \cn variant.} \label{fig:flowchart}
\end{figure}

In this section, we describe in detail our \cn framework. We first
describe the merging procedure, and then the entire compression algorithm for
\cn, and decompression.

\subsection{Overview of \cn}

The \cn framework is based on \emph{merging} a pair of nodes that share a
common neighborhood.  Fig.~\ref{fig:flowchart} shows the high level
flowchart of this framework.  The algorithm runs in iterations. In every iteration, it finds all
pairs of nodes in $G$ that are eligible to be merged. If there exists such an
eligible pair, it determines the pair that offers the maximum compression
gain.  If the variant is lossy, i.e., for \inter and \uni, it checks whether merging of the given pair is
\emph{safe}, i.e., it would not violate the neighborhood loss constraint
stated in Eq.~\ref{eq:loss-constraint}.  
If the merge is found to be safe, the given
pair is merged; else, it is rejected.
For the lossless variant \iu, it is guaranteed to
be safe (we prove this later).
This process terminates when there is
no other eligible pair left that can be safely merged. The final
summary graph $G_s$ is then returned.

At any stage of the algorithm, consider the graph $G'_s=(V'_s,E'_s)$.
Initially, $G'_s = G$.  A pair of nodes $(u,v)$ (where $u, v \in G'_s$ such
that $v \in N^2_{G'_s}(u)$, i.e., $v$ is in the 2-hop neighborhood of $u$) is
merged if this leads to \emph{compression gain}, i.e., space savings.  Suppose
the graph resulting at the end of the above merge operation is referred to as
$G_s^* =(V_s^*,E_s^*)$.  The compression gain achieved by merging the pair
$(u,v)$ is $|G'_s|-|G_s^*|$. A merge operation is performed only if the
compression gain is positive, i.e., $|G_s^*| < |G'_s|$.

The following section describes the \emph{merging} procedure in detail.

\subsection{Merging Nodes with Common Neighborhoods}
\label{sec:merge}

Consider the graph $G'_s$ at any stage.
Before merging, $G_s^*$ is initialized to $G'_s$.  If the nodes $u$ and $v$
are merged, a \emph{supernode} $s = sn(u,v)$ is created and added to the vertex
set $V_s^*$ of $G_s^*$. The pair of nodes $\{u,v\}$ form the \emph{parents} of
the supernode $s$.  The supernode $s$ is referred to as the \emph{child} of
both $u$ and $v$.  Based on the variant of \cn, the \emph{neighborhood} of the
merged node $s$ is set according to the \emph{intersection} or the \emph{union} of
the neighborhoods of $u$ and $v$. The details are as follows.

First, consider the \inter variant. For each common neighbor $w$, i.e., when the
pair of edges, $(u,w), (v,w) \in E'_s$, only a single edge $(s,w)$ is
added to $E_s^*$. The nodes $u$ and $v$ are removed from $G_s^*$ along with all
their incident edges.  This, thus, corresponds to the \emph{intersection} of
the neighborhoods of $u$ and $v$.

Next, consider the \uni variant.  Consider merging of $u$ and $v$ to form a supernode
$s = sn(u,v)$.  If $(u,v) \in E'_s$, it is retained in $E_s^*$ along with its incident nodes.  For every other
edge of the form $(u,w) \in E'_s$ (or $(v,w) \in E'_s$), the edge is removed from
$E_s^*$, and an edge $(s,w)$ is added to $E_s^*$.  If $(u,v) \notin E'_s$, then the nodes $u$ and $v$ are removed from $G_s^*$.
This operation
corresponds to the \emph{union} of the neighborhoods of $u$ and $v$.

Finally, consider the \iu variant. For each pair of edges, $(u,w), (v,w) \in E'_s$, a single edge $(s,w)$ is added to $E_s^*$ and the
former pair of edges is removed from $E_s^*$. However,
edges corresponding to the difference of the neighborhood sets are not
removed but retained. Formally, for each edge $(u,w) \in E'_s$, such
that $(v,w) \notin E'_s$, the edge $(u,w)$ is still retained in $E_s^*$ along with
its incident nodes.  Similarly, every edge $(v,w) \in E'_s$, where $(u,w) \notin
E'_s$, is retained in $E_s^*$ along with its incident nodes.  

Notably, the above merge operation is \emph{not} limited to only a pair of
nodes, but can be performed between any node and a supernode, or between any
pair of supernodes as well.

We assume that the graphs are stored using adjacency lists. Thus,
$|G|=|V|+2\cdot|E|$. If there are $k$ merge operations, each producing a supernode
having exactly two parents, the space overhead of $G_s$ is
$|G_s|=|V_s|+2\cdot|E_s|+2k$. Thus, the compression gain
achieved by merging the pair $u,v\in G'_s$ is
\begin{align}
	\label{eq:cg}
	cg(u,v )=(|V'_s|-|V_s^*|)+2(|E'_s|-|E_s^*|)-2 
\end{align}

\begin{table*}[t]
		\centering
	\begin{tabular}{c|M{\examplewidth}|M{\examplewidth}M{\examplewidth}|M{\examplewidth}}
		\toprule
		\bf \fp & \bf Input graph $G$ & \multicolumn{2}{c|}{\bf Subsequent merge operations in $G_s$} & \bf Reconstructed graph $G_r$ \\
		\midrule
		\multirow{2}{*}{$I$} & \includegraphics[width=\figwidth]{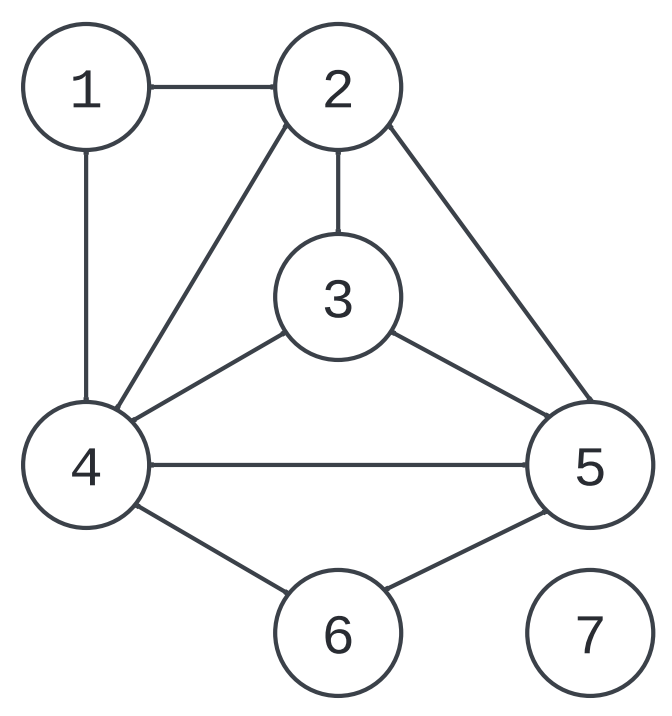} &\includegraphics[width=\figwidth]{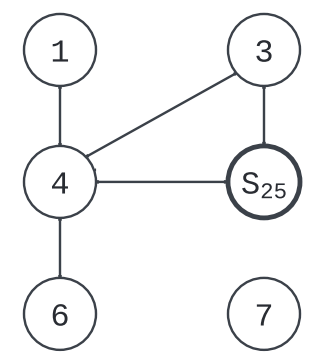} & \includegraphics[width=\figwidth]{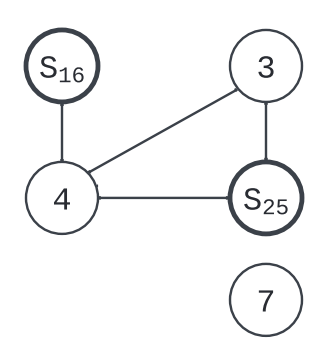} &   \includegraphics[width=\figwidth]{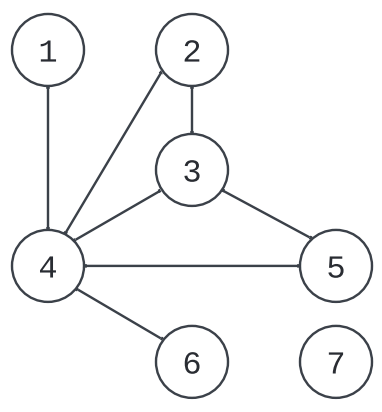}\\
		& (a) $|G| = 27$ & (b) $|G_s| = 18$ & (c) $|G_s| = 17$ & (d) $|G_r| = 21$\\
		\midrule
		\multirow{2}{*}{$U$} & \includegraphics[width=\figwidth]{input_graph} &\includegraphics[width=\figwidth]{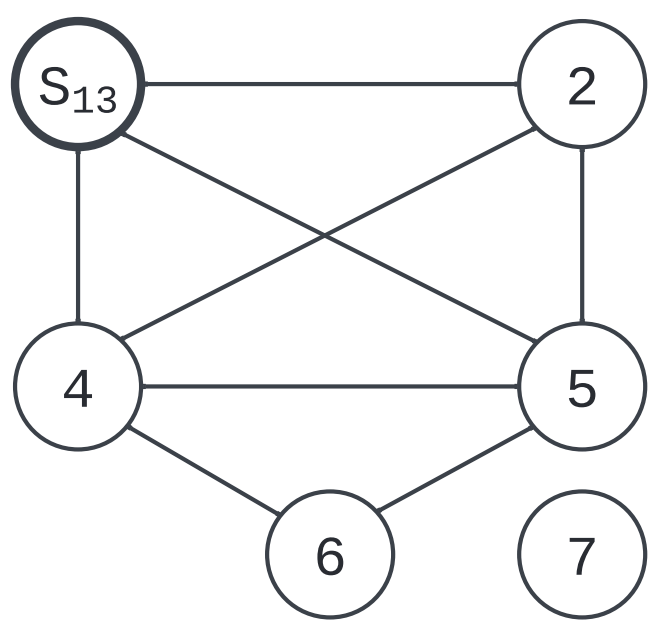} & \includegraphics[width=\figwidth]{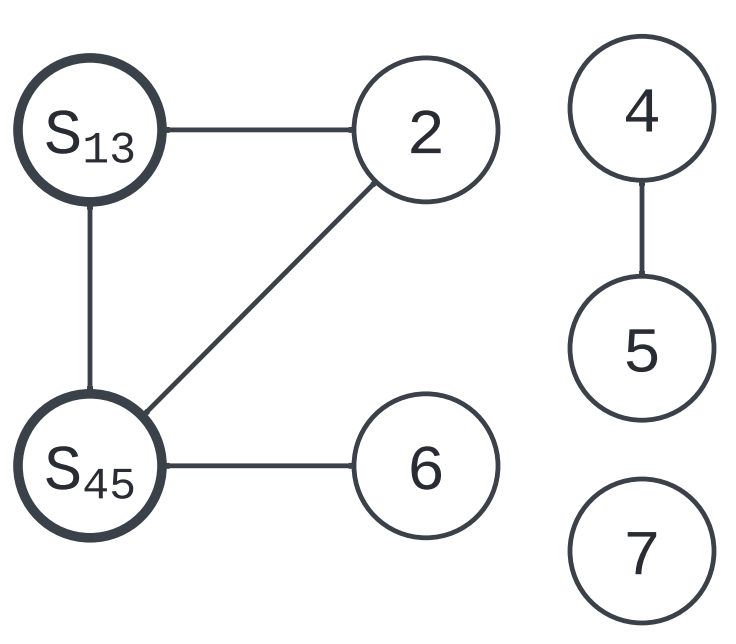}  & \includegraphics[width=\figwidth]{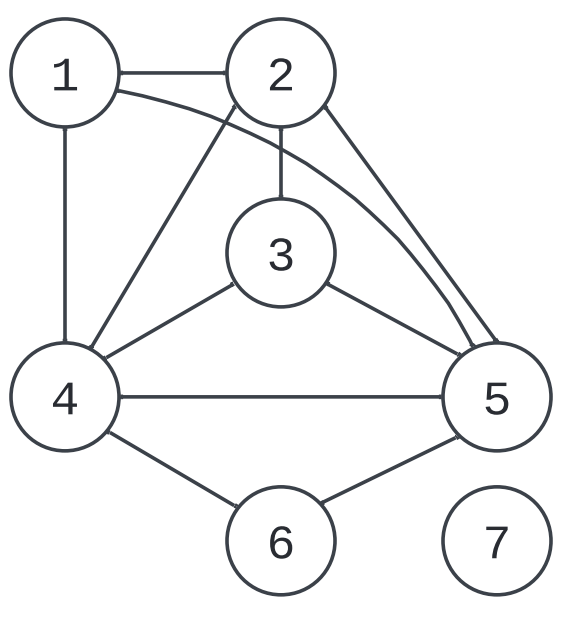}  \\
		& (e) $|G| = 27$ & (f) $|G_s| = 24$ & (g) $|G_s| = 21$ & (h) $|G_r| = 29$ \\
		\midrule
		\multirow{2}{*}{$E$} & \includegraphics[width=\figwidth]{input_graph} & \includegraphics[width=\figwidth]{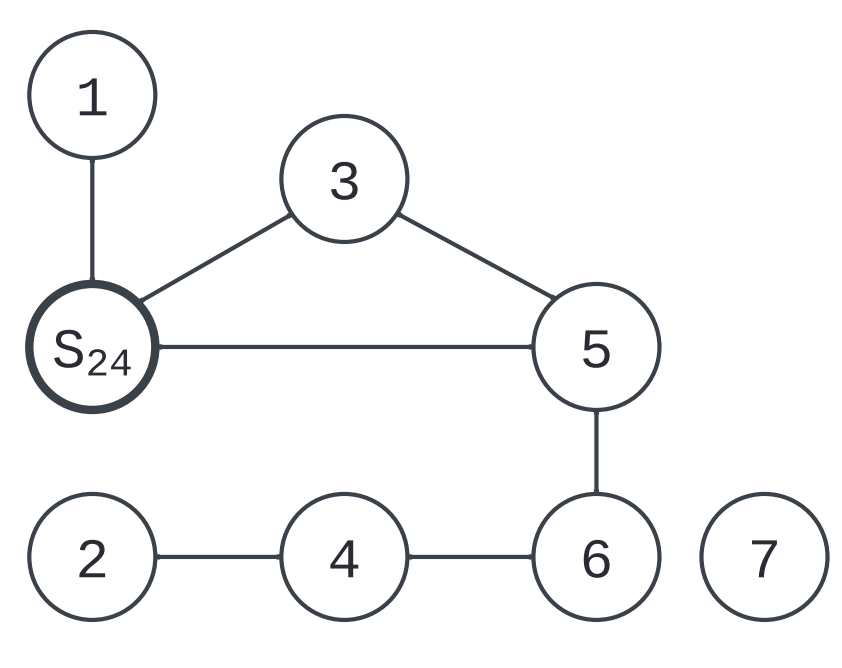} &    & \includegraphics[width=\figwidth]{input_graph}\\
		& (i) $|G| = 27$ & (j) $|G_s| = 24$ & & (k) $|G_r| = 27$ \\
		\bottomrule
	\end{tabular}
	\caption{Illustration of how different variants of \cn work for the
	same input graph}
	\label{tab:example}
\end{table*}

\begin{table}[t]
	\centering
	\begin{tabular}{ c c cc cc}
		\toprule
		\bf Compression &  & \multicolumn{2}{c}{\bf Step 1} & \multicolumn{2}{c}{\bf Step 2}  \\
		\cline{3-6}
		\bf Variant & & $(u,v)$ & $cg(u,v)$ & $(u,v)$ & $cg(u,v)$  \\
		\midrule 
		\inter ($\eps = 1/2$) & & $(2,5)$ & $9$ & $(1,6)$ & $1$ \\
		\uni ($\eps = 3/5$)   & & $(1,3)$ & $3$ & $(4,5)$ & $3$\\
		\iu                   & & $(2,4)$ & $3$ & - & - \\
		\bottomrule
	\end{tabular}
	\caption{Working of \cn algorithm for the examples in Table~\ref{tab:example}, where $cg(u,v)$ is the compression gain achieved by merging $(u,v)$.}
	\label{tab:working}
\end{table}

\subsection{Example}

We illustrate the merge operation for different variants of \cn with
an example shown in Table~\ref{tab:example}. The size of the graph is shown
below each figure.

First, consider \inter. Consider the merging of the node pair $(2,5)$ in the
input graph $G$ shown in Fig.~(a). The resulting graph is shown in Fig.~(b).
The pair $(2,5)$ are merged to form the supernode $s=S_{2,5}$. The common
neighbors (i.e., \emph{intersection}) $\{3, 4\}$ form the neighborhood of $s$.
The parents of $s$, i.e., $\{2,5\}$ are removed along with their incident
edges.  Consequently, edges $(2,1)$ and $(5,6)$ get deleted.  Fig.~(c) shows
the merging of the node pair $(1,6)$.  Note that formation of the supernode
$S_{1,6}$ does not lose any edge.

Next, consider \uni. Consider the merging of the node pair
$(1,3)$ in the input graph $G$ shown in Fig.~(e). The resulting graph is shown
in Fig.~(f). The pair $(1,3)$ is merged to form the supernode $s=S_{1,3}$. The
\emph{union} of neighbors of $1$ and $3$ is $\{2,4,5\}$, which form the
neighborhood of $s$. Similar to \inter, the parents of $s$, i.e., $\{1,3\}$
are removed along with their incident edges.  Fig.~(g) shows the merging step
for the node pair $(4,5)$.  

Finally, consider \iu. Consider the merging of the node pair $(2,4)$ in the
input graph $G$ shown in Fig.~(i). The resulting graph is shown in Fig.~(j).
The pair $(2,4)$ are merged to form the supernode $s=S_{2,4}$. The common
neighborhood of $2$ and $4$ is $\{1,3,5\}$.  Once the common neighbors of $2$
and $4$ are absorbed by $s$, the neighborhood sets of $2$ and $4$ reduce to
$\{4\}$ and $\{2,6\}$ respectively. Unlike $\fp=I$, these two nodes are
retained with their updated neighborhood sets.

The compression gain achieved for each merge operation in the example in
Table~\ref{tab:example} is shown in Table~\ref{tab:working}.  The key
difference between the three summarization variants lies in the way the nodes are
merged which, in turn, influences the compression ratio and the neighborhood
loss.

\section{The \cn Algorithm}
\label{sec:greedycn}

This section presents the proposed graph summarization scheme, \cn, that is based
on the idea of merging nodes with common neighborhoods (described in Sec.~\ref{sec:merge}).  The \cn
algorithm (pseudo code shown in Algo.~\ref{alg:greedycn}) is parameterized with the summarization variant parameter \fp.  It acts
as \inter or \iu or \uni, corresponding to $\ls = I$ or $\ls = E$ or $\ls =
U$ respectively.
Given a graph $G$ and a neighborhood loss tolerance threshold parameter set $\nlt$, it returns a summary graph $G_s$.

The \greedycn procedure (Algo.~\ref{alg:greedycn}) is based on the greedy paradigm. It iteratively merges a pair of nodes that offer the maximum compression gain. 
 First, it initializes $G_s$ to $G$.  Next, through the \buildheap
procedure, it discovers all pairs of nodes $(u,v)$ (where $v \in N^2_G(u)$) that
can be merged potentially (line 6). The details of this procedure are described next.

\begin{algorithm}[tbh]
	\caption{\greedycn \label{alg:greedycn}}
	\KwInput{An undirected, unweighted graph $G=(V,E)$, the summarization variant parameter \fp and a set of neighborhood loss tolerance thresholds \nlt.}
	\KwOutput{A summary graph $G^*_s$ that  is an approximate solution to the \gc problem for the given graph $G$.}	
	\SetKwFunction{Fgreedy}{\greedycn}	
	\Fn{\Fgreedy{$G, \fp,\nlt$}}{        
	$G_s \gets G$\;
	\ForEach{ $u \in G_s$}{
		$P_{G_s}(u) \gets \emptyset$\; 
		$\deg_{G_r}(u) \gets |N_G(u)|$\;
		}
		$H \gets$ \buildheap$(G,\fp)$\; 
		\While{ $H$ is non-empty}{ \label{line:begin-while}	
		$h \gets$ \exmax$(H)$\; 
		\If {$key(h) \le 0$} 
		{
			\KwRet  $G_s$  	\label{lin:algo terminates}\;
		}
		$(u^*,v^*) \gets value(h)$\;
		\If
		{ 
			($\fp=E$) $\vee$ (IsSafeMerge$(G_s,\fp, u^*,v^*,\nlt)$ is \true)
		} 
		{ 
			\merge$(G_s,\fp,H,u^*,v^*)$ \;
		}
		\label{line:end-while}
		} 
		\KwRet  $G_s$\; 
		}
\end{algorithm}

\subsection{The \buildheap Procedure}
\label{sec:buildheap}

The \buildheap procedure (Algo.~\ref{alg:buildheap}) discovers the pairs of nodes
that are eligible to be merged, i.e., all pairs of nodes $(u,v)$ where $v$ is a 2-hop neighbor of $u$, such that their compression gain $cg(u,v)>0$. To this end, the \cn algorithm, computes the the size of the intersection as $\ic(u,v)=|N_G(u) \cap N_G(v)|$ through the \nic procedure (Algo.~\ref{alg:nic}). Then, it computes the compression gain using the \computecg procedure (Algo.~\ref{alg:computecg}). Based on the value of the compression variant $\fp \in \{I,E,U\}$, this procedure returns the appropriate compression gain (as detailed in Sec.~\ref{sec:merge}). It computes the size of the union $\uc(u,v)=|N_G(u \cup N_G(v))|$ using the degrees of the nodes $u$ and $v$ and $\ic(u,v)$, as stated in line~3. Once the compression gain is computed, and it is found to be positive, it is added to a binary max-heap where the compression gain of the pair acts as its key.  The reason we choose binary max-heap is because the algorithm needs to iteratively update the compression gain of one or more pairs. and retrieve the pair that offers the maximum compression gain. Since each of these operations can be performed in logarithmic time in a binary max-heap, it is the best data structure for this purpose.

\begin{algorithm}[tbh]
	\caption{\buildheap \label{alg:buildheap}}
	\SetKwFunction{Fbuildheap}{\buildheap}
	\Fn{\Fbuildheap{$G,\fp$}}{
		Create an empty binary max-heap $H$.\;
		\ForEach{  $u \in G$}{
			\nic(G,u)\;
			\ForEach{  $v \in N_G(u)$}{
				\ForEach{  $w \in N_G(v)$}{
					\If{$u<w$}{
						$cg \gets$ \computecg(G,\fp,u,w) \;
						\If{$cg>0$}{
							\insertkey($H,u,w,cg$)\tcp*{inserts a node into the heap $H$ with $key=cg$ and $value=(u,w)$}\;
						}
					}
				}
			}
		}
		\KwRet $H$\;
	}
\end{algorithm}

\begin{algorithm}[tbh]
	\caption{\nic \label{alg:nic}}
	\SetKwFunction{Fnic}{\nic}
	\Fn{\Fnic{$G,u$}}{
		\ForEach{  $v \in N_G(u)$}{
			\ForEach{  $w \in N_G(v)$}{
				\If{$u<w$}{
					$\ic(u,w) \gets 0$\;
				}
			}
		}
		\ForEach{  $v \in N_G(u)$}{
			\ForEach{  $w \in N_G(v)$}{
				\If{$u<w$}{
					$\ic(u,w) \gets \ic(u,w) + 1$\;
				}
			}
		}
	}
\end{algorithm}

\subsection{The \computecg Procedure}

\begin{algorithm}[tbh]
	\caption{\computecg \label{alg:computecg}}
	\SetKwFunction{Fcomputecg}{\computecg}
	\Fn{\Fcomputecg{$G_s, \fp,u,v$}}{
		\If{$\fp=I$}{
			$\uc(u,v) \gets \deg_{G_s}(u)+\deg_{G_s}(v)-\ic(u,v)$\;
			$cg \gets 2\uc(u,v)-3$\;
			\If{ $v\in  N_{G_s}(u)$}{
				$cg \gets cg -2$\;
			}	
		} 
		\ElseIf{$\fp=E \vee \fp=U$}{
			$cg \gets 2\ic(u,v)- 3$
		}
		\ForEach{$w \in \{u,v\}$}{	
			\If{$|P_{G_s}(w)|=0$}{
				\If{$\fp=I$}{
					$cg \gets cg+1$
				} 
				\ElseIf{$\fp=U \wedge v\notin N_{G_s}(u)$}{	
					$cg \gets cg+1$
				}	
				\ElseIf{$\fp=E$}{
					$w' \gets \{u,v\}-\{w\}$\;
					$W \gets N_{G_s}(w)-N_{G_s}(w')$\;
					\If{$|W|=0$}{
						$cg \gets cg+1$
					}	
				}
			}
		}
		\KwRet $cg$ 
	}
\end{algorithm}

The \greedycn procedure maintains those pairs $(u,v)$ that offer
positive compression gain, i.e., $cg(u,v)>0$ using the binary max-heap $H$ that is generated by the \buildheap procedure, as discussed above. Then, the algorithm proceeds in
iterations. In each iteration, it selects a pair of nodes $(u^*,v^*)$ that
offers the \emph{maximal} compression gain (line 8).

However, if $\fp=I$ or $\fp=U$, it is not necessary that a merge operation is
\emph{safe}, i.e., it does not violate the neighborhood loss constraint
(Eq.~\eqref{eq:loss-constraint}). Hence, the algorithm runs the procedure \issafemerge (line 12) to check
that. If the merge is found to be safe, the algorithm executes the merge operation through the \merge procedure (Algo.~\ref{alg:merge}). The algorithm terminates when there are no more pairs 
that offer positive compression gain.

\subsection{The \issafemerge Procedure}
\label{sec:issafemerge}

The \issafemerge procedure (Algo.~\ref{alg:issafemerge}) validates whether
merging a given pair $(u,v)$ is safe, i.e., it would not violate the
neighborhood loss constraint (Eq.~\eqref{eq:loss-constraint}). To realize this, this procedure  correctly computes the degree $\deg_{G_r}(u)$ of each node $u$ in the
reconstructed graph $G_r$, without performing the actual reconstruction.  Using
this, checking the neighborhood loss constraint simplifies to checking whether
$|\deg_{G_r}(u) - \deg_G(u)|/\deg_G(u) \le \eps_u$ for all $u \in G$. This
check is performed based on the notion of \emph{ancestors} defined as follows.

\begin{algorithm}[phtb]
	\caption{\issafemerge \label{alg:issafemerge}}
	\SetKwFunction{Fsafe}{\issafemerge}
	\Fn{\Fsafe{$G_s,\fp,u^*,v^*, \nlt$}}{ 
		\If{$\fp=E$}{
			\KwRet \true\;
		}
		$U^* \gets \snmap(G_s,u^*)$\;
		$V^* \gets \snmap(G_s,v^*)$ \;
		\If{$\fp=I$}{
			$U' \gets \{u' \in \snmap(G_s,u'')| u'' \in (N_{G_s}(u^*)-N_{G_s}(v^*))\}$\;
			$V' \gets \{v' \in \snmap(G_s,v'')| v'' \in (N_{G_s}(v^*)-N_{G_s}(u^*))\}$\;
			\ForEach{$u \in U^*$}
			{
				\If {$\frac{\deg_{G_r}(u)-|U'|}{\deg_G(u)} < 1-\eps_u$}
				{
					\KwRet \false
				}
			}
			\ForEach{$v \in V^*$}
			{
				\If {$\frac{\deg_{G_r}(v)-|V'|}{\deg_G(v)} < 1-\eps_v$}
				{
					\KwRet \false
				}
			}
			\ForEach{$u \in U'-V^*$}
			{
				\If {$\frac{\deg_{G_r}(u)-|U^*|}{\deg_G(u)} < 1-\eps_{u}$}
				{
					\KwRet \false
				}
			}
			\ForEach{$v \in V'-U^*$}
			{
				\If {$\frac{\deg_{G_r}(v)-|V^*|}{\deg_G(v)} < 1-\eps_{v}$}
				{
					\KwRet \false
				}
			}
			\KwRet \true \;
		}
		\ElseIf{$\fp=U$}{
			$U' \gets \{u' \in \snmap(G_s,u'')| u'' \in (N_{G_s}(u^*)-\{v^*\}-N_{G_s}(v^*))\}$\;
			$V' \gets \{v' \in \snmap(G_s,v'')| v'' \in (N_{G_s}(v^*)-\{u^*\}-N_{G_s}(u^*))\}$\;
			\ForEach{$u \in U^*$}
			{
				\If {$\frac{\deg_{G_r}(u)+|V'|}{\deg_{G}(u)} > 1+\eps_u$}
				{
					\KwRet \false
				}
			}
			\ForEach{$v \in V^*$}
			{
				\If {$\frac{\deg_{G_r}(v)+|U'|}{\deg_{G}(v)} > 1+\eps_v$}
				{
					\KwRet \false
				}
			}
			\ForEach{$u \in U'$}
			{
				\If {$\frac{\deg_{G_r}(u)+|V^*|}{\deg_{G}(u)} > 1+\eps_{u}$}
				{
					\KwRet \false
				}
			}
			\ForEach{$v \in V'$}
			{
				\If {$\frac{\deg_{G_r}(v)+|U^*|}{\deg_{G}(v)} > 1+\eps_{v}$}
				{
					\KwRet \false
				}
			}
			\KwRet \true \;
		}
	}  
\end{algorithm}

\begin{algorithm}[tbh]
	\caption{\snmap \label{alg:snmap}}
	\SetKwFunction{Fsnmap}{\snmap}
	\Fn{\Fsnmap{$G_s,s$}}{ 
		$Nodes \gets \emptyset$\;
		Let $Stack$ be an empty stack.\;
		$Stack.push(s)$\;
		\While{$Stack$ is non-empty}{
			$u=Stack.pop()$ \;
			\If{$u \notin G_s \vee P_{G_s}(u)=\emptyset$}{
				\tcp*{$u$ is a simple node.}
				$Nodes \gets Nodes \cup \{u\}$
			}
			\Else{
				\tcp*{$u$ is a super-node.}
				\ForEach{$v \in P_{G_s}(u)$}{
					$Stack.push(v)$
				}
			}
		}
		\KwRet $Nodes$\;
	}
\end{algorithm}

\begin{algorithm}[tbh]
	\caption{\merge \label{alg:merge}}
	\SetKwFunction{Fmerge}{\merge}
	\Fn{\Fmerge{$G_s,\fp,H,u^*,v^*$}}{
		Add supernode $s=sn(u^*,v^*)$ to $V_s$\;
		$P_{G_s}(s) \gets \{u^*,v^*\}$ \;
		\If {$\fp=I \vee \fp=E$}{
			$N_{G_s}(s) \gets N_{G_s}(u^*) \cap N_{G_s}(v^*)$\;	
		}
		\ElseIf {$\fp=U$}{
			$N_{G_s}(s) \gets N_{G_s}(u^*) \cup N_{G_s}(v^*)-\{u^*,v^*\}$\;	
		}	
		\If{$\fp=I \vee \fp=U$}{
			\updatedegree($G_s,\fp,u^*,v^*$)\;	
		}
		\updateheap($G_s,\fp, H,s, u^*,v^*$)\;
		\ForEach{$w \in N_{G_s}(s)$} {
			$N_{G_s}(w) \gets N_{G_s}(w) -\{u^*,v^*\} \cup \{s\}$
		}
		\If {$\fp = I$}{
			\ForEach{$w \in N_{G_s}(u^*)-N_{G_s}(s)$} {
				$N_{G_s}(w) \gets N_{G_s}(w) -\{u^*\}$\;
			}
			\ForEach{$w \in N_{G_s}(v^*)-N_{G_s}(s)$} {
				$N_{G_s}(w) \gets N_{G_s}(w) -\{v^*\}$\;
			}
			$N_{G_s}(u^*) \gets \emptyset$, 
			$N_{G_s}(v^*) \gets \emptyset$\;
		}
		\ElseIf{$\fp=E \vee \fp=U$}{
			$N_{G_s}(u^*) \gets N_{G_s}(u^*) - N_{G_s}(s)$\;
			$N_{G_s}(v^*) \gets N_{G_s}(v^*) - N_{G_s}(s)$\;
		}
		\ForEach{ node $w \in \{u^*,v^*\}$}{
			\If{$|N_{G_s}(w)|=0 \wedge |P_{G_s}(w)|=0$}{
				Remove node $w$ from $G_s$
			}	
		}		
	}	
\end{algorithm}

A node $u_1 \in G$ in the original graph is said to be an \emph{ancestor} of a
node $u_l \in G_s$ in the summary graph, if there exist a sequence of nodes
$u_1, \dots, u_l$ in the interim summary graphs such that for each
$i=1,\dots,l-1$, $u_i$ is a \emph{parent} of $u_{i+1}$. As a special case, if
$l=1$, i.e., a node $u_1 \in G$ has no child, then $u_1$ is its own ancestor as
well. For any node $u^* \in G_s$, let $\snmap_{G_s}(u^*)$ denote the set of
ancestors of $u^*$. For computing the set of ancestors, the algorithm employs the \snmap procedure (Algo.~\ref{alg:snmap}). To realize these ancestor sets, the algorithm maintains the parents of each super-node $sn$ as they get produced after each merge operation, denoted by $P_{G_s}(sn)$. If a node $u \in G_s$ did not undergo any merge operation, its parent set is deemed to be empty.

Subsequently, the procedure evaluates all the cases applicable to each of the lossy variants $\fp \in \{I,U\}$ where the degree of a node in the reconstructed graph can potentially get altered due to a subsequent merge operation. For this purpose, it computes the potential neighborhood loss for each node in the graph $G$ and correctly computes the altered degree of each node in the potential reconstructed graph.If none of the nodes violate the neighborhood loss constraint, the procedure returns true. Otherwise, it returns false.

\subsection{The \snmap Procedure}

The \snmap procedure
(Algo.~\ref{alg:snmap}) computes the ancestors of a node $u \in G_s$, where the ancestors are defined in Sec.~\ref{sec:issafemerge}. The procedure successively pushes the parents of the super-node $u$ into a stack. It iteratively pops out a node, and pushes its parents, in turn. Finally, when the stack becomes empty, it returns the ancestor set.  

\subsection{The \merge Procedure}

The \merge procedure (Algo.~\ref{alg:merge}) performs the merge operation on a pair of nodes $(u,v)$, after it is confirmed that the given pair is safe to be merged. Given the choice of the \cn variant, i.e., $\fp \in \{I,E,U\}$, this merge operation is carried out based on the details stated in Sec.~\ref{sec:merge}.
Once the nodes $(u,v)$ are merged to a super-node $sn(u,v)$, the parent set of $sn(u,v)$is set to contain the nodes $\{u,v\}$. Then,  the neighborhoods of $u$ and $v$ are updated, along with that of the super-node $sn(u,v)$. 

\begin{algorithm}[tbh]
	\caption{\updatedegree \label{alg:updatedegree}}
	\SetKwFunction{Fupdegree}{\updatedegree}
	\Fn{\Fupdegree{$G_s, \fp, u^*,v^*$}}{ 
		$U^* \gets \snmap(G_s,u^*)$\;
		$V^* \gets \snmap(G_s,v^*)$ \;
		\If{$\fp=I$}{
			$U' \gets \{u' \in \snmap(G_s,u'')| u'' \in (N_{G_s}(u^*)-N_{G_s}(v^*))\}$\;
			$V' \gets \{v' \in \snmap(G_s,v'')| v'' \in (N_{G_s}(v^*)-N_{G_s}(u^*))\}$\;
			\ForEach{$u \in U^*$}
			{
				$\deg_{G_r}(u) \gets \deg_{G_r}(u)-|U'| $
			}
			\ForEach{$v \in V^*$}
			{
				$\deg_{G_r}(v) \gets \deg_{G_r}(v)-|V'| $
			}
			\ForEach{$u \in U'-V^*$}
			{
				$\deg_{G_r}(u) \gets \deg_{G_r}(u)-|U^*|$
			}
			\ForEach{$v \in V'-U^*$}
			{
				$\deg_{G_r}(v) \gets \deg_{G_r}(v)-|V^*| $
			}
		}
		\ElseIf{$\fp=U$}{
			$U' \gets \{u' \in \snmap(G_s,u'')| u'' \in (N_{G_s}(u^*)-\{v^*\}-N_{G_s}(v^*))\}$\;
			$V' \gets \{v' \in \snmap(G_s,v'')| v'' \in (N_{G_s}(v^*)-\{u^*\}-N_{G_s}(u^*))\}$\;
			\ForEach{$u \in U^*$}
			{
				$\deg_{G_r}(u) \gets \deg_{G_r}(u)+|V'|$
			}
			\ForEach{$v \in V^*$}
			{
				$\deg_{G_r}(v) \gets \deg_{G_r}(v) + |U'|$
			}
			\ForEach{$u \in U'$}
			{
				$\deg_{G_r}(u) \gets \deg_{G_r}(u)+|V^*|$
			}
			\ForEach{$v \in V'$}
			{
				$\deg_{G_r}(v) \gets \deg_{G_r}(v)+|U^*|$
			}
		}
	}  
\end{algorithm}

\subsection{The \updatedegree Procedure}

Once a given pair
is merged, the \updatedegree procedure (Algo.~\ref{alg:updatedegree}) updates
the degree of each node $u\in G_r$, given by $\deg_{G_r}(u)$.As in the case of \issafemerge procedure, this procedure evaluates the neighborhood loss for each node in the reconstructed graph, and updates its degree accordingly.

\subsection{The \updateheap Procedure}

The \updateheap
procedure (Algo.~\ref{alg:updateheap}) is responsible for updating the heap $H$
after a given pair of nodes is merged. As a consequence of a merge operation of a pair of nodes $(u,v)$, the compression gain of other pairs may get affected. To this end, this procedure not only identifies the relevant pairs whose compression gain can potentially be affected due to the previous merge operation, but it also computes the change in the  compression gain for all such pairs. It evaluates all possible scenarios for each of the three variants of the \cn framework. For efficiency, the above updates are
performed in a manner that avoids accessing nodes whose neighborhoods do not
change. The details are stated in Algo.~\ref{alg:updateheap}.

\begin{algorithm}[tbh]
	\caption{\updateheap \label{alg:updateheap}}
	\SetKwFunction{Fupheap}{\updateheap}
	\Fn{\Fupheap{$G_s,\fp,H,s,u^*,v^*$}}{ 
		$\nic(G_s,s)$\;
		\ForEach{ pair $(s,u)$ such that $u \in N^2_{G_s}(s)$}{
			$cg \gets $ \computecg($G_s,\fp,s,u$)\;
			\If{$cg>0$}{
				\insertkey($H,s,u,cg$)\;
			}
		}
		\If{$\fp=I \vee \fp=U$}{
			\ForEach{ pair $(u,v^*) \in H$}{
				\deckey($H,u,v^*,\infty$) 
			}
			\ForEach{ pair $(u^*,v) \in H$}{
				\deckey($H,u^*,v,\infty$) 
			}
		}
		\If{$\fp=I$}
		{
			\ForEach{ pair $(u,v) \in H$, such that either $u$ or $v$ or both lie in  $N_{G_s}(u^*) \cup N_{G_s}(v^*)$}{
				\deckey($H,u,v,2$) 
			}
		}
		\ElseIf{$\fp=U$}
		{
			\ForEach{pair $(u,v) \in H$, such that $u\in (N_{G_s}(u^*)-N_{G_s}(v^*))$, $v \in (N_{G_s}(v^*)-N_{G_s}(u^*))$}{
				\inckey($H,u,v,2$) 
			}
			\ForEach{pair $(u,v) \in H$, such that both $u,v\in (N_{G_s}(u^*) \cap N_{G_s}(v^*))$}{
				\deckey($H,u,v,2$) 
			}
		}
		\ElseIf{$\fp=E$}
		{
			\ForEach{ pair $(u,v^*) \in H$}{
				$value \gets 2|N_{G_s}(u^*) \cap N_{G_s}(v^*) \cap N_{G_s}(u)|$ \;
					\deckey($H,u,v^*,value$) 
			}
			\ForEach{ pair $(u^*,v) \in H$}{
				$value \gets 2|N_{G_s}(u^*) \cap N_{G_s}(v^*) \cap N_{G_s}(v)|$ \;
					\deckey($H,u^*,v,value$) 
			}
			\ForEach{ pair $(u,v)$ such that either $u$ or $v$ or both lie in $N_{G_s}(u^*) \cap N_{G_s}(v^*)$}{
				\deckey($H,u,v,2$) 
			}
		}
	}
\end{algorithm}

\subsection{Working Example}

The working of the \cn algorithm on the example stated in
Table~\ref{tab:example} is shown in Table~\ref{tab:working}. Fig.~(a),~(e)
and~(i) show the input graphs for $\fp=I,U,E$ respectively. For
$\fp=I,U,E$, there are 2, 2 and 1 merge operations respectively. The
compression gains achieved in each merge operation is stated in
Table~\ref{tab:working}. The final summary graphs for $\fp=I,E$ and $U$ are
shown in Fig.~(c),~(g) and~(j) respectively.

For example, the compression gain in the first (and only) step of \iu is $|G| - |G_s|$. Using Eq.~\eqref{eq:sizeg}, $|G| = 7 + 2 \times 10 = 27$.
Similarly, $|G_s| = 8 + 2 \times 7 + 2 \times 1 = 24$ ($2$ is added for $1$ supernode).
Hence, the compression gain is $27 - 24 = 3$.

\begin{algorithm}[bh]
	\caption{\decomp \label{alg:decomp}}
	\SetKwFunction{Fdecomp}{\decomp}
	\Fn{\Fdecomp{$G_s$}}{ 
	$V_r \gets \emptyset$, $E_r \gets \emptyset$\;
	\ForEach{ node $u \in V_s$}{
		$U \gets \snmap(G_s,u)$, 
		$V_r \gets V_r \cup U$
		}
		\ForEach{ edge $(u,v) \in E_s$}{
			$U \gets \snmap(G_s,u), V \gets  \snmap(G_s,v)$\;
			$E_r \gets E_r \cup (U\times V)$ \label{line:decomp-a}\;
			}
			\KwRet $G_r=(V_r,E_r)$  \;
			}
\end{algorithm}

\subsection{The \decomp Algorithm}
\label{sec:decomp}

The \emph{decompression (or reconstruction) algorithm} (Algo.~\ref{alg:decomp}) takes as input a given
summary graph $G_s=(V_s,E_s)$ that is generated by the \cn framework, and
computes the reconstructed graph $G_r=(V_r,E_r)$.  For each node $u^* \in G_s$,
the algorithm computes the ancestors $\snmap_{G_s}(u^*)$ using the parent
relationships. The set of nodes $V_r$ is the union of all such ancestor nodes,
i.e., $V_r=\{u|u \in \snmap_{G_s}(u^*)|u^* \in G_s\}$. The edge set $E_r$ is
the union of all edges formed by taking the Cartesian product of the ancestors of the end-vertices corresponding to each edge in the summary graph. Formally,
$E_r = \{(u,v) | (u^*,v^*) \in E_s, u \in \snmap_{G_s}(u^*), v \in \snmap_{G_s}(v^*) \}$.

Referring to the example stated in Table~\ref{tab:example}, the reconstructed
graphs returned for $\fp=I,E,U$ are shown in Fig.~(d), (g) and (k)
respectively.  The vertex set $V_r=V$ for $\fp=I,E,U$.  For $\fp=I$, the edge
set $E_r \subseteq E$; for $\fp=E$, $E_r = E$; and, for $\fp=U$, $E_r \supseteq
E$.  The neighborhood loss for each $u \in G$ is $rl(u) \le \eps$, for
$\fp=I,E,U$.

\section{Query Processing through \cn}
\label{sec:query}

\subsection{Local versus Global Decompression}

The \cn framework can process the queries in two ways: (1)~Querying with local
decompression, (2)~Querying with global decompression. The former case is
applicable when it is not possible to reconstruct the full graph due to space
constraints or the query workload is sufficiently small. In this scenario,
the queries are run on the summary graph $G_s$ directly without reconstructing the graph $G_r$. More specifically, it performs \emph{local
decompression}, i.e., decompressing only those nodes that are necessary to
answer the given query. The latter case is applicable when there is no space
constraint to prohibit the reconstruction of the graph $G_r$ and the query workload is sufficiently large. In this scenario, it is better to reconstruct the
 graph $G_r$, and then run the queries on $G_r$, since each query
runs faster on $G_r$ than $G_s$ due to no requirement of local decompression.

It should be noted that local decompression and global decompression yield
the \emph{same} query answer.  This is because the local decompression steps
essentially mimic the global decompression, but only at a per-neighborhood (i.e.,
local) level. The global decompression runs the steps for all
neighborhoods.

\subsection{Query Processing}

Given that the \cn framework allows reconstruction of the input graph, and evaluation of neighborhood queries with either \emph{no loss}, or with \emph{bounded loss guarantees} depending on the choice of the summarization variant, the proposed framework can be used to retrieve multiple graph queries with either \emph{no loss} or with fairly \emph{high accuracy}, as demonstrated in Sec.~\ref{sec:results}. As examples of graph queries, we chose to evaluate the neighborhood queries,  reachability queries and shortest path queries. Next,
we describe the processing of these queries.

\begin{algorithm}[tbh]
	\caption{\nq \label{alg:nq}}
	\SetKwFunction{Fnq}{\nq}
	\Fn{\Fnq{$G_s,u$}}{ 
		$N_{G_r}(u) \gets \emptyset$ \;
		$Desc \gets$ \desc$(G_s,u)$\;
		\ForEach{ node $v \in Desc$}{
			\ForEach{ node $w \in N_{G_s}(v)$}{
				$N_{G_r}(u) \gets N_{G_r}(u) \cup \snmap(G_s,w)$\;
			}
		}
		\KwRet $N_{G_r}(u)$\;
	}
\end{algorithm}

\begin{algorithm}[tbh]
	\caption{\desc \label{alg:desc}}
	\SetKwFunction{Fdesc}{\desc}
	\Fn{\Fdesc{$G_s,u$}}{ 
		$Nodes \gets \emptyset$\;
		Let $Stack$ be an empty stack.\;
		$Stack.push(u)$\;
		\While{$Stack$ is non-empty}{
			$v=Stack.pop()$ \;
			\ForEach{$w \in \child_{G_s}(v)$}{
				$Nodes \gets Nodes \cup \{w\}$\;
				$Stack.push(w)$
			}
		}
		\KwRet $Nodes$\;
	}
\end{algorithm}

\subsubsection{{\bf Neighborhood Query}} 
To answer the neighborhood query $N_G(u)$ posed on a node $u\in G$, we report
its neighborhood $N_{G_r}(u)$ in the reconstructed graph $G_r$.  First, the
\child sets are computed using parent relationships: if $u \in P_{G_s}(v)$,
then $v \in \child_{G_s}(u)$. Next, for the given query node $u$, its
\emph{descendants} are computed by enumerating the \child sets in a depth-first
manner. A node $v$ is said to be a descendant of $u$, if $u$ is an ancestor of
$v$.  Subsequently, for each descendant $v$ of $u$, its neighborhood
$N_{G_s}(v)$ is computed. For each $w \in N_{G_s}(v)$, the ancestors of $w$ are
reported as the answer $N_{G_r}(u)$. The detailed procedures are given in Algo.~\ref{alg:nq} and Algo.~\ref{alg:desc}.

\subsubsection{{\bf Reachability Query and Shortest Path Query}} 
We answer the reachability query $R_G(u,v)$ and the distance query
$\dist_G(u,v)$ on a node pair $u,v \in G$, by its corresponding queries
$R_{G_r}(u,v)$ and $\dist_{G_r}(u,v)$ respectively in the reconstructed graph
$G_r$.  The reachability query is solved by a breadth-first search on the
reconstructed graph $G_r$ using the neighborhood querying procedure, discussed
above. The search begins at node $u$ and terminates when it finds the node $v$
or when it has discovered all the nodes in the underlying connected component
of $G_r$.  The shortest path query is solved using the standard Dijkstra's
algorithm \cite{cormen2022introduction}.

\subsection{Query Characteristics}
\label{sec:query characteristics}

We next verify the query characteristics of the three summarization
variants of \cn, as stated in Table~\ref{tab:query properties}. Since \iu is a
lossless scheme, it retrieves all the queries exactly. 

Let us next consider \inter. From Eq.~\eqref{eq:variants}, Eq.~\eqref{eq:varrl}
and Eq.~\eqref{eq:loss-constraint}, it follows that for any node $u$, $N_{G_r}(u)
\subseteq N_G(u)$, $1-\eps_u)|N_G(u)| \le |N_{G_r}(u)| \le |N_G(u)|$.
To prove that $R_{G_r}(u) \subseteq R_G(u)$,
it is sufficient to show that $R_{G_r}(u) \setminus R_G(u)= \emptyset$.
Since no new edge is added in $G_r$ with respect to $G$, if $v$ is \emph{not
reachable} from a given node $u$ in $G$, then it continues to remain
unreachable in $G_r$.  However, since $G_r$ loses edges present in $G$, it may
happen that $v$ that was originally reachable from $u$ in $G$, now becomes
unreachable in $G_r$.  Thus, $R_{G_r}(u) \subseteq R_G(u)$.  Following the same
logic, since $E_r \subseteq E$, the \emph{shortest} path length between any
pair of nodes $u,v$ can only get increased in $G_r$. Hence, $\dist_{G_r}(u,v)
\ge \dist_G(u,v)$.

Next, let us consider \uni. From Eq.~\eqref{eq:variants}, Eq.~\eqref{eq:varrl} and
Eq.~\eqref{eq:loss-constraint}, it follows that for any node $u$, $N_{G_r}(u)
\supseteq N_G(u)$, $|N_G(u)| \le |N_{G_r}(u)| \le (1+\eps_u)|N_G(u)|$.  Next,
consider the reachability query.  Since no edge is lost in $G_r$ that is
present in $G$, if $u$ is reachable from any given node $v$ in $G$, then they
continue to be so in $G_r$, as well. Thus, $R_{G}(u) \subseteq R_{G_r}(u)$.
Although the reconstructed graph $G_r$ generated by \uni may contain extra edges with respect to the original graph $G$,  the extra edges are always
between the 2-hop neighbors. This indicates that nodes that remain in separate
connected components in $G$, continue to remain in separated connected components in $G_r$. In
other words, the connected components remain unaltered. Thus, $R_{G_r}(u)
\subseteq R_G(u)$. Together, $R_{G_r}(u) = R_G(u)$.  Since $E_r \supseteq E$,
i.e., extra edges are possibly added, lengths of paths can get reduced.  Hence,
$\dist_{G_r}(u,v) \le \dist_G(u,v)$.

\section{Correctness of \cn}
\label{sec:cgs_correctness}

The following result helps to establish the correctness of the \cn algorithm.

The correctness of the \cn algorithm is established through the correctness of the procedures: 
\updatedegree and \issafemerge, that are established as follows

\begin{lem} \label{lem:updatedegree}
	The procedure \updatedegree (Algo.~\ref{alg:updatedegree}) is correct.
\end{lem}

\begin{proof}
	To prove the correctness of the procedure \updatedegree, we need to show that after each merge operation, the procedure correctly computes the degree $\deg_{G_r}(u)$ for each node $u \in G_r$. Note that this procedure is executed only for $\fp \in \{I,U\}$. For $\fp=E$, the degree $\deg_{G_r}(u)=\deg_G(u)$ for each node $u\in G$. Hence there is no necessity to update the degrees. Next, consider $\fp \in \{I,U\}$. Firstly, we note that the given procedure correctly computes the set of all the nodes in $G_r$ (the sets $U^*,V^*,U'$ and $V'$, as stated in Algorithm~\ref{alg:updatedegree}) whose degrees are possibly altered as a result of the given merge operation. In other words, the procedure does not miss any node $u \in G_r$ whose degree gets altered due to the given merge operation. Next, it is noted that the procedure correctly computes the new degree for each such node $u \in G_r$ whose degree gets altered due to the given merge operation. This can be shown inductively.
\end{proof}

\begin{lem} \label{lem:issafemerge}
	The procedure \issafemerge (Algo.~\ref{alg:issafemerge}) is correct.
\end{lem}

\begin{proof}
	To prove the correctness of the \issafemerge procedure, it is sufficient to show that the procedure returns true if and only if it is safe to merge any pair of nodes, $u^*,v^* \in G_s$, i.e., it would not violate the neighborhood loss constraint (Eq.~\eqref{eq:loss-constraint}). If $\fp=E$, this fact is trivially true. If $\fp \in \{I,U\}$, the algorithm carefully scans each node $u \in G_r$, whose degree may get changed due to the given merge operation. To this end, the given procedure computes the sets  $U^*,V^*,U'$ and $V'$, as stated in Algorithm~\ref{alg:issafemerge}, that contain all the nodes whose degree may get altered due to a merge operation. 
	From Lem.~\ref{lem:updatedegree}, it is clear that the \cn algorithm correctly keeps track of $\deg_{G_r}(u)$ for each node $u \in G_r$. Using these degree value $\deg_{G_r}(u)$, the \issafemerge procedure  computes the possible updated degree of $u$ (considering realization of the given potential merge operation) and checks if the new degree values would violate the neighborhood loss constraint (Eq.~\eqref{eq:loss-constraint}). The algorithm returns true if and only if each such node $u$ passes this test. 
\end{proof}

\begin{thm} \label{thm: correctness of cn appendix}
	Given an input graph $G=(V,E)$ that is summarized by the \cn algorithm to produce the summary graph $G_s=(V_s,E_s)$ that is later decompressed to produce the reconstructed graph $G_r=(V_r,E_r)$. Then, $V_r=V$. Also, \uni, \iu and \inter respectively guarantee that $E_r \supseteq E$, $E_r=E$ and $E_r \subseteq E$. Additionally, $G_r$ respects the neighborhood loss constraint stated in Eq.~\eqref{eq:loss-constraint}. 
\end{thm}

\begin{proof}
	Suppose there is a node $u \in G$, such that $u \notin G_r$. Since the \cn algorithm does not remove a node unless it merges with another node, there must exist a node $u^*$ such that $u$ is an ancestor of $u^*$. If $u$  did not undergo any merge operation, then $u^*=u$. From the reconstruction algorithm, it follows that $u \in G_r$. This is a contradiction. Hence, $V \subseteq V_r$. Similarly, it can be proved that $V_r \subseteq V$. Hence, $V=V_r$.
	
	Next, consider $\fp=U$. Suppose there exists an  edge  $(u,v) \in E$ such that $(u,v) \notin E_r$. Referring to the merge operation for $\fp=U$ (Sec.~\ref{sec:merge}),  
	there must exist at least a pair of nodes, $u^*,v^* \in V_s$, (where $u^* \neq v^*$) such that $u$ and $v$ are ancestors of $u^*$ and $v^*$, respectively, and  $(u^*,v^*) \in E_s$. Using the reconstruction algorithm, since 
	$(u^*,v^*) \in E_s$, therefore, $(u,v) \in E_r$, which is a contradiction. Hence, \uni ensures that $E_r \supseteq E$.
	
	Similarly, \inter ensures that $E_r \subseteq E$, and \iu ensures that $E_r=E$.
	
	Next, consider $\fp\in \{I,U\}$.  In either case, from Lem.~\ref{lem:issafemerge}, it follows that a merge operation is realized only after it is confirmed that it would not violate the neighborhood loss constraint (Eq.~\eqref{eq:loss-constraint}). Therefore, when the algorithm terminates, it is guaranteed that the given constraint remains satisfied. The reconstruction algorithm also respects the neighborhood loss constraint.  
\end{proof}

From the above discussion, it follows that the summary graph $G_s$ generated by the Configurable Graph Summarizer (\cn) framework need not necessarily have the optimal compression ratio but the reconstructed graph $G_r$ and the query answers  satisfy all the constraints stated in the \gc problem (Problem~\ref{prob:gc}).

\section{Graph Properties that Guarantee Compression by the \cn Framework}
\label{sec:properties} 

The following results analyze the necessary and sufficient conditions to guarantee that $|G_s| < |G|$, i.e., $cr<1$, for each of the three algorithms, namely, \iu, \inter and \uni.

\begin{thm} \label{thm:property_iu appendix}
	Let $G_s$ be the summary graph of $G$ as produced by the \iu algorithm. Then, 
	$|G_s| < |G|$ if and only if there exists at least a pair of nodes $u,v \in G$ such that either  $|N_G(u) \cap N_G(v)| \ge 2$ or $|N_G(u) \cap N_G(v)| =1\ \wedge N_G(u)=N_G(v)$.
\end{thm}

\begin{proof}
	Suppose the graph $G$ has a pair of nodes $u,v \in G$ such that $|N_G(u) \cap N_G(v)| \ge 2$. It is sufficient to show that the \iu algorithm executes at least one \merge operation, and eventually produces a graph $G_s$ such that $|G_s|<|G|$. For this to happen, it is sufficient to show that the heap $H$ has at least one pair of nodes that offer positive compression gain. If such a pair exists, the \iu algorithm would eventually extract and merge such a pair.  We show that the pair $(u,v)$ satisfies the above requirement.
	
	From the \computecg algorithm (Algo.~\ref{alg:computecg}),  we observe that for the given pair $(u,v)$, $cg(u,v) \ge 2|N_G(u) \cap N_G(v)|-3+x$ where $x$ is given as follows:
	\begin{align} \label{eq:x}
		x = \begin{cases}
			0 & \text{ if } |N_G(u)-N_G(v)| \ge 1 \wedge |N_G(v)-N_G(u)| \ge 1\\
			1 & \text{ if either } N_G(u)\subset N_G(v) \text{ or } N_G(v)\subset N_G(u) \\
			2 & \text{ if } N_G(u)=N_G(v) 
		\end{cases}
	\end{align} 
	
	If  $|N_G(u) \cap N_G(v)| \ge 2$ then $cg(u,v)\ge1$. If $|N_G(u) \cap N_G(v)| =1$ and $N_G(u)=N_G(v)$, then $x=2$. Consequently, $cg(u,v)=1$. Thus, in either of the above cases, $cg(u,v)>0$.
	It follows that the \iu algorithm would have executed at least one merge operation, thereby producing a graph $G_s$ such that $|G_s| < |G|$.
	
	Next consider the case $|G_s|<|G|$. This is only possible if the graph $G$ undergoes at least one iteration of \merge. This can only happen if there exists at least a pair of nodes $u,v\in G$ such that $cg(u,v)>0$. From the \computecg algorithm (Algo.~\ref{alg:computecg}), we know that $cg(u,v) \ge 2|N_G(u) \cap N_G(v)|-3+x$ where $x$ is given in Eq.~\eqref{eq:x}. Thus, if $cg(u,v)>0$, it implies that  $G$ has a pair of nodes $u,v \in G$ such that $|N_G(u) \cap N_G(v)| \ge 2$ or $|N_G(u) \cap N_G(v)| =1\ \wedge N_G(u)=N_G(v)$.
\end{proof}

\begin{thm} \label{thm:property_inter appendix}
	Let $G_s$ be the summary graph of $G$ as produced by the \inter algorithm. Then, 
	$|G_s| < |G|$ if and only if there exists at least a pair of nodes $u,v \in G$ such that $v \in N^2_G(u)$, $|N_G(u)|\le |N_G(u) \cap N_G(v)|/(1-\eps_u)$ and $|N_G(v)|\le |N_G(u) \cap N_G(v)|/(1-\eps_v)$.
\end{thm}

\begin{proof}
	Suppose the graph $G$ has a pair of nodes $u,v \in G$ such that $v \in N^2_G(u)$,  $|N_G(u) \cup N_G(v)| \ge 2$, $|N_G(u)|\le |N_G(u) \cap N_G(v)|/(1-\eps_u)$ and $|N_G(v)|\le |N_G(u) \cap N_G(v)|/(1-\eps_v)$. 
	It is sufficient to show that the \inter algorithm executes at least one \merge operation, and eventually produces a graph $G_s$ such that $|G_s|<|G|$. For this to happen, it is sufficient to show that the heap $H$ has at least one pair of nodes that offer positive compression gain and their merging does not violate the neighborhood loss constraint. If such a pair exists, the \inter algorithm would eventually extract and merge such a pair.  We show that the pair $(u,v)$ satisfies the above requirement.
	
	Firstly,  from  the \computecg algorithm (Algo.~\ref{alg:computecg}),  we observe that for a given pair $u,v \in G$, $cg(u,v) \ge 2|N_G(u) \cup N_G(v)|-3 +y$ where $y$ is given as follows:
	\begin{align} \label{eq:y}
		y = \begin{cases}
			0 & \text{ if } (u,v) \in E \\
			2 & \text{ otherwise } 
		\end{cases}
	\end{align}   
	
	Since $v \in N^2_G(u)$, therefore, it follows that $|N_G(u) \cup N_G(v)|\ge 1$. Moreover, if $(u,v) \in E$, then, $|N_G(u) \cup N_G(v)|\ge 2$. Hence, for either of the cases: $(u,v)\in E$ or $(u,v)\notin E$, $cg(u,v)>0$.
	
	Next, we note that $rl(u)=|N_G(u)-N_{G_r}(u)|/|N_G(u)|= (|N_G(u)-(N_G(u) \cap N_G(v))|)/|N_G(u)|$. Since $(N_G(u) \cap N_G(v)) \subseteq N_G(u)$, therefore, the above expression simplifies to $rl(u)=1-|N_G(u) \cap N_G(v)|/|N_G(u)|$. From the condition stated above, $|N_G(u) \cap N_G(v)| \ge (1-\eps_u)|N_G(u)|$. Hence, $rl(u) \le \eps_u$. In a similar manner, it can be shown that $rl(v) \le \eps_v$. Therefore, if the pair of nodes $(u,v)$ is merged, it does not violate the neighborhood loss constraint.  It follows that the \inter algorithm would have executed at least one merge operation, thereby producing a graph $G_s$ such that $|G_s| < |G|$.

	Next consider the case $|G_s|<|G|$. This is only possible if the graph $G$ undergoes at least one iteration of \merge. This can only happen if there exists at least a pair of nodes $u,v\in G$ such that $cg(u,v)>0$ and their merging does not violate the neighborhood loss constraint. From the \computecg algorithm (Algo.~\ref{alg:computecg}), we know that $cg(u,v) \ge 2|N_G(u) \cup N_G(v)|-3+y$ where $y$ is given by Eq.~\eqref{eq:y}. Thus, if $cg(u,v)>0$, then it implies that  $v \in N^2_G(u)$. Further, if the merging of the pair $(u,v)$ does not violate the neighborhood loss constraint, then,  $rl(u)\le \eps_u$ and $rl(v) \le \eps_v$. It follows that the graph $G$ has a pair of nodes $u,v \in G$ such that $v \in N^2_G(u)$,  $|N_G(u)|\le |N_G(u) \cap N_G(v)|/(1-\eps_u)$ and $|N_G(v)|\le |N_G(u) \cap N_G(v)|/(1-\eps_v)$.
\end{proof}

\begin{thm} \label{thm:property_uni appendix}
	Let $G_s$ be the summary graph of $G$ as produced by the \uni algorithm. Then, 
	$|G_s| < |G|$ if and only if there exists at least a pair of nodes $u,v \in G$ such that either $|N_G(u) \cap N_G(v)| \ge 2$ or $|N_G(u) \cap N_G(v)| =1 \wedge (u,v)\notin E$; $|N_G(u)|\ge |N_G(u) \cup N_G(v)|/(1+\eps_u)$ and $|N_G(v)|\ge |N_G(u) \cup N_G(v)|/(1+\eps_v)$.
\end{thm}

\begin{proof}
	Suppose the graph $G$ has a pair of nodes $u,v \in G$ such that $|N_G(u) \cap N_G(v)| \ge 2$,  $|N_G(u)|\ge |N_G(u) \cup N_G(v)|/(1+\eps_u)$ and $|N_G(v)|\ge |N_G(u) \cup N_G(v)|/(1+\eps_v)$. 
	It is sufficient to show that the \uni algorithm executes at least one \merge operation, and eventually produces a graph $G_s$ such that $|G_s|<|G|$. For this to happen, it is sufficient to show that the heap $H$ has at least one pair of nodes that offer positive compression gain and their merging does not violate the neighborhood loss constraint. If such a pair exists, the \uni algorithm would eventually extract and merge such a pair.  We show that the pair $(u,v)$ satisfies the above requirement.

	Firstly, from the \computecg algorithm (Algo.~\ref{alg:computecg}), we observe that for a given pair $(u,v)$, $cg(u,v) \ge 2|N_G(u) \cap N_G(v)|-3 +z$ where $z$ is given as follows:  
	\begin{align} \label{eq:z}
		z = \begin{cases}
			0 & \text{ if } (u,v) \in E \\
			2 & \text{ otherwise } 
		\end{cases}
	\end{align} 

	Since $v \in N^2_G(u)$, therefore,  $|N_G(u) \cap N_G(v)| \ge 1$. Note that if $(u,v)\notin E$, then $z=2$. Consequently, $cg(u,v)>0$. Else, if $|N_G(u) \cap N_G(v)| \ge 2$, then $cg(u,v)>0$.

	Next, we note that $rl(u)=|N_{G_r}(u)-N_G(u)|/|N_G(u)|= (|(N_G(u) \cup N_G(v)-N_G(u))|)/|N_G(u)|$. Since $N_G(u) \subseteq (N_G(u) \cup N_G(v)) $, therefore, the above expression simplifies to $rl(u)=|N_G(u) \cup N_G(v)|/|N_G(u)|-1$. From the condition stated above, $|N_G(u) \cup N_G(v)| \le (1+\eps_u)|N_G(u)|$. Hence, $rl(u) \le \eps_u$. In a similar manner, it can be shown that $rl(v) \le \eps_v$. Therefore, if the pair of nodes $(u,v)$ is merged, it does not violate the neighborhood loss constraint.  It follows that the \uni algorithm would have executed at least one merge operation, thereby producing a graph $G_s$ such that $|G_s| < |G|$.

	Next consider the case $|G_s|<|G|$. This is only possible if the graph $G$ undergoes at least one iteration of \merge. This can only happen if there exists at least a pair of nodes $u,v\in G$ such that $cg(u,v)>0$ and their merging does not violate the neighborhood loss constraint. From the \computecg algorithm (Algo.~\ref{alg:computecg}), we know that $cg(u,v) \ge 2|N_G(u) \cap N_G(v)|-3+z$ where $z$ is given in Eq.~\eqref{eq:z}. Thus, if $cg(u,v)>0$, then it implies that   either $(u,v) \notin E$ or $|N_G(u) \cap N_G(v)| \ge 2$. Further, if the merging of the pair $(u,v)$ does not violate the neighborhood loss constraint, then,  $rl(u)\le \eps_u$ and $rl(v) \le \eps_v$.  It follows that the graph $G$ has a pair of nodes $u,v \in G$ such that  either $|N_G(u) \cap N_G(v)| \ge 2$ or $|N_G(u) \cap N_G(v)| =1 \wedge (u,v)\notin E$;  $|N_G(u)|\ge |N_G(u) \cup N_G(v)|/(1+\eps_u)$ and $|N_G(v)|\ge |N_G(u) \cup N_G(v)|/(1+\eps_v)$.
\end{proof}

\section{Complexity of \cn}
\label{sec:complexity}

In this section, we analyze the time complexity of the \cn algorithm. Suppose
the maximum degree of any node $u \in G$ is $\maxd$, i.e.,  $|N_G(u)| \le
\maxd$. Assume the neighborhood sets $N_G(u),N_{G_s}(u)$ are maintained as
hash-sets for each node $u$. Thus, neighborhood checks such as$v \in N_G(u)$
are answered in $O(1)$ time. Further, let $|V|=n$, and $|E|=m$. Let
$\eps=\max\{\eps_u|u \in G\}$. The complexity of the \cn algorithm is
established through the following series of results.

\begin{lem}
	\label{lem:eligible pairs}
	Given a graph $G$ of $n$ nodes, the number of pairs of nodes that are eligible to be merged is $O(\min(m\maxd,\binom{n}{2}))$.
\end{lem}

\begin{proof}
	A pair of nodes $u,v\in G$ is eligible to be merged if $v \in N_G^2(u)$.
	The number of 2-hop neighbors of $u \in G$, is $|N_G^2(u)|=\sum_{v \in N_G(u)} (\deg_G(v)-1)$.  Thus, the total number of pairs of nodes that are eligible to be merged is at most $\sum_{u \in G}|N_G^2(u)| \le \sum_{u \in G}\sum_{v \in N_G(u)} (\deg_G(v)-1) \le \sum_{u \in G} \deg_G(u)(\deg_G(u)-1)\le 2m\maxd$. Given that the graph $G$ has $n$ nodes, the total number of possible pairs of nodes is $\binom{n}{2}$. Therefore, the number of eligible pairs is  $O(\min(m\maxd,\binom{n}{2}))$. 
\end{proof}

\begin{lem} \label{lem:initial-heap-size}
	The \buildheap procedure creates a binary max-heap $H$, whose size, i.e., the number of nodes in $H$, is $|H|=O(m\maxd)$.
\end{lem}

\begin{proof}
	From Lem.~\ref{lem:eligible pairs}, the number of pairs eligible to be merged is $O(m\maxd)$.
	Note that the \buildheap procedure inserts one node $h$ into the heap $H$ for each such pair $(u,v)$ whose $cg(u,v)>0$.  Thus, the total number of nodes inserted by the \buildheap procedure is $|H|=O(m\maxd)$. 
\end{proof}

\begin{lem} \label{lem:buildheap-complexity}
	The \buildheap procedure runs in $O(m\maxd)$ time.
\end{lem}

\begin{proof}
	The \nic procedure, when invoked with parameters $(G,u)$, runs in $O(\deg_G(u)\maxd)$ time. The \nic procedure computes the value  $\ic(u,v)$ for each pair $(u,v)$ such that $v \in N_G^2(u)$, where $\ic(u,v)=|N_G(u) \cap N_G(v)|$. Using these $\ic$ values, the \computecg procedure runs in $O(1)$ time.  From Lem.~\ref{lem:eligible pairs}, the number of pairs eligible to be merged is $O(m\maxd)$. Since \computecg is called for each such pair, the \buildheap procedure requires $O(m\maxd)$ time.
\end{proof}

\begin{lem} \label{lem:number-of-super-nodes}
	Let $S$ be the set of super-nodes in the summary graph $G_s$. Then, 
	\begin{align}
		|S| \le \begin{cases}
			m & \text{ if } \fp \in \{E,U\}\\
			n-1 & \text{ if } \fp= I
		\end{cases}
	\end{align}
\end{lem}

\begin{proof}
	Consider the case $\fp\in\{E,U\}$. A given pair $(u,v)$ is merged to form exactly one super-node if $cg(u,v)>0$. For this to happen, referring to Eq.~\eqref{eq:cg}, it follows that the number of edges in $G_s$ must reduce by at least $1$ after the merging of $(u,v)$. 
	Given that $G_s$ is initialized to $G$ which has $m$ edges,  the number of super-nodes that are added into $G_s$ is at most $m-|E_s| \le m$.
	
	Now suppose $\fp = I$. In this scenario, each super-node has exactly 2 parents, and more importantly, each node has at most one child. In view of this fact, the evolution of super-nodes during the course of the \cn algorithm, can be modeled as a binary tree where the nodes in the input graph $G$ form the leaf nodes, and the super-nodes are represented by the internal nodes. Since a binary tree with $n$ leaf nodes has at most $n-1$ internal nodes, therefore, the number of super-nodes produced by the \cn algorithm is at most $n-1$.
\end{proof}

\begin{lem} \label{lem:size of neighborhood}
	At any stage of the \cn algorithm, the size of the 1-hop neighborhood of any node $u \in G_s$ is bounded as follows:
	\begin{align}
		|N_{G_s}(u)| &\le
		\begin{cases}
			\maxd & \text{if } \fp \in\{ I,E\} \\
			\maxd(1+\eps) & \text{if } \fp=U
		\end{cases}
	\end{align}
\end{lem}

\begin{proof}
	At the beginning of the \cn algorithm, the \greedycn procedure is invoked that initializes the graph $G_s$  to $G$. At this stage, the neighborhood of any node $u\in G_s$ ia bounded as follows: $|N_{G_s}(u)| \le |N_{G}(u)| \le \maxd$. Next consider any subsequent stage of the algorithm.
	
	Consider any node $u \in G_s$ which could either be a simple node or a super-node. Firstly, suppose $u$ is a simple node. Referring to the \merge procedure (Algorithm~\ref{alg:merge}), one finds that during the course of the algorithm, the size of the neighborhood set $N_{G_s}(u)$ either shrinks or remains unchanged. If a new neighbor is added, simultaneously, one or more existing neighbors are removed. 
	Thus, for any such node, $|N_{G_s}(u)| \le \maxd$. Next, suppose $u$ is a super-node and $v$ is one of its ancestors, i.e., $v \in \snmap(G_s,u)$.
	
	Consider the case $\fp \in\{ I,E\}$. From the \merge procedure, it is clear that $N_{G_s}(u) \subseteq N_{G}(v)$. Thus, it follows that 
	$|N_{G_s}(u)| \le |N_{G}(v)| \le \maxd$. Further, during the course of the algorithm, the size of the neighborhood $|N_{G_s}(u)|$ never increases. Thus, at each stage of the algorithm, $|N_{G_s}(u)|\le \maxd$. 
	
	Next, consider $\fp=U$. Suppose $|N_{G_s}(u)|> \maxd(1+\eps)$.  At this stage, suppose $G_s$ is decompressed. From the \decomp procedure, it follows that $|N_{G_r}(v)|\ge |N_{G_s}(u)|> \maxd(1+\eps) $. This violates the neighborhood loss constraint stated in Eq.~\eqref{eq:loss-constraint}. This contradicts the correctness of the \cn algorithm, established in Th.~\ref{thm: correctness of cn appendix}. 
\end{proof}

\begin{lem} \label{lem: max heap size}
	At any stage of the \cn algorithm, the maximum size of the heap $H$ is given as follows:
	\begin{align} \label{eq:max heap size}
		|H| =\begin{cases}
			O(m\maxd+n\maxd^2) & \text{ if } \fp  =I \\
			O(m\maxd^2) & \text{ if } \fp \in \{E,U\}
		\end{cases}
	\end{align}
\end{lem}

\begin{proof}
	From Lem.~\ref{lem:initial-heap-size}, it follows that at the end of the \buildheap procedure, the size of the heap $H$ is $O(m\maxd)$. Subsequently, suppose the algorithm produces $k$ super-nodes. From Lem.~\ref{lem:number-of-super-nodes}, it follows that if $\fp = I$, $k=O(n)$; and if $\fp\in \{E,u\}$, then $k=O(m)$.
	For each such super-node $s$, and for each $u \in N^2_{G_s}(s)$, the \updateheap procedure inserts a node into the heap $H$. From Lem.~\ref{lem:size of neighborhood}, and using the fact $\eps \le 1$, it follows that $|N^2_{G_s}(s)|=O(\maxd^2)$.   Combining these facts, lead to the above result.
\end{proof}

\begin{lem} \label{lem:iterations of greedycn}
	The number of iterations of the while loop stated in line~\ref{line:begin-while} of the \greedycn procedure is given as follows:
	\begin{align*}
		O(|H|) &  \text{ for } \fp \in \{I,U\}\\
		O(m) & \text{ for } \fp =E
	\end{align*}
	
	where $|H|$ is given in Eq.~\eqref{eq:max heap size}.
\end{lem}

\begin{proof}
	Consider the case  $\fp \in \{I,U\}$.
	The iterative phase of the \greedycn procedure (starting from line~\ref{line:begin-while}) continues as long as there exists a pair $(u,v) \in H$ such that $cg(u,v)>0$. Therefore, the number of iterations of the while loop (stated in line~\ref{line:begin-while}) of the \greedycn procedure is at most the maximum size of the heap $H$ that is stated in Eq.~\eqref{eq:max heap size}.
	
	Next, consider the case $\fp=E$. In each iteration, the \greedycn procedure extracts a pair of nodes and merges them. Since each merge operation produces exactly one super-node, and the number of super-nodes is at most $m$ (Lem.~\ref{lem:number-of-super-nodes}), hence the number of iterations is at most $O(m)$.
\end{proof}

\begin{lem} \label{lem:complexity of issafemerge}
	The \issafemerge procedure and the \updatedegree procedure run in $O(n)$ time.
\end{lem}

\begin{proof}
	Firstly, we show that the \snmap procedure runs in $O(n)$ time. Referring to the \snmap procedure, one finds that if $s$ is a super-node that has $k$ ancestors, then $\snmap(G_s,s$ takes $O(k)$ time. Since $k \le n$, hence, the \snmap procedure runs in $O(n)$ time.
	
	Next, referring to the \issafemerge procedure and the \updatedegree procedure, it is clear that the size of each of the sets $U^*,V^*,U'$ and $V'$ is at most $n$ and their computation takes $O(n)$ time. The remaining steps of these procedures run in $O(n)$ time. This leads to the above result.
\end{proof}

\begin{lem} \label{lem:complexity of updateheap}
	The \updateheap procedure runs in $O(\maxd^3\log n)$ time.
\end{lem}

\begin{proof}
	The \updateheap procedure begins by calling the \nic procedure with the parameters $(G_s,s)$. From Lem.~\ref{lem:size of neighborhood}, it follows that $|N^2_{G_s}(s)|=O(\maxd^2)$. Thus, 
	the number of pairs $(s,u)$ is $O(\maxd^2)$. Thus, the \nic procedure runs in $O(\maxd^2)$ time. for each pair of nodes $(s,u)$ where $u \in N_{G_s}^2$, the \computecg procedure runs in $O(1)$ time using the \ic values computed by the \nic procedure. 
	The total time spent in computation of $cg(s,u)$ for each pair $(s,u)$ is thus, $O(\maxd^2)$. For each pair $(s,u)$ if $cg(s,u)>0$, then the pair $(s,u)$ is inserted into the heap $H$. This insertion step takes $O(\log|H|)$ for a given pair. Using Eq.~\eqref{eq:max heap size}, and the relations: $m=O(n^2)$ and $\maxd=O(n)$, it follows that $O(\log|H|)=O(\log n)$. 
	In the subsequent steps, the \updateheap procedure scans at most $O(\maxd^3)$ nodes in the heap $H$, and for each such node, it invokes a \deckey or an \inckey operation that takes $O(\log|H|)=O(\log n)$ time. 
	Additionally, there are neighborhood set operations such as union, intersection and set difference, each of which requires $O(\maxd)$ time. Summing up all these computation costs lead to the above lemma.
\end{proof}

\begin{lem} \label{lem:complexity of merge}
	The \merge procedure runs in:
	\begin{align*}
		O(n+\maxd^3\log n) &  \text{ for } \fp \in \{I,U\}\\
		O(\maxd^3\log n) & \text{ for } \fp =E
	\end{align*}
\end{lem}

\begin{proof}
	The most expensive step of the \merge procedure that is common to all the three variants of the \cn algorithm is the \updateheap procedure. From Lem.~\ref{lem:complexity of updateheap}, this step runs in $O(\maxd^3\log n)$ time. For $\fp \in \{I,U\}$, there is an additional expensive step in the form of the \updatedegree procedure. From Lem.~\ref{lem:complexity of issafemerge}, this step runs in $O(n)$ time. For $\fp \in E$, this step is not required. This leads to the stated result.
\end{proof}

Using the above results, the following theorem, finally, establishes the time complexity of the \cn summarization scheme.

\begin{thm} \label{thm:complexity of cn}
	The \greedycn procedure runs in:
	\begin{align*}
		O(m\maxd^3\log n) & \text{ for } \fp =E \\
		O(mn\maxd+n^2\maxd^2 + (m\maxd + n^2+n\maxd^3) \log n) & \text{ for } \fp =I \\
		O(mn\maxd^2+m\maxd^3\log n) & \text{ for } \fp =U
	\end{align*}
\end{thm}

\begin{proof}
	The \greedycn procedure invokes the \buildheap procedure that takes $O(m\maxd)$ time (Lem.~\ref{lem:buildheap-complexity}). Next, the procedure invokes the iterative stage from line~\ref{line:begin-while}.
	
	Next consider the case $\fp=E$. Following Lem.~\ref{lem:iterations of greedycn}, the number of iterations is $O(m)$. In each iteration, the procedure extracts a pair $(u,v) \in H$. This takes $O(\log n)$ time.  Subsequently, the pair is merged by calling the \merge procedure. Following Lem.~\ref{lem:number-of-super-nodes},  the number of super-nodes produced is at most $O(m)$. From Lem.~\ref{lem:complexity of merge}, it follows that the total computation cost of \merge procedure across all the iterations, is $O(m\maxd+^3\log n)$. Finally, summing up all the above computation costs leads to $O(m\maxd+^3\log n)$.
	
	Consider the case $\fp = I$. From Lem.~\ref{lem:iterations of greedycn}, the number of iterations is $O(m\maxd+n\maxd^2)$. In each iteration, the procedure extracts a pair $(u,v) \in H$. This takes $O(\log n)$ time. Further, if $cg(u,v)>0$, then it is checked whether the merging of the pair $(u,v)$ is safe, using the \issafemerge procedure. Since the \issafemerge procedure runs in $O(n)$ time (Lem.~\ref{lem:complexity of issafemerge}), the total computation cost of \issafemerge across all the iterations, is $O(mn\maxd+n^2\maxd^2)$. If it is safe to merge $(u,v)$, the pair is merged by calling the \merge procedure. Following Lem.~\ref{lem:number-of-super-nodes},  the number of super-nodes produced is at most $O(n)$. From Lem.~\ref{lem:complexity of merge}, it follows that the total computation cost of \merge procedure across all the iterations, is 
	$O(n^2+n\maxd^3\log n)$. Finally, summing up all the above computation costs leads to $O(mn\maxd+n^2\maxd^2 + (m\maxd + n^2+n\maxd^3) \log n
	)$.
	
	Next, consider the case $\fp = U$. From Lem.~\ref{lem:iterations of greedycn}, the number of iterations is $O(m\maxd^2)$. In each iteration, the procedure extracts a pair $(u,v) \in H$. This takes $O(\log n)$ time. Further, if $cg(u,v)>0$, then it is checked whether the merging of the pair $(u,v)$ is safe, using the \issafemerge procedure. Since the \issafemerge procedure runs in $O(n)$ time (Lem.~\ref{lem:complexity of issafemerge}), the total computation cost of \issafemerge across all the iterations, is $O(mn\maxd^2)$. If it is safe to merge $(u,v)$, the pair is merged by calling the \merge procedure. Following Lem.~\ref{lem:number-of-super-nodes},  the number of super-nodes produced is at most $O(m)$. From Lem.~\ref{lem:complexity of merge}, it follows that the total computation cost of \merge procedure across all the iterations, is $O(mn+m\maxd^3\log n)$.
	Finally, summing up all the above computation costs leads to $O(mn\maxd^2+m\maxd^3\log n)$.
\end{proof}

\section{Experimental Results}
\label{sec:results}

In this section, we evaluate the summarization and the query processing capabilities of \cn. 
It is organized as follows. Sec.~\ref{sec:research questions} lists the key research questions addressed in this work. Sec.~\ref{sec:exp methodology} describes the methodology adopted for conducting the experimental evaluation.  Sec.~\ref{sec:exp compression ratio} and Sec.~\ref{sec:exp reconstruction error} present the comparison of compression ratio and reconstruction error respectively with existing baseline summarization techniques. Sec.~\ref{sec:exp effect of neighborhood loss tolerance threshold} assesses the effect of the neighborhood loss tolerance threshold. Sec.~\ref{sec:exp query processing} evaluates the query processing performance. The scalability results are discussed in Sec.~\ref{sec:exp scalability}. Finally, the key take aways of the experimental findings are summarized in Sec.~\ref{sec:exp summary}.    

\begin{table}[th]
		\begin{tabular}{l r r r r r}
			\toprule
			\multirow{2}{*}{\textbf{Graph}} & \textbf{\# Vertices} & \textbf{\# Edges} & \textbf{Avg. Degree} & \textbf{Max. Degree} & \textbf{Density} \\
			& \multicolumn{1}{c}{\bf $|V|$} & \multicolumn{1}{c}{\bf $|E|$} & \multicolumn{1}{c}{\bf $\frac{2 |E|}{|V|}$} & \multicolumn{1}{c}{\textbf{$d$}} & \multicolumn{1}{c}{\bf $\frac{2 |E|}{|V| (|V| - 1)}$} \\
			\midrule
			Youtube (YT)            & 1,134,890  & 2,987,624  &  5.27 & 28754 & 4.64E-06\\
			DBLP (DB)               &   317,080  & 1,049,866  &  6.62 &   343 & 2.09E-05\\
			Live-Journal (LJ)       &   119,685  &   400,000  &  6.68 & 15232 & 5.58E-05\\
			soc-Epinions1 (SE)      &    75,879  &   405,740  & 10.70 &  3044 & 1.41E-04\\
			Slashdot0811 (SD)       &    77,360  &   469,180  & 12.12 &  2539 & 1.57E-04\\
			Email-Enron (EE)        &    36,692  &   183,831  & 10.02 &  1383 & 2.73E-04\\
			CA-CondMat (CM)         &    23,133  &    93,439  &  8.08 &   279 & 3.49E-04\\
			Gemsec-Facebook (GF)    &    50,515  &   819,090  & 32.42 &  1469 & 6.42E-04\\
			CA-AstroPh (AP)         &    18,772  &   198,050  & 21.10 &   504 & 1.12E-03\\
			CA-HepPh (HP)           &    12,008  &   118,489  & 19.74 &   491 & 1.64E-03\\
			\bottomrule  
		\end{tabular}
	\caption{Summary of real graph datasets used in this study. The datasets are sorted based on their density.}
	\label{tab:datasets}
\end{table}

\subsection{Research Questions and Objectives}
\label{sec:research questions}

The primary goal of this work is to propose a general-purpose configurable graph summarization framework that offers high compression and good query support.
To this end, we strive to address the following pertinent research questions:
\begin{itemize}
	\item \textbf{RQ1:} How does the proposed \name framework perform against the existing summarization methods in terms of configurability and query support capabilities?
	\item \textbf{RQ2:} How does the compression achieved by the \name framework perform against the existing baseline schemes?
	\item \textbf{RQ3:} With respect to the lossy variants of the \cn framework, how does the reconstruction error compare against the existing baseline methods?
	\item \textbf{RQ4:} How does the neighborhood loss tolerance threshold \nlt impact the compression achieved by the lossy variants of the \name framework?
	\item \textbf{RQ5:} How is the performance of neighborhood queries, reachability queries and shortest path queries achieved by the \name framework?
	\item \textbf{RQ6:} Is the \name framework scalable for large graphs?
\end{itemize}

Based on the qualitative comparison of the configurable summarization properties listed in 
Table~\ref{tab:related} and the literature review in Sec.~\ref{sec:related}, it is clear that \cn offers better configurability and query support than the existing summarization methods. This addresses the research question RQ1.

The remaining research questions, namely, RQ2, RQ3, RQ4, RQ5, and RQ6, are addressed in Sec.~\ref{sec:exp compression ratio}, Sec.~\ref{sec:exp reconstruction error}, Sec.~\ref{sec:exp effect of neighborhood loss tolerance threshold}, Sec.~\ref{sec:exp query processing} and Sec.~\ref{sec:exp scalability}, respectively.

\subsection{Experimental Methodology}
\label{sec:exp methodology}

This section describes the experimental methodology used to probe the research questions.

\subsubsection{\bf Experimental Setup}

The proposed \cn framework is implemented in C++. The implementation code and the datasets are available at:
\begin{center}
\url{https://github.com/sonaelzasimon/CGS_Configurable_Graph_Summarization}
\end{center}
The experiments are evaluated
on a 2.6 GHz Intel 56-core machine equipped with 504 GB RAM running Ubuntu 24.04
Linux.

\subsubsection{\bf Datasets}
\cn is evaluated on $10$ real-world graph datasets (Table~\ref{tab:datasets})
obtained from Stanford Large Network Dataset Collection (SNAP) \cite{snapnets}. The datasets are sorted based on their density, where the density of a graph
$G(V,E)$ is defined as $\frac{2|E|}{|V|(|V|-1)}$. We selected these datasets
because many of the prior works (MoSSo \cite{ko2020incremental}, SLUGGER
\cite{lee2022slugger}, SWeG \cite{shin2019sweg}, GraphZip
\cite{rossi2018graphzip}, SSumM \cite{lee2020ssumm} and
\cite{navlakha2008graph}) chose these datasets to benchmark their summarization
schemes.
Additionally, the performance is also evaluated over synthetic datasets
that are based on two widely used graph models, namely, Barab\'asi-Albert (BA)
model \cite{barabasi} and Erd\"os-R\'enyi (ER) model \cite{erdos}.
For the BA model, we use the $G(n,p)$ model where $n$ denotes the number of
nodes and $p$ denotes the number of new edges that gets attached to each
incremental node.  For the ER model, we use the $G(n,m)$ variant that generates
a graph that is chosen uniformly from all possible random graphs of $n$ nodes
and $m$ edges.
The number of nodes vary from
$10^4$ to $10^6$ and the number of edges vary from $10^6$ to $1.8 \times 10^7$ for
the synthetic graphs.

\subsubsection{\bf Lossless Baselines}
For benchmarking of lossless summarization schemes, we consider four
state-of-the-art techniques, namely, SLUGGER
\cite{lee2022slugger}, MoSSo \cite{ko2020incremental},  SWeG \cite{shin2019sweg}, and GraphZip \cite{rossi2018graphzip}.   The choice of these schemes as baselines is justified as follows. SLUGGER yielded up to 29.6\%
smaller summaries than SWeG with
similar execution times. MoSSo and SWeG  claim to achieve
better compression than SAGS \cite{khan2015set}.
GraphZip is shown to offer better compression than the Layered
Label Propagation (LLP) scheme \cite{boldi2011layered}.  
Results for all the baselines are
reported according to the parameter values that offer the best compression, as indicated
in the respective papers.

\subsubsection{\bf Lossy Baselines}
For benchmarking of lossy summarization schemes, 
we chose SSumM \cite{lee2020ssumm} that is shown to offer better summarization than \cite{beg2018scalable, riondato2017graph, lefevre2010grass}. Though there exist several other
competing techniques, as listed in Table~\ref{tab:related}, we could not 
compare with them due to the following reasons: (1)~In spite of multiple
requests, the source codes are not available. (2)~Though the lossless
variant of SweG is available, their lossy variant is not
publicly available as confirmed by the authors themselves. (3)~Few works such as LDME \cite{yong2021efficient} use
Web Graph format to input graph datasets. However, when the chosen SNAP
graph datasets are converted to the Web Graph format, they reported negative compression, i.e., instead of compression, the graphs expanded.

\subsubsection{\bf Evaluation Metrics}
The performance of \cn is
evaluated using \emph{compression ratio} $cr=|G_s|/|G|$, \emph{reconstruction error} $re$, 
\emph{neighborhood query loss}, \emph{error in shortest path}, \emph{reachability query accuracy}, \emph{query response time} and \emph{compression time} $ct$. The reconstruction error is measured using the $L_1$ error norm stated in \cite{lee2020ssumm}:
\begin{align}
	re =   \frac{ |E \setminus E_r| + |E_r \setminus E| }{{|V| \choose 2}}
	\end{align} 
where $V$ is the set of vertices, and $E$ and $E_r$ are the set of edges of the original and reconstructed graphs, respectively.

\subsubsection{\bf \cn Framework Parameters}
The lossy schemes of \cn, namely, \inter and \uni, are
parameterized by the set of \emph{neighborhood loss tolerance thresholds},
denoted by $\nlt=\{\eps_u| u \in V\}$, where  $0< \eps_u < 1$ is a
user-specified parameter that controls the maximum neighborhood loss of node $u$, given
by $rl(u)$, as stated in Eq.~\eqref{eq:varrl}.  To evaluate the impact of $\nlt$, we assume $\eps_u=\eps$ for each
node $u \in V$, where $0 < \eps < 1$ is a parameter specified by the user.
Though the experiments are evaluated on multiple thresholds in the range
$0<\eps<1$, three
representative values $\eps \in \{0.25,0.50,0.75\}$ are chosen.
Note that \eps only represents the \emph{maximum} possible neighborhood loss; the actual \emph{average} neighborhood losses are much smaller, as shown in
Table~\ref{tab:anl}.
 
\subsection{Comparison of Compression Ratio with Existing Baselines}
\label{sec:exp compression ratio}

This section addresses the research question RQ2. Firstly, we compare the compression performance of \cn with lossless baselines, followed by the lossy baselines.

\begin{figure*}[t]
	\includegraphics[width=0.65\linewidth]{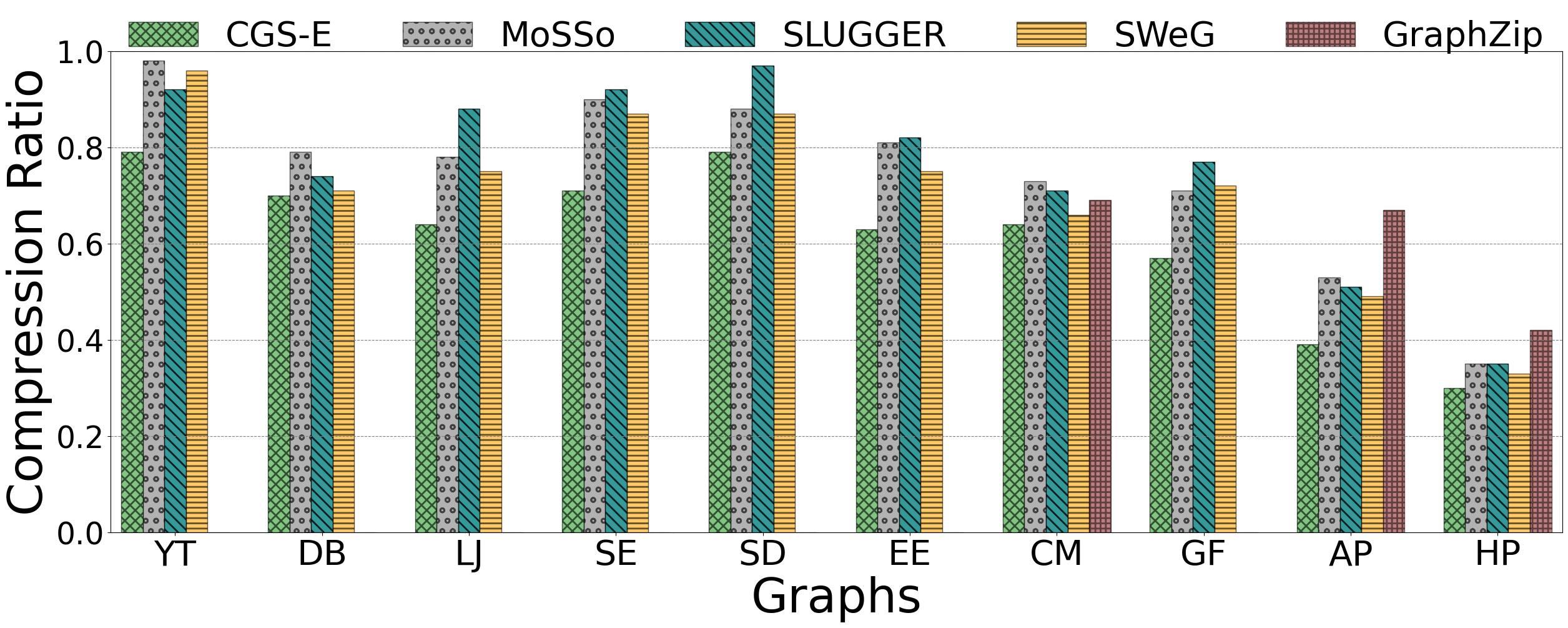}
        \Description[This figure shows the compression ratios $cr$ of various lossless schemes on the real graphs stated in Table~\ref{tab:datasets} where lower $cr$ is preferred. \iu offers the best compression]{\iu offers the best compression than all the baseline schemes on all the real graphs. As graphs become denser, the compression achieved by CGS-E is usually higher.}
	\caption{Compression ratios $cr$ of lossless schemes on real graphs listed in Table~\ref{tab:datasets} (lower $cr$ is better). \iu offers the best compression.}
	\label{fig:cr_lossless_results}
\end{figure*}

\begin{figure*}[t]
\minipage{0.45\textwidth}
	\includegraphics[width=\linewidth]{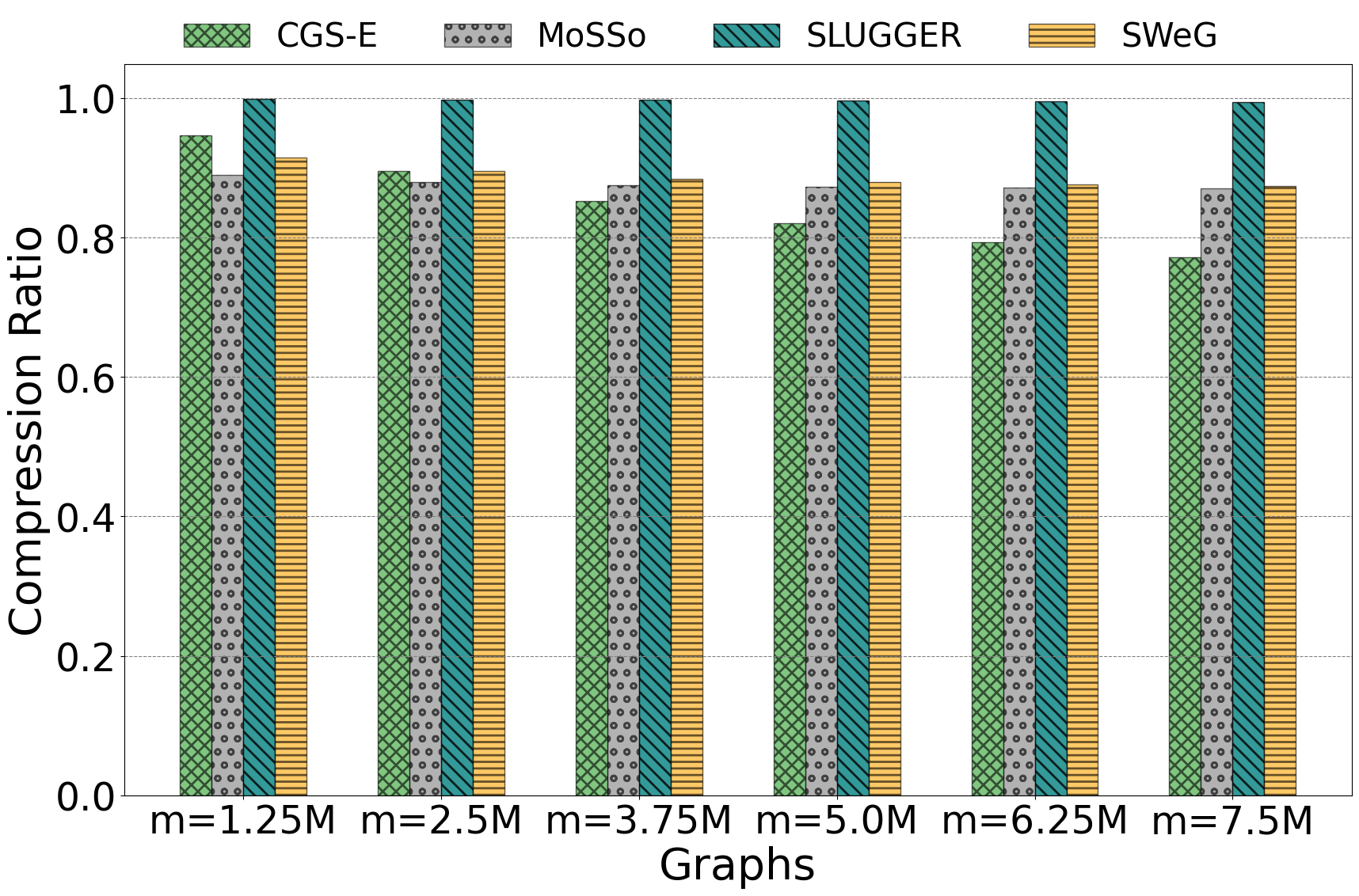}
	\Description[This figure shows the compression ratios $cr$ of various lossless schemes on the synthetic BA graphs with $n = 100,000$ nodes and the number of edges $m$ is shown in million (M) where lower $cr$ is preferred. \iu offers the best compression.]{\iu offers the best compression than all the baseline schemes on all the real graphs. As graphs become denser, the compression achieved by\iu is usually higher.}
	\caption{Compression ratios $cr$ of lossless schemes on synthetic BA graphs with $n = 100,000$ nodes and the number of edges $m$ is shown in million (M) (lower $cr$ is better). \iu offers better compression with increasing density.}
	\label{fig:AB_cr_baselines}
\endminipage
\hfill
	{\ }
\minipage{0.45\textwidth}
	\includegraphics[width=\linewidth]{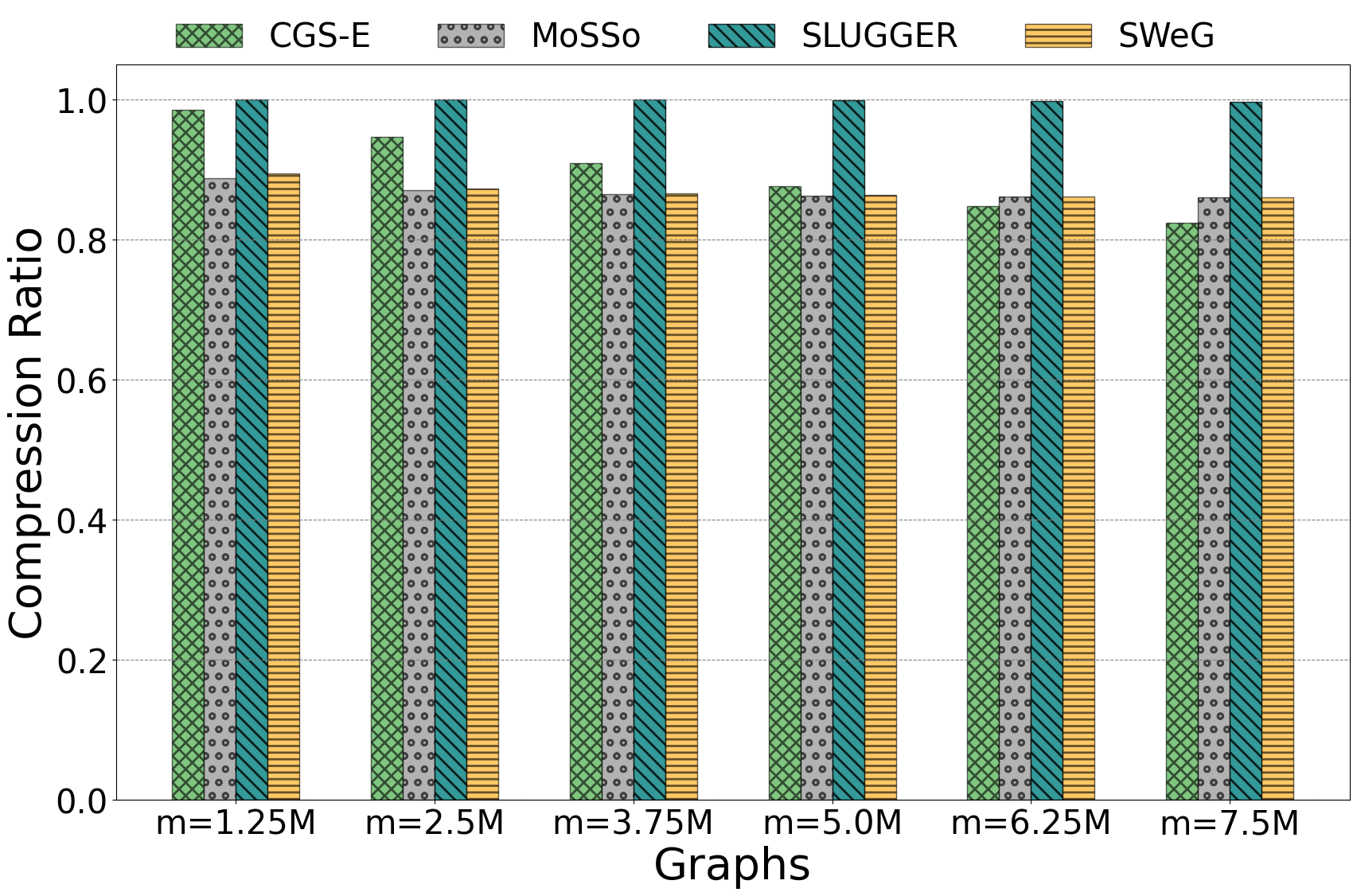}
       \Description[This figure shows the compression ratios $cr$ of various lossless schemes on the synthetic ER graphs with $n = 100,000$ nodes and the number of edges $m$ is shown in million (M) where lower $cr$ is preferred. \iu offers smaller summaries.]{\iu offers the smaller summaries than all the baseline schemes on all the real graphs. As graphs become denser, \iu offers smaller summaries.}
	\caption{Compression ratios $cr$ of lossless schemes on synthetic ER graphs with $n = 100,000$ nodes and the number of edges $m$ is shown in million (M) (lower $cr$ is better). \iu offers smaller summaries with increasing density.}
	\label{fig:RG_cr_baselines}
	\endminipage\hfill
\end{figure*}

\subsubsection{\bf Lossless Schemes}
Fig.~\ref{fig:cr_lossless_results} shows 
the compression ratio obtained by the lossless summarization schemes on the real datasets.
 \iu scheme offers the best compression ratio on all the graphs. More specifically, \iu achieves better compression than: SLUGGER by up to $27\%$  (on LJ dataset), MoSSo by up to $26\%$ (on AP dataset), SWeG by up to $20\%$ (on GF dataset) and GraphZip by up to  $41\%$ (on AP dataset).
Since the source code of GraphZip is not available\footnote{Despite multiple requests to the authors, we could not get their source codes.}, we show the results for only those datasets whose results are reported in \cite{rossi2018graphzip}.

In the figure, the graphs are ranked based on their density (lower to higher). We observe that as the graphs become denser, the compression achieved by \iu is usually higher. This is because as the density increases, the probability that any pair of nodes $(u,v)$ in a graph (where $v$ is a 2-hop neighbor of $u$) has high number of common neighbors, increases. This results in higher compression.

Fig.~\ref{fig:AB_cr_baselines} and Fig.~\ref{fig:RG_cr_baselines} show the compression results for the lossless schemes on synthetically generated BA graphs and ER graphs, respectively. For these experiments, the number of nodes is fixed to $10^5$, and the density is varied from $2.50 \times 10^{-4}$ to $1.50 \times 10^{-3}$. As in the case of real world SNAP data sets shown earlier, the \iu scheme is a clear winner for the synthetic datasets, as well. Notably, in both scenarios, ie., real and synthetic datasets, we observe that as the density of the graph increases, \iu generally yields better compression.

For some datasets such as GF and AP, the compression ratio of SSumM exceeds $1$. This is because $cr=|G_s|/|G|$, and as stated in Sec.~\ref{sec:formulation}, $|G_s|$ takes into account the space overhead of storing the super-node mappings in addition to the nodes and edges of the summary graph.  

Comparing the compression performance of \iu and \inter on the real world SNAP graphs, it appears that  \iu offers better compression on majority of the datasets. However, it must be accounted that this comparison is for the chosen value of the reconstruction error $re$ (or corresponding value of the neighborhood loss tolerance threshold \eps) that is chosen for fair comparison with SSumM. In fact, we shall shortly see that with higher values of \eps, \inter outperforms \iu (Sec.~\ref{sec:exp effect of neighborhood loss tolerance threshold}). Moreover, even for the chosen value of $re$, \inter offers better compression than \iu on CM and DB graphs. 

\subsubsection{\bf Lossy Schemes}
\label{sec:exp compression ratio lossy}
Fig.~\ref{fig:cr_lossy_results} shows the compression ratios obtained by 
the lossy summarization schemes on different real graph datasets.  Since the reconstruction error $re$ is an output parameter for SSumM, we cannot directly control it. Therefore, for uniform comparison, $re$ is chosen close to $10^{-4}$. 
Except for SE and YT datasets, \inter offers better or similar compression than SSumM on all the $8$ other datasets. In particular, \inter achieves up to $46\%$ smaller summaries than SSumM, as reported for HP dataset. The two graphs, YT and SE, where SSumM outperforms \iu, are relatively sparse. While YT is the least dense graph, SE is the fourth lowest dense graph. Therefore, based on the above observation, we can conclude that \inter outperforms SSumM on the relatively denser graphs. Among \inter and \uni, \inter invariably offers better compression on all the considered real graphs.
 
Fig.~\ref{fig:AB_cr_baselines_lossy} and Fig.~\ref{fig:RG_cr_baselines_lossy}  show the compression results for the lossy schemes on synthetically generated BA graphs and ER graphs, respectively. For experiments, we fix the number of nodes to $100,000$.
For fair comparison, both \cn and SSumM are evaluated at the similar reconstruction error $re$.
\uni outperforms \inter and SSumM for both BA and ER graphs. This is in contrast to the behaviour seen on the real graphs earlier, where \inter outperformed \uni on all the datasets. Although \inter provides reasonable compression for BA graphs, it yields no significant compression for ER graphs. This is because, in case of a merge operation for a given pair of nodes executed by the \inter, the exclusive neighors are lost. In ER graphs, the neighborhoods are uniformly distributed; hence, for a given pair of nodes $(u,v)$ where $v$ is a 2-hop neighbor of $u$, the number of exclusive neighbors is sufficiently high that prohibits meeting the neighborhood loss constraint. For both BA and ER graphs at lower densities, \inter achieve better compression than the SSumM baseline. As density increases, SSumM outperforms \inter. From the above we observe that there is no clear winner among \inter and \uni. Nevertheless, one of the lossy variants of \cn usually offers better compression than SSumM. The superior performance of \cn over the compared baselines is largely owing to its summarization approach that carefully examines the compression gains of all eligible pairs of nodes/supernodes that can be merged, and the greedy strategy to iteratively choose the pair that offers the maximum compression gain.

\begin{figure*}[t]
\minipage{0.45\textwidth}
	\includegraphics[width=\linewidth]{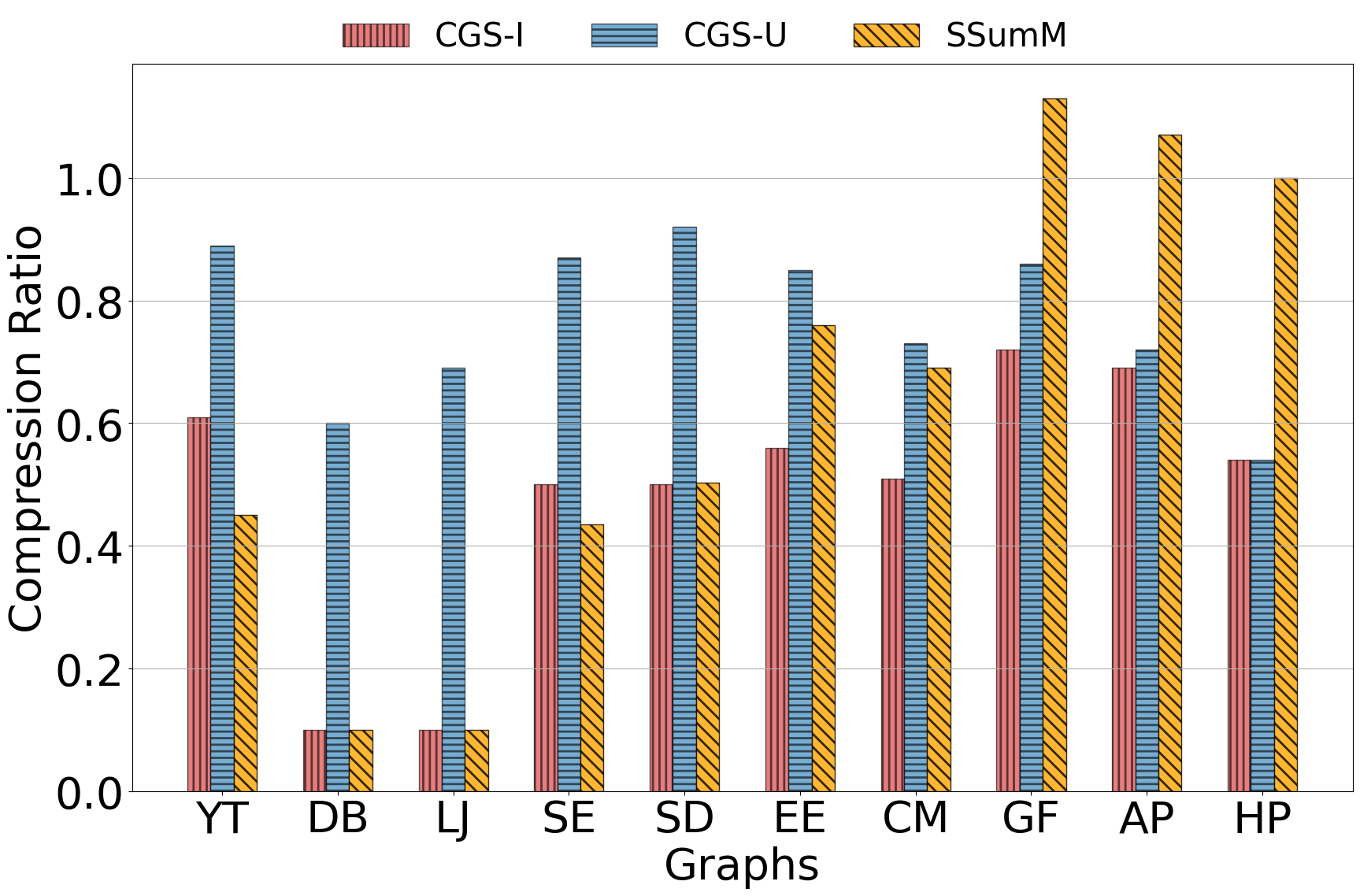}
	\cutfigabove
        \Description[The figure shows the compression ratios $cr$ of various lossy schemes on the real graphs mentioned in Table~\ref{tab:datasets} where lower $cr$ is preferred. \inter outperforms SSumM]{\inter outperperforms SSumM on all the real graphs at $re \sim 10^{-4}$. As graphs become denser, \inter performs better than SSumM.}
	\caption{Compression ratios $cr$ of lossy schemes on real graphs listed in Table~\ref{tab:datasets} at $re \sim 10^{-4}$ (lower $cr$ is better). \inter outperforms SSumM on the relatively denser graphs.}
	\label{fig:cr_lossy_results}
\endminipage
\hfill
	\minipage{0.45\textwidth}
	\includegraphics[width=\linewidth]{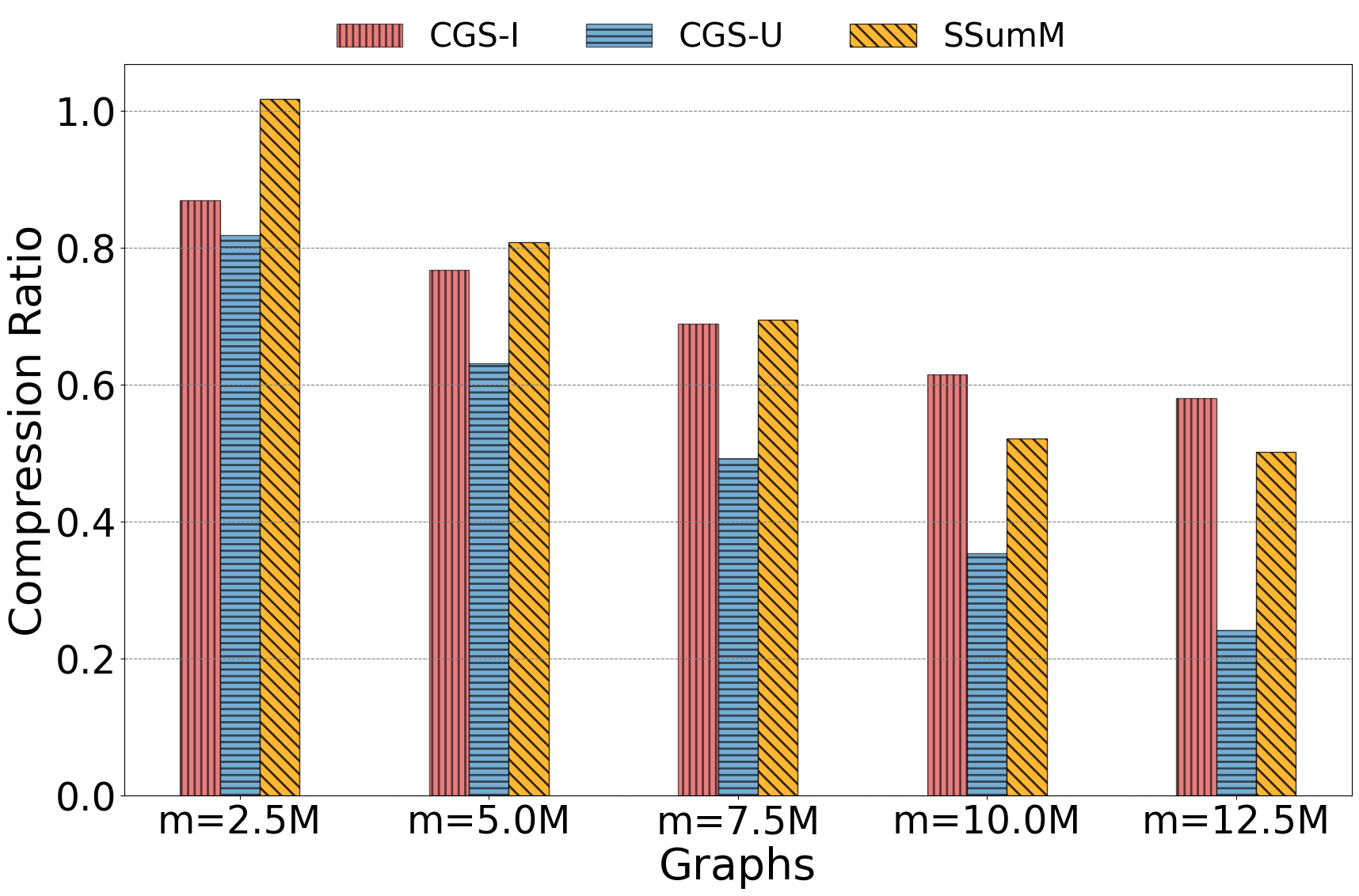}
      \Description[This figure shows the compression ratios $cr$ of various lossy schemes on the synthetic BA graphs with $n = 100,000$ nodes and the number of edges $m$ is shown in million (M) where lower $cr$ is preferred. \uni offers smaller summaries.]{\uni offers the smaller summaries than SSumM and \inter on the BA graphs. As number of edges increases, \uni offers smaller summaries.}
	\caption{Compression ratios $cr$ of lossy schemes on synthetic BA graphs with $n=100,000$
 nodes and the number of edges $m$ is shown in million (M) (lower $cr$ is better). Overall, \uni offers smaller summaries.}
	\label{fig:AB_cr_baselines_lossy}
\endminipage\hfill
\end{figure*}

\begin{figure*}[t]
\minipage{0.45\textwidth}
	\includegraphics[width=\linewidth]{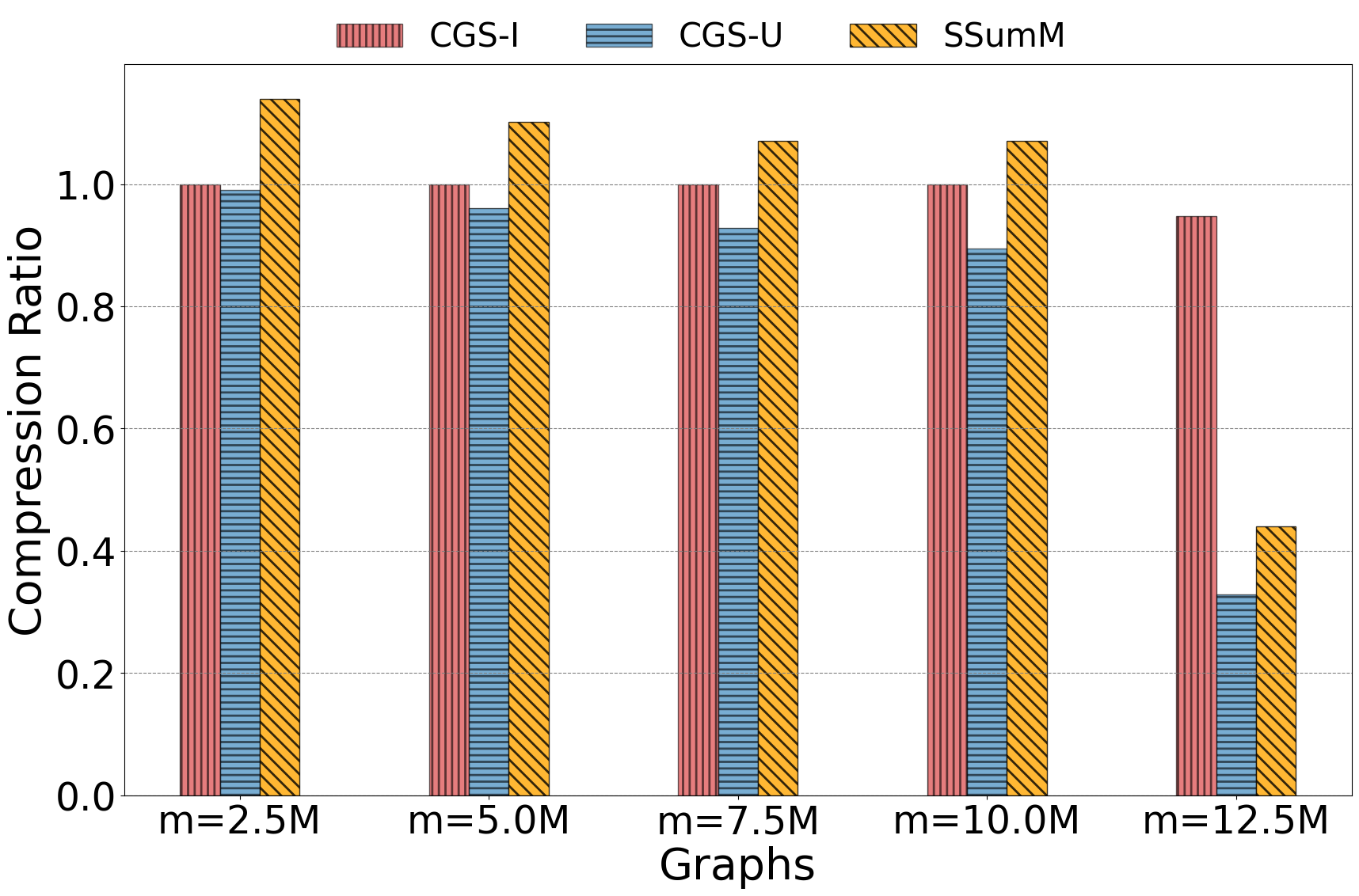}
         \Description[This figure shows the compression ratios $cr$ of various lossy schemes on the synthetic ER graphs with $n = 100,000$ nodes and the number of edges $m$ is shown in million (M) where lower $cr$ is preferred. \uni offers better compression.]{\uni offers better compression than SSumM and \inter on the ER graphs. As number of edges increases, \uni offers better compression.}
	\caption{Compression ratios $cr$ of lossy schemes on synthetic ER graphs with $n=100,000$ nodes and the number of edges $m$ is shown in million (M) (lower $cr$ is better). Overall, \uni offers better compression.}
	\label{fig:RG_cr_baselines_lossy}
	\endminipage
\hfill
\minipage{0.45\textwidth}
	\includegraphics[width=\linewidth]{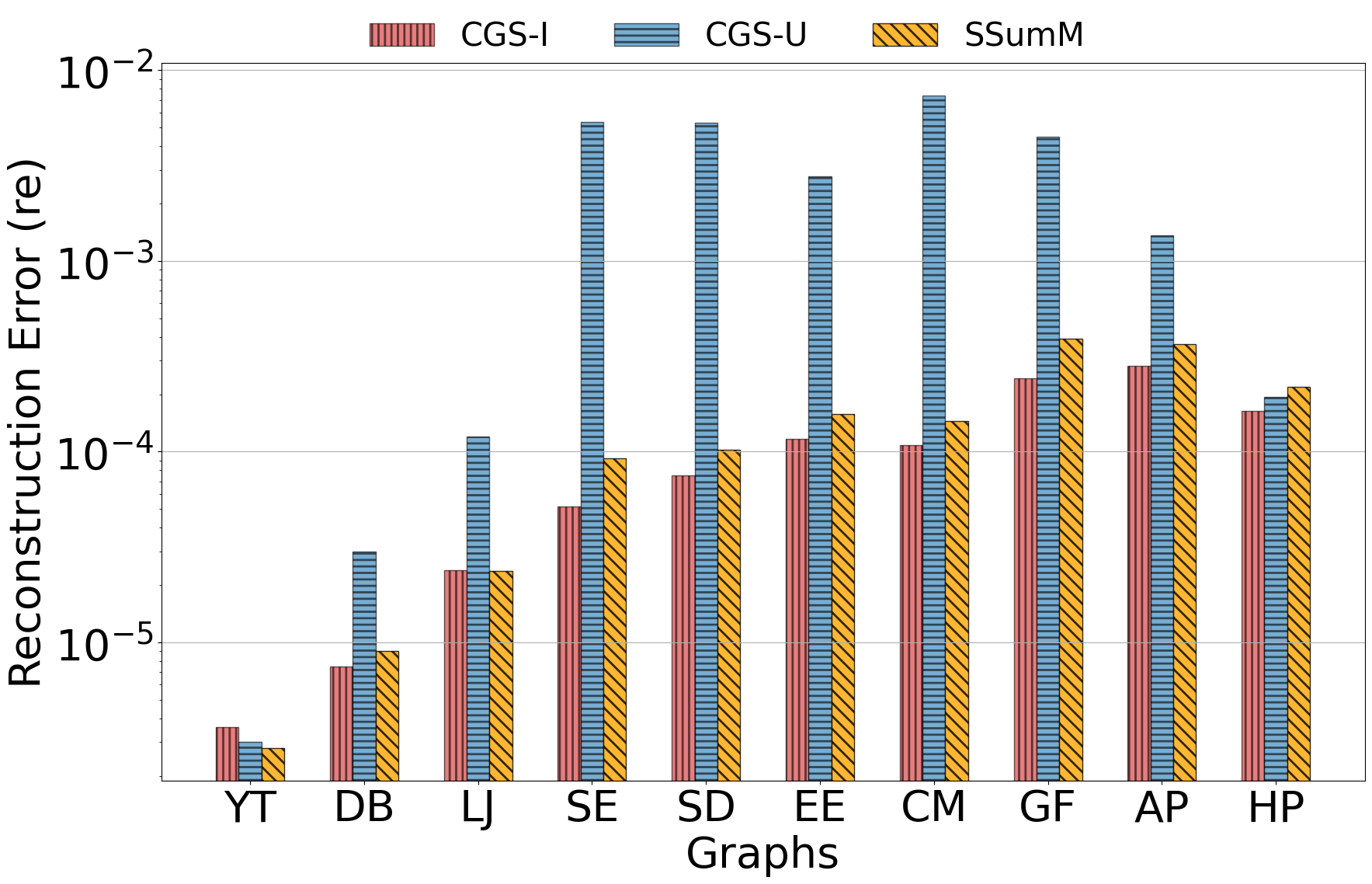}
       \Description[This figure shows the reconstruction error $re$ of various lossy schemes on all the real graphs stated in Table~\ref{tab:datasets}. where lower $re$ is better. \inter offers the lowest errors.]{$re$ results are evalauted for $cr$ close to 0.5. \inter provides the lowest reconstruction errors than SSumM and \uni on all the real graphs.}
	\caption{Reconstruction error $re$ of lossy schemes at $cr \sim 0.5$
	(lower $re$ is better): Overall, \inter offers lowest reconstruction errors.}
	\label{fig:re_result}
	\endminipage\hfill
\end{figure*}

\begin{figure*}[t]
	\subfloat[AP dataset]
	{
	\includegraphics[width=0.45\linewidth]{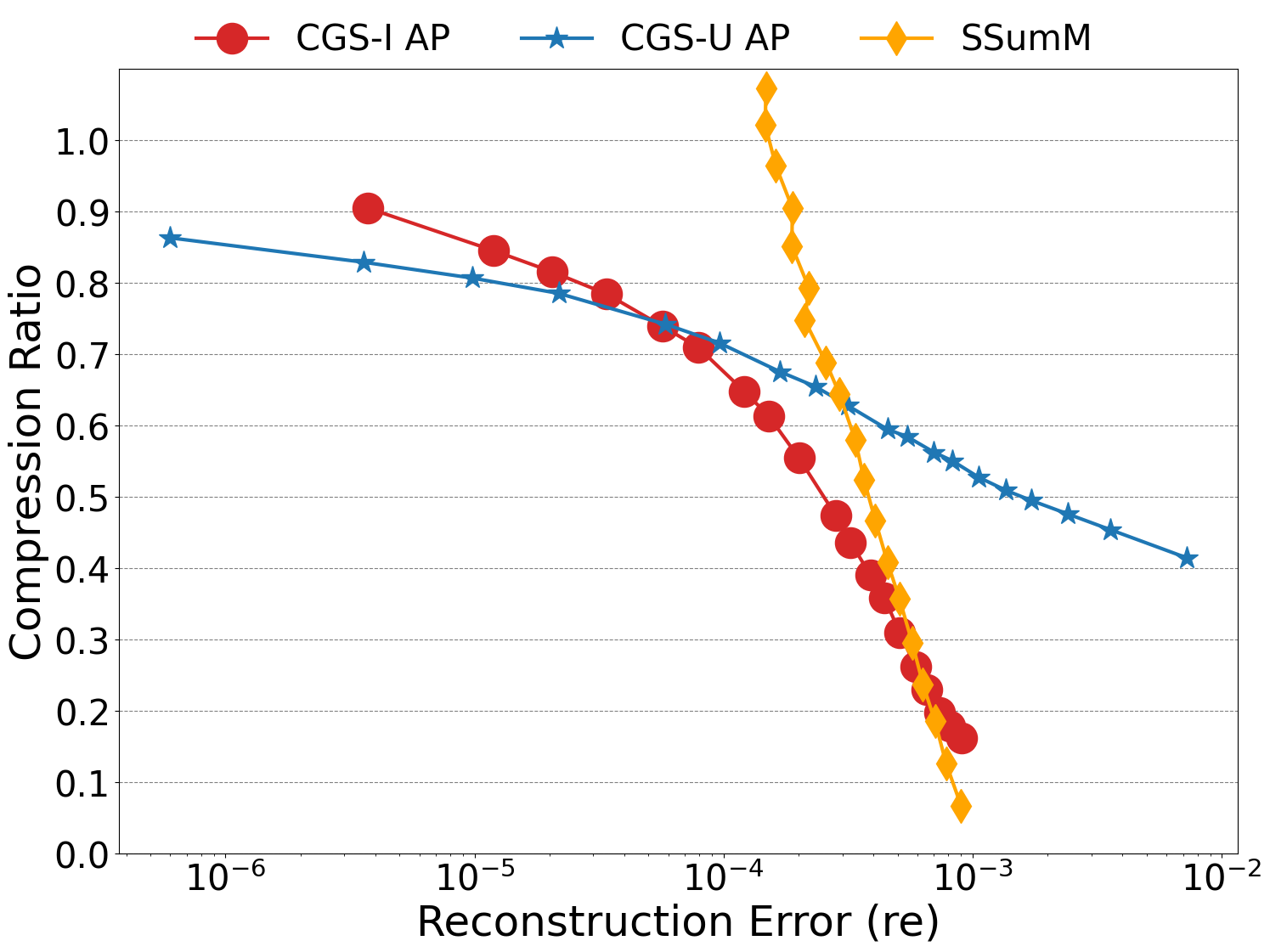}
	\Description[This figure shows the performance of Compression ratio $cr$ vs. reconstruction error $re$ for the AP dataset. \inter shows good $cr$.]{\inter offers the best compression than SSumM and \uni on the AP dataset. As $re$ increases, compression ratio decreases.}
	\label{fig:cr_re_AP}
	}
	\subfloat[GF dataset]
	{
	\includegraphics[width=0.45\linewidth]{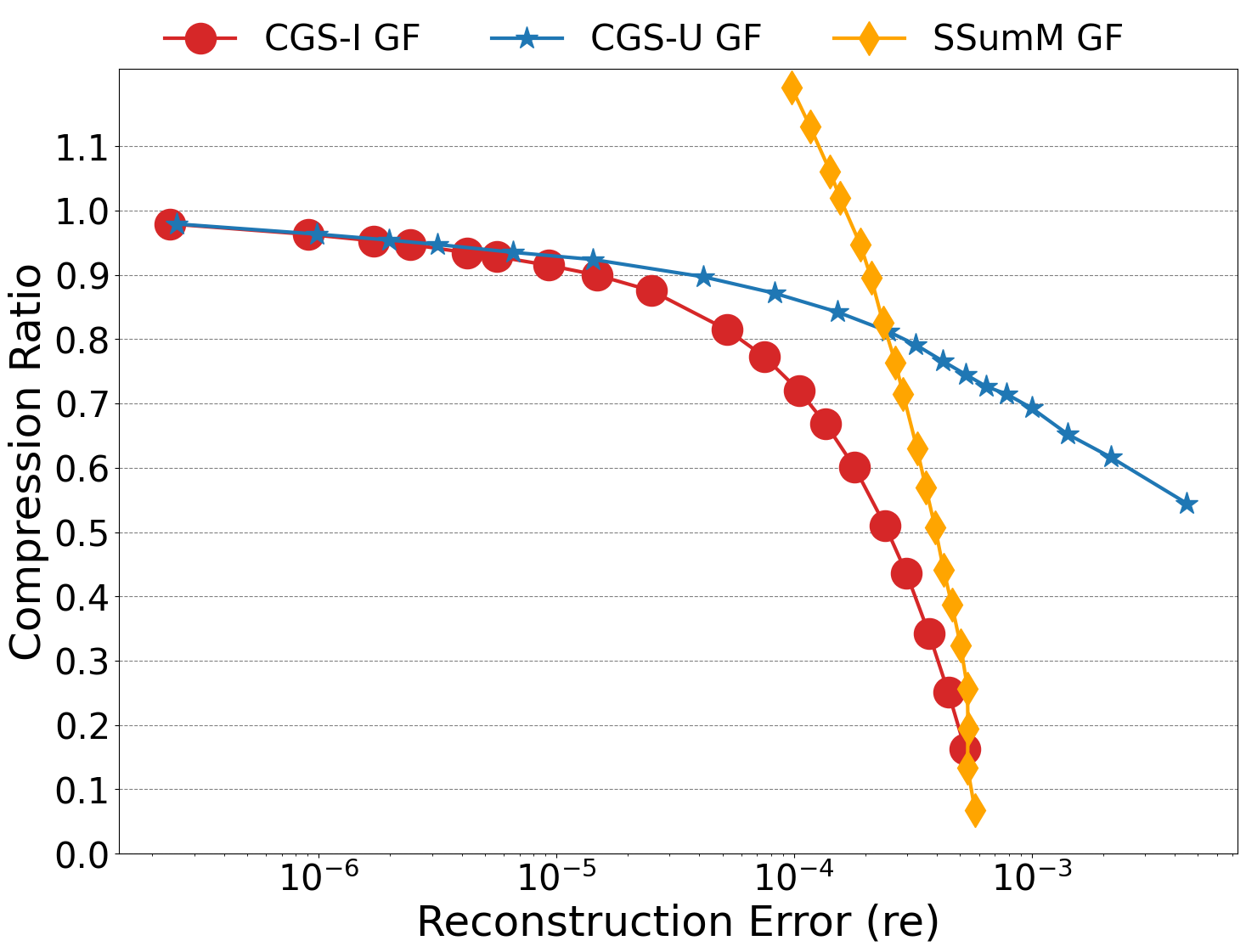}
       \Description[This figure shows the performance of Compression ratio $cr$ vs. reconstruction error $re$ for the GF dataset. \inter shows good $cr$.]{\inter offers the best compression than SSumM and \uni on the GF dataset. As $re$ increases, compression ratio decreases.}
	\label{fig:cr_re_GF}
	}
	\caption{Compression ratio $cr$ vs. reconstruction error $re$: At practical low reconstruction error rates, \inter shows good compression ratio.}
	\label{fig:cr vs re}
	\end{figure*}

\subsection{Comparison of Reconstruction Error with Existing Baselines}
\label{sec:exp reconstruction error}

We next consider the research question RQ3.
Fig.~\ref{fig:re_result} shows the reconstruction error $re$ obtained by 
the lossy summarization schemes.  Since $cr$ is an output parameter, we cannot directly control it. Hence, the reconstruction error results are evaluated for $cr$ close to $0.5$. \inter achieves lower or similar reconstruction errors than SSumM on all datasets.
The reconstruction error of \uni is, however, worse than SSumM on all the datasets other than HP.

Fig.~\ref{fig:cr vs re} shows the variation of compression ratio $cr$ with respect to the reconstruction error $re$ for the lossy schemes on two representative datasets, AP and GF.
The compression ratio $cr$ monotonically decreases (becomes better) as the reconstruction error $re$ increases, since more error allows more compression. For lower reconstruction errors, i.e., $10^{-6} \le re \le 10^{-4}$, both \inter and \uni offer considerably better compression ratios than SSumM. This gap reduces as $re$ exceeds $10^{-4}$. The results for other datasets are similar. 

\begin{figure*}[t]
	\subfloat[Compression ratio]
	{
		\includegraphics[width=0.45\linewidth]{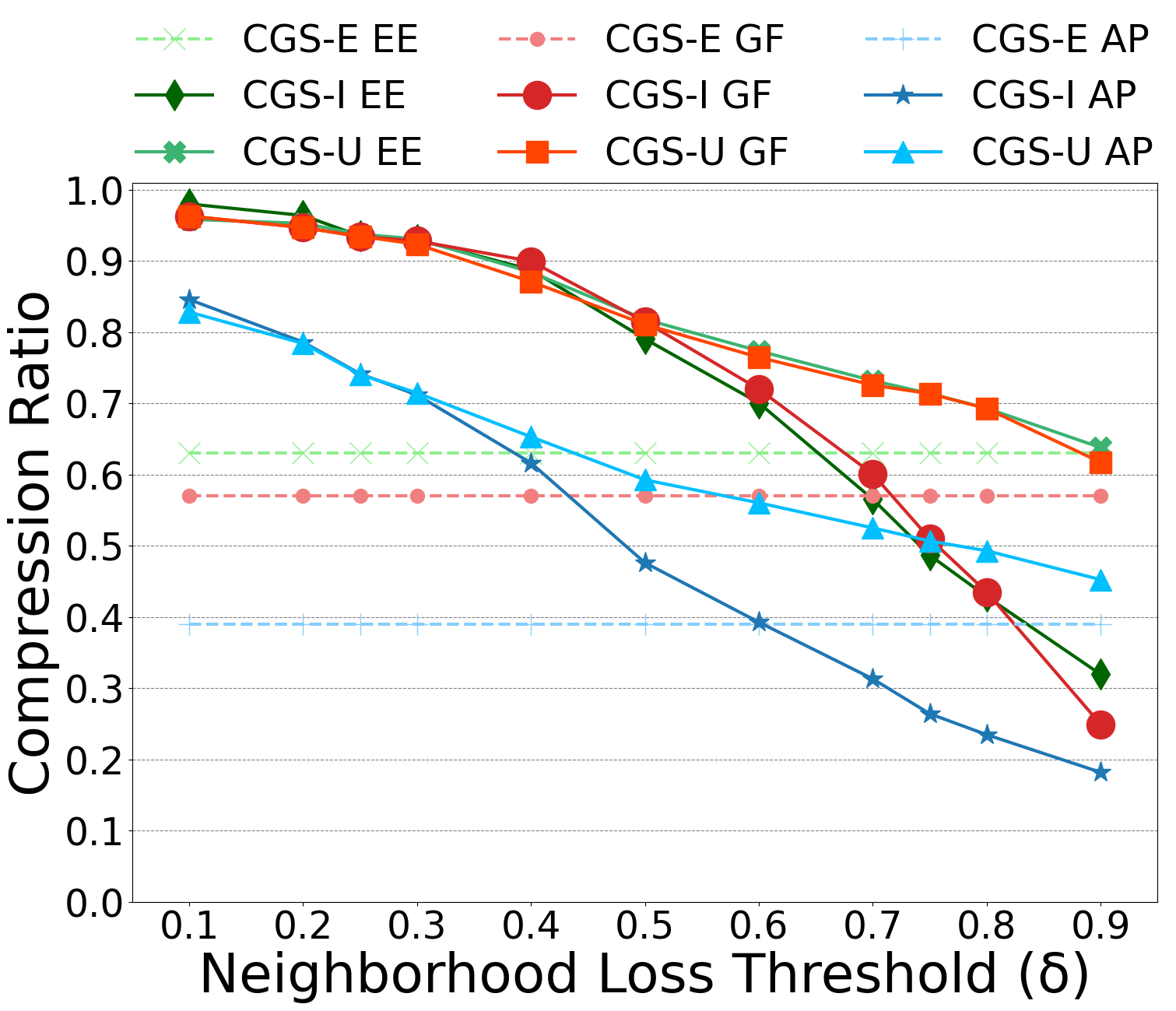}
		\Description[This figure demonstrates the effect of neighborhood loss threshold $\eps$ on compression ratio on the AP, GF and EE graphs. \iu offers smaller summaries.]{\iu shows smaller summaries than \uni and \inter on the AP, GF and EE graphs. As \eps increases, the compression ratios of both \inter and \uni show a steady decline. For lower $\eps$, \iu offer smaller summaries.}
		\label{fig:eps_cr}
	}
	\subfloat[Compression time]
	{
		\includegraphics[width=0.45\linewidth]{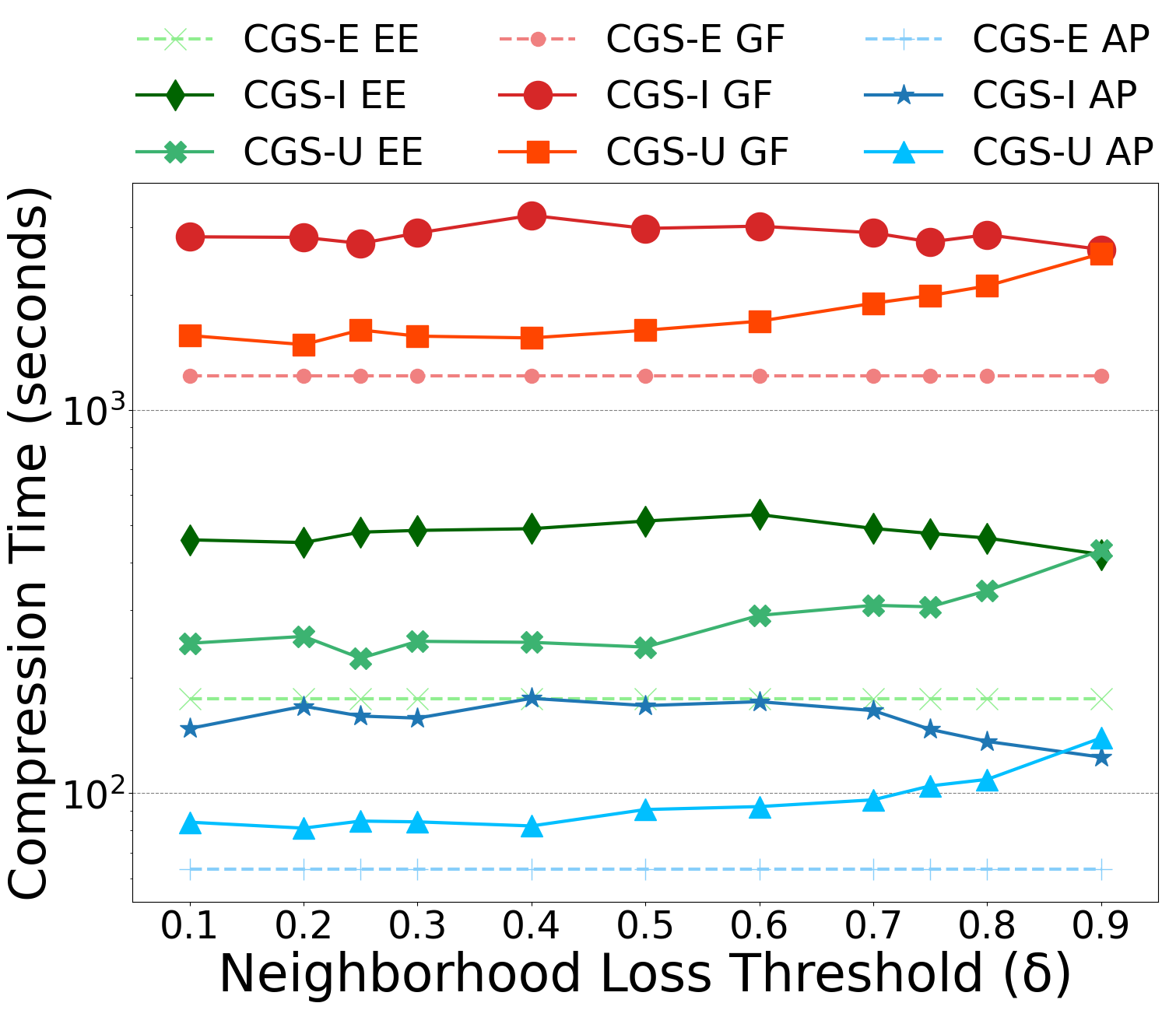}
             \Description[This figure demonstrates the effect of neighborhood loss threshold $\eps$ on compression time (seconds) on the AP, GF and EE graphs. \iu shows least compression time.]{As $\eps$ increases, \iu shows least compression time than \uni and \inter on the AP, GF and EE graphs.}
		\label{fig:eps_ct}
	}
	\caption{Effect of neighborhood loss threshold $\eps$ on compression
	ratio and compression time on real graphs. For lower loss thresholds, \iu offers smaller summaries. \iu takes least compression time. }
	\label{fig:effect of eps}
	\end{figure*}

\subsection{Effect of Neighborhood Loss Tolerance Threshold $\eps$} \label{sec:exp effect of neighborhood loss tolerance threshold}

Next, the research question RQ4 is probed.
Fig.~\ref{fig:effect of eps} shows the effect of neighborhood loss tolerance threshold $\eps$ for the lossy variants of \cn, namely \inter and \uni, on 3 representative datasets, EE, GF and AP. These datasets are chosen based on their average degrees, as stated in Table~\ref{tab:datasets}. As \eps increases, the compression ratios of both \inter and \uni show a steady decline. This is expected, since as \eps increases, it allows more compression at the cost of higher neighborhood loss. The results on other datasets are similar.

For comparison with \iu, the compression ratio of \iu is also shown. For lower values of \eps ($0 \leq \eps \leq 0.6$), \iu outperforms the lossy schemes. However, for higher values of \eps ($\eps \geq 0.75$), \inter outperforms \iu. Fig.~\ref{fig:eps_ct} shows that \iu has lower compression times than both the lossy variants of \cn. 
Refering to ~\ref{fig:eps_cr}, \inter gives more compression than \uni on the chosen datasets. This behavior is in line with that seen in Fig.~\ref{fig:cr_lossy_results}.
 Fig.~\ref{fig:eps_ct}   shows that there is no significant effect of \eps on the running time of either \inter or \uni.

\begin{figure*}[t]
	\subfloat[BA graphs]
	{
		\includegraphics[width=0.45\linewidth]{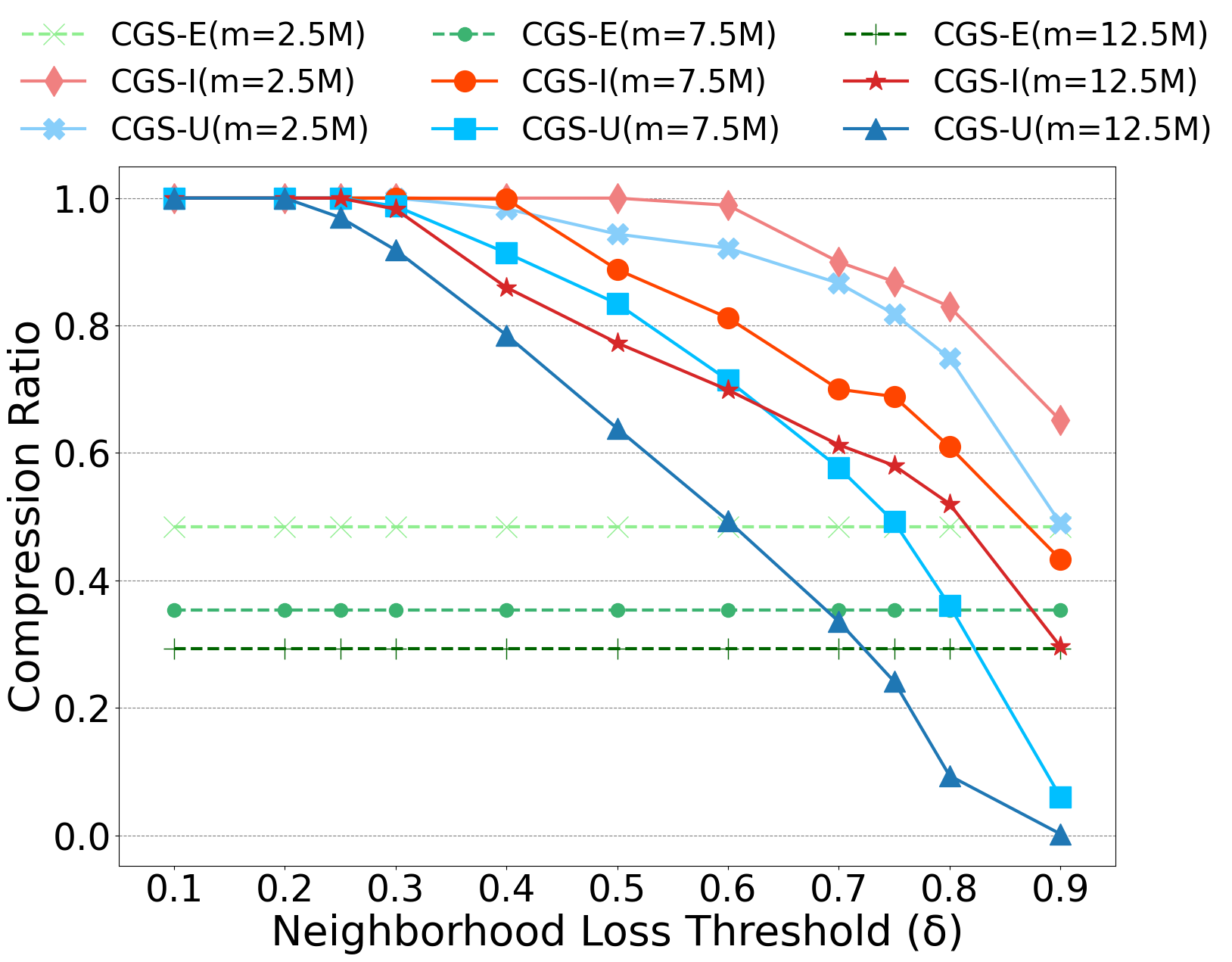}
		\Description[This figure demonstrates the effect of neighborhood loss threshold $\eps$ on compression ratio on the synthetic BA graphs with $n=10,000$ nodes and the number of edges. \uni offers better compression.]{\uni offers better compression than \inter. As $\eps$ increases, the compression ratio decreases, since larger $\eps$ values allow greater compression. For BA graph, denser graphs undergo more compression, leading to a lower compression ratio.}
		\label{fig:eps_cr_AB}
	}
	\subfloat[ER graphs]
	{
		\includegraphics[width=0.45\linewidth]{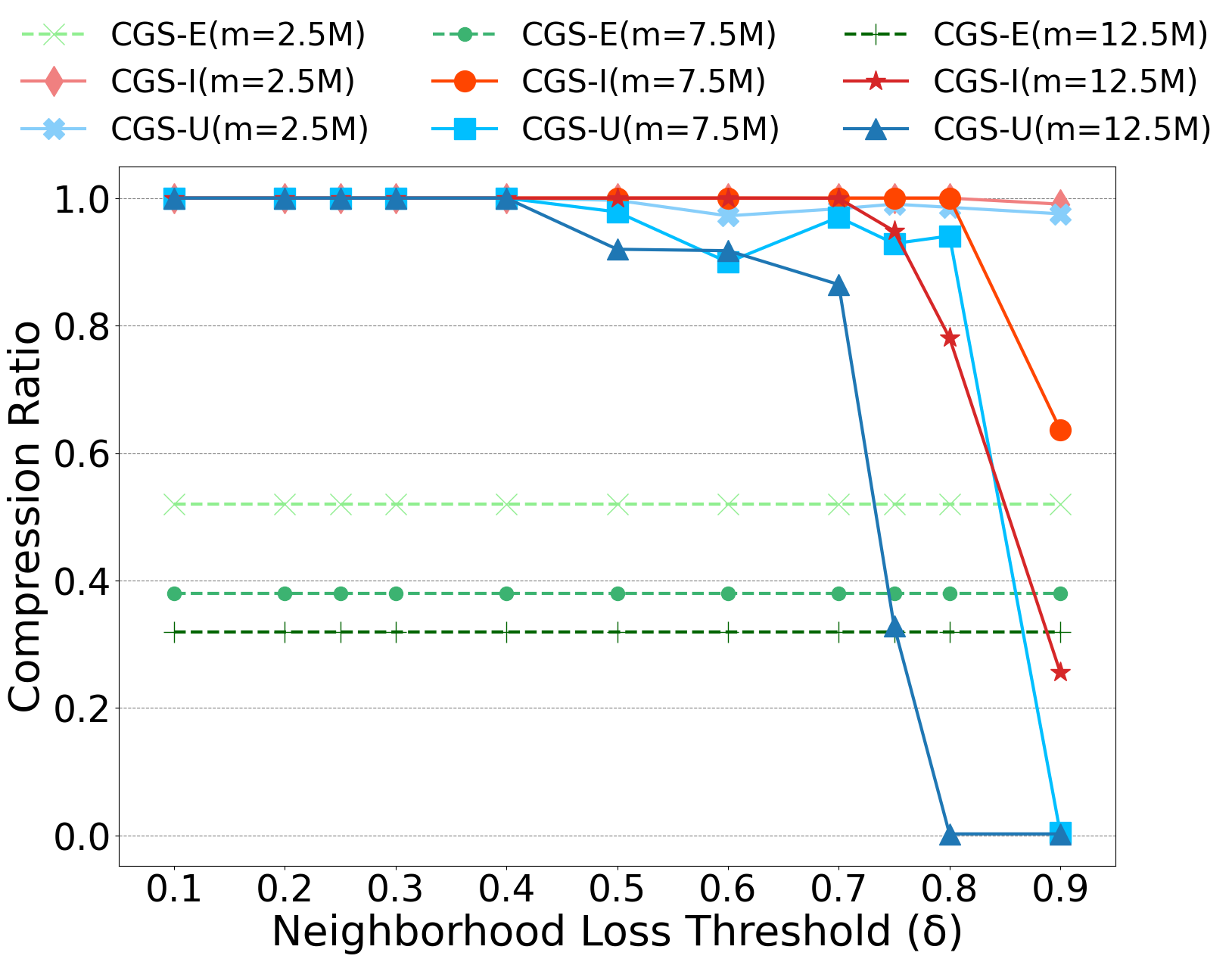}
             \Description[This figure demonstrates the effect of neighborhood loss threshold $\eps$ on compression ratio on the synthetic ER graphs with $n=10,000$ nodes and the number of edges. \uni offers better compression.]{\uni offers better compression than \inter. As $\eps$ increases, the compression ratio decreases, since larger $\eps$ values allow greater compression. For ER graph, significant compression occurs only beyond a higher $\eps$ threshold and primarily for dense graphs.}
		\label{fig:eps_cr_ER}
	}
	\caption{Effect of neighborhood loss threshold $\eps$ on compression ratio for the synthetic graphs with $n=10,000$ nodes and the number of edges as shown. Higher compression is achieved at higher loss thresholds.}
	\label{fig:effect of eps synthetic graphs}  
\end{figure*}

\begin{figure*}[t]
	\subfloat[BA graphs]
	{
		\includegraphics[width=0.45\linewidth]{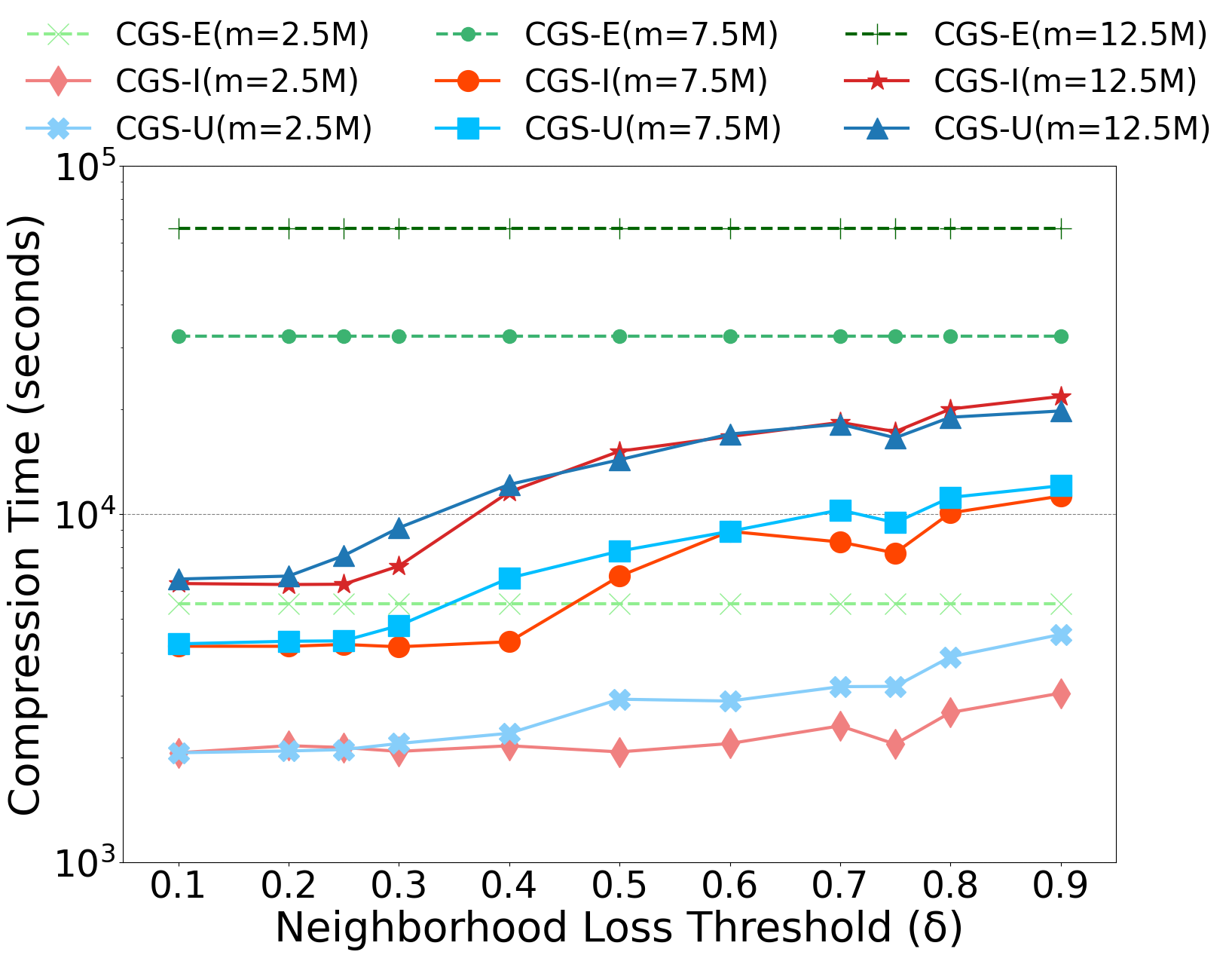}
		\Description[This figure demonstrates the effect of neighborhood loss threshold $\eps$ on compression time(seconds) on the synthetic BA graphs with $n=10,000$ nodes and the number of edges. Compression time increases with density.]{For all the three \cn variants, the compression time increases as graphs become larger and denser, due to the greater number of eligible node pairs with positive compression gain.}
		\label{fig:eps_ct_AB}
	}
	\subfloat[ER graphs]
	{
		\includegraphics[width=0.45\linewidth]{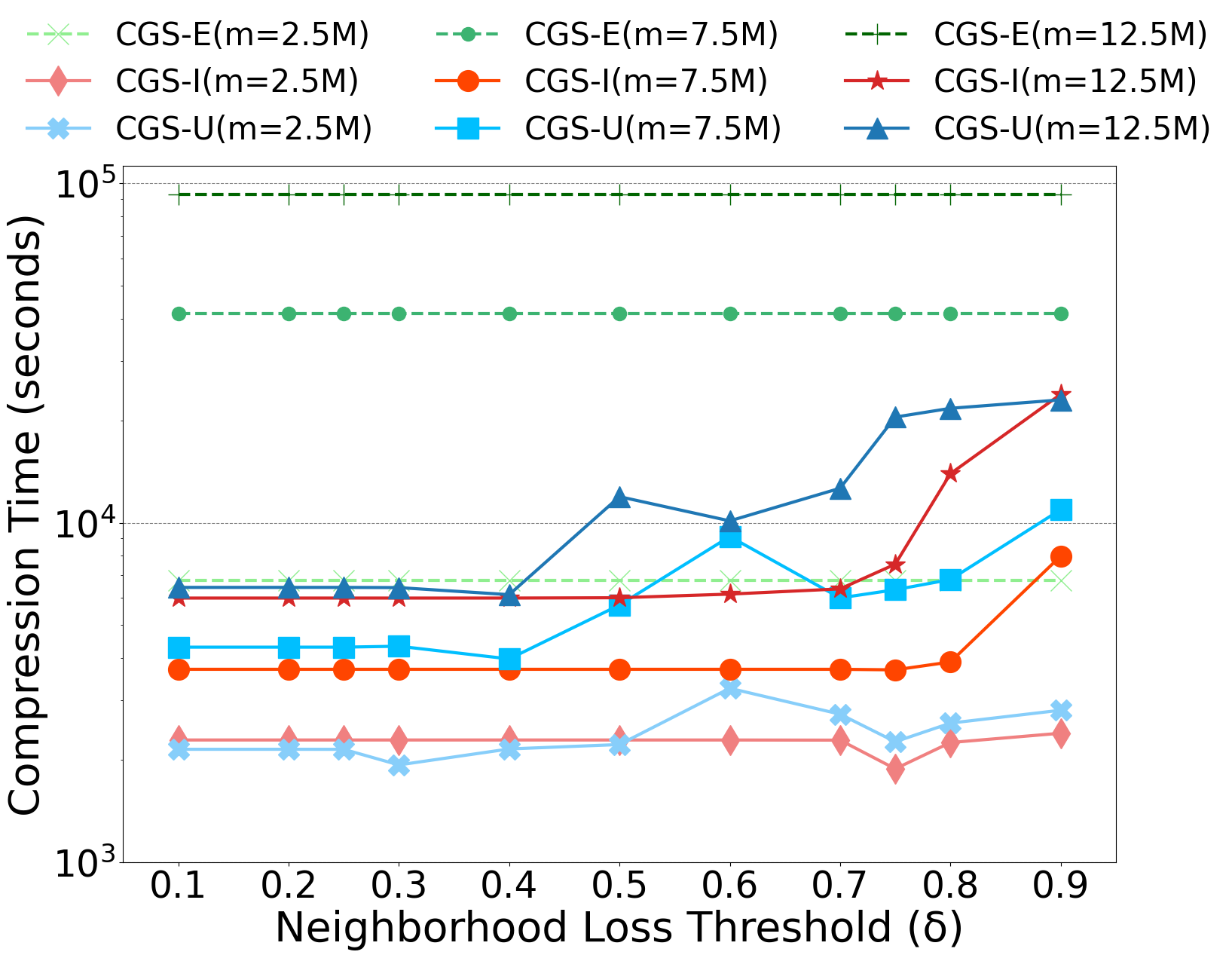}
             \Description[This figure demonstrates the effect of neighborhood loss threshold $\eps$ on compression time(seconds) on the synthetic ER graphs with $n=10,000$ nodes and the number of edges. Compression time increases with density.]{For all the three \cn variants, the compression time increases as graphs become larger and denser, due to the greater number of eligible node pairs with positive compression gain.}
		\label{fig:eps_ct_ER}
	}
	\caption{Effect of neighborhood loss threshold $\eps$ on compression time for the synthetic graphs with $n=10,000$ nodes and the number of edges as shown. Compression time increases with density.}
	\label{fig:effect of eps on ct using synthetic graphs}
\end{figure*}

\begin{figure*}[t]
	\subfloat[BA graphs]
	{
		\includegraphics[width=0.45\linewidth]{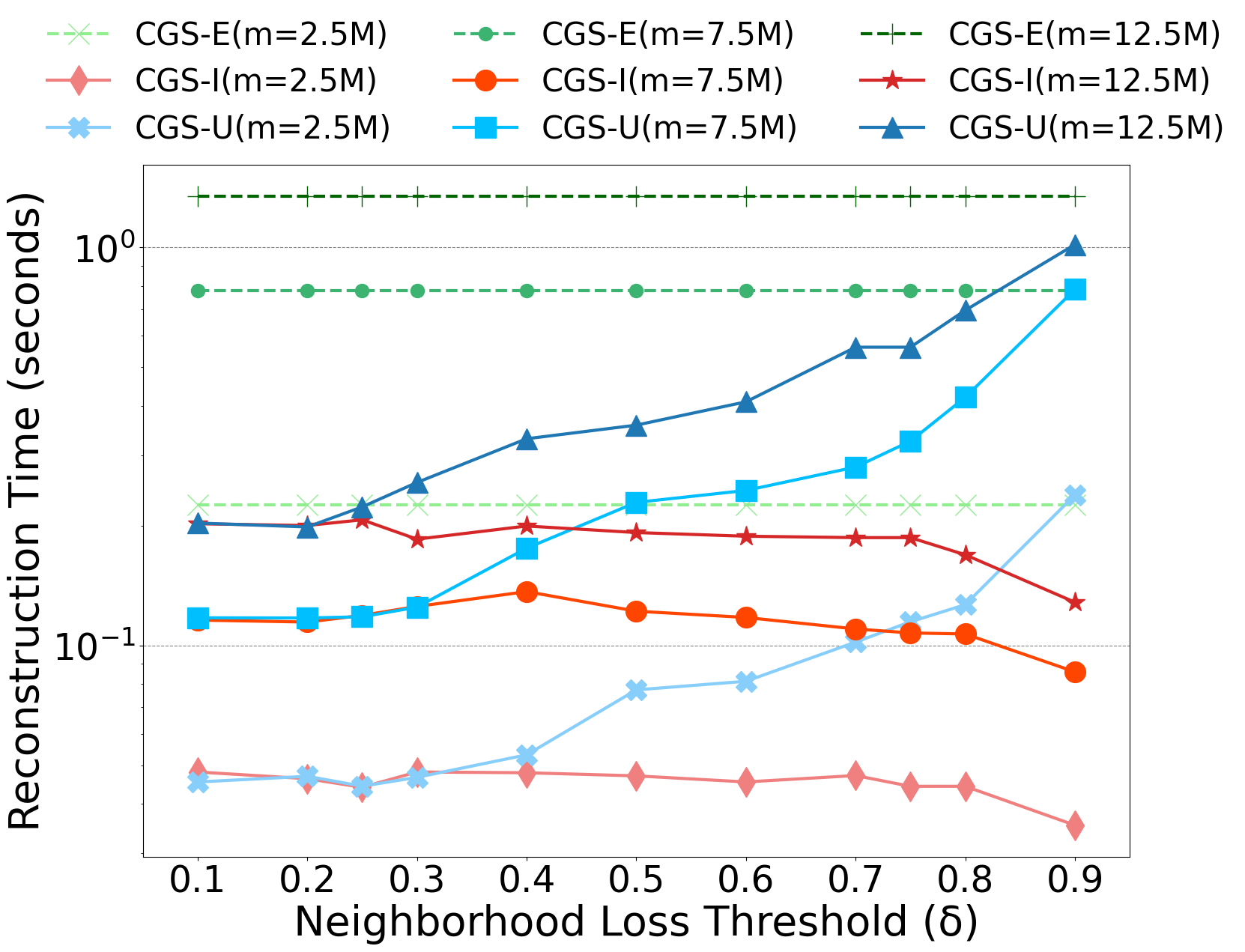}
		 \Description[This figure demonstrates the effect of neighborhood loss threshold $\eps$ on reconstruction time(seconds) on the synthetic BA graphs with $n=10,000$ nodes and the number of edges. \inter has the lowest reconstruction time.]{The reconstruction time increases as the graphs become larger and denser. Among the three \cn variants, \inter has the lowest reconstruction time.}
		\label{fig:eps_rt_AB}
	}
	\subfloat[ER graphs]
	{
		\includegraphics[width=0.45\linewidth]{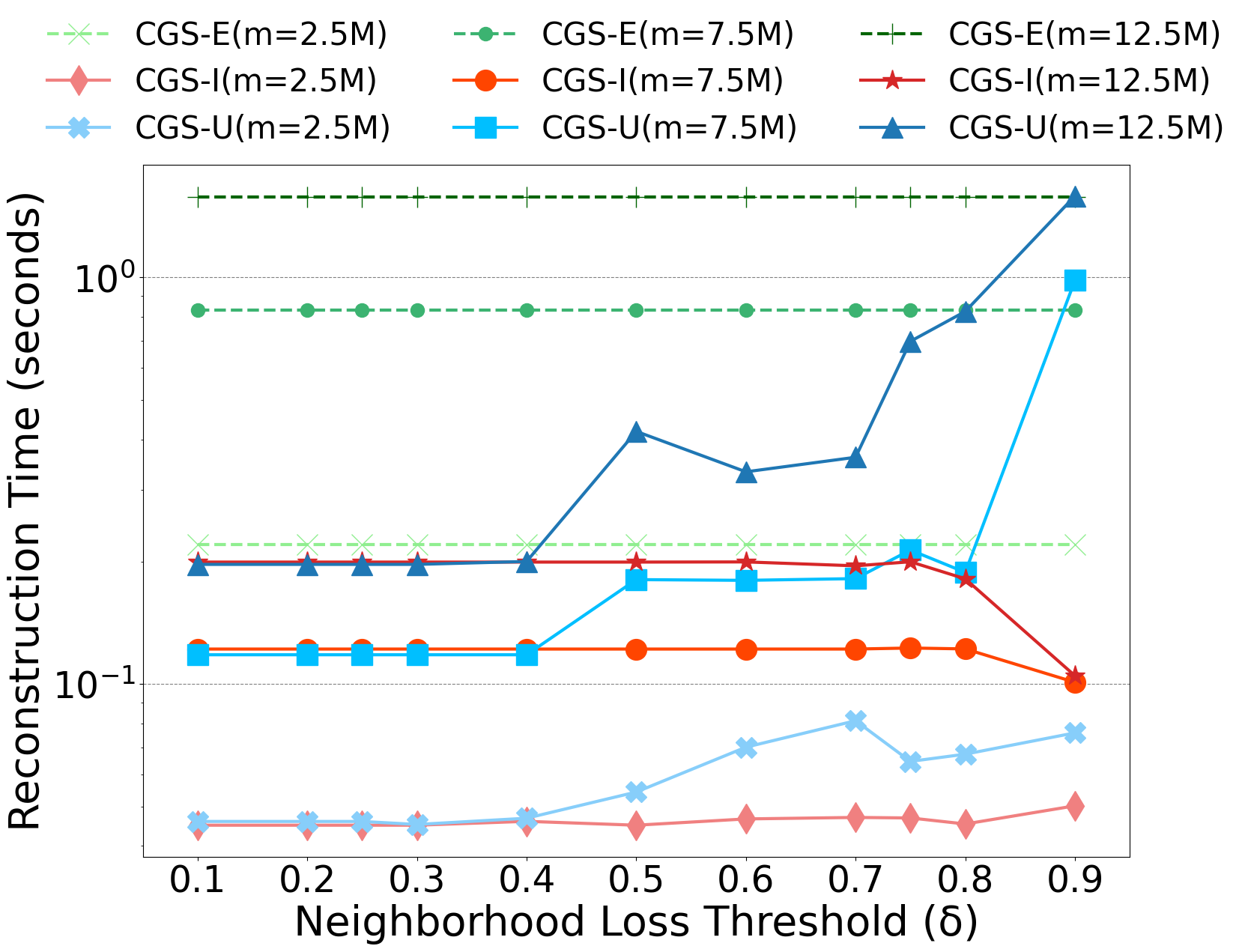}
		 \Description[This figure demonstrates the effect of neighborhood loss threshold $\eps$ on reconstruction time(seconds) on the synthetic ER graphs with $n=10,000$ nodes and the number of edges. \inter has the lowest reconstruction time.]{The reconstruction time increases as the graphs become larger and denser. Among the three \cn variants, \inter has the lowest reconstruction time. }
		\label{fig:eps_rt_ER}
	}
	\caption{Effect of neighborhood loss tolerance threshold $\eps$ on reconstruction time for the synthetic graphs with $n=10,000$ nodes and the number of edges as shown. \inter has the lowest reconstruction time.}
	\label{fig:effect of eps on reconstruction time using synthetic graphs}
\end{figure*}

Fig.~\ref{fig:effect of eps synthetic graphs} shows the effect of \eps on the compression ratio of synthetic graphs, BA and ER. As $\eps$ increases, the compression ratio decreases, since larger $\eps$ values allow greater compression. For BA graphs (Fig.\ref{fig:eps_cr_AB}), denser graphs undergo more compression, leading to a lower compression ratio. For ER graphs (Fig.~\ref{fig:eps_cr_ER}), significant compression occurs only beyond a higher $\eps$ threshold and primarily for dense graphs. For both BA graphs and ER graphs, \uni offers better compression than \inter.  Interestingly, \uni outperforms \iu as the \eps increases. The poor performance of \inter in case of ER graphs is similar to what we observed in Sec.~\ref{sec:exp compression ratio lossy} and the explanation for the same is already presented there.

Next, we analyze the effect of $\eps$ on the compression time and the reconstruction times of synthetic BA and ER graphs. Fig.~\ref{fig:effect of eps on ct using synthetic graphs} shows the compression time results for BA and ER graphs. For all the three \cn variants, the compression time increases as graphs become larger and denser, due to the greater number of eligible node pairs with positive compression gain. Fig.~\ref{fig:effect of eps on reconstruction time using synthetic graphs} shows that the reconstruction time increases as the graphs become larger and denser. Among the three \cn variants, \inter has the lowest reconstruction time. This is because  during reconstruction, \inter regenerates a subgraph of the original graph, while producing the same node set. On the other hand, \iu reconstructs the original graph exactly and \uni produces the original graph along with additional false positive edges. This effect gets more pronounced with higher loss thresholds. 

Comparing the three variants of \cn in terms of compression, it is important to note that while the loss of the lossy schemes is configurable (both in terms of loss type as well as maximum tolerable neighborhood loss), the lossless scheme \iu is not configurable. Interestingly, as shown in Sec.~\ref{sec:exp compression ratio}, there are graphs where \iu outperforms \inter; then, there are graphs where \inter outperforms \uni; and, there are graphs where \uni outperforms \inter and \iu. Therefore, none of the three variants is universally the best or universally the worst. This is an important finding since this shows why the configurability offered by the \cn framework (in terms of \inter, \uni and \iu) is necessary for graph summarization.

\subsection{Query Processing through \cn}
\label{sec:exp query processing}

The next set of experiments evaluates the query performance and the query response times of \cn (research question RQ5). Since \iu is lossless, it answers all the queries exactly with no loss. Thus, for evaluation of query performance, we only consider the lossy variants of \cn, i.e., \inter and \uni.

Given that the chosen lossy baseline scheme, SSumM \cite{lee2020ssumm} does not support queries, and the source codes of other lossy baseline schemes that support queries are not available, it is not possible to compare the query processing results against any existing baseline. 
As stated in Sec.~\ref{sec:query}, though the \cn framework can be employed to answer arbitrary graph queries, we evaluate its query performance on 3 representative graph queries, namely, the neighborhood queries, the shortest path queries, and the reachability queries.

\subsubsection{\bf Neighborhood Queries}

To evaluate the neighborhood queries, we measure the neighborhood loss
$\rl(u)$ (stated in Eq.~\eqref{eq:varrl}) for each node $u$ in the
reconstructed graph $G_r$, and then compute the \emph{average
neighborhood loss}, denoted by $anl$. The $anl$ values along with the
standard deviation are shown against the maximum possible neighborhood
loss \eps in Table~\ref{tab:anl}. The average neighborhood loss $anl$
(of either \inter or \uni) is typically much lower than $\eps$.  For
example, when $\eps = 0.5$, the average loss for \inter for the three
datasets are $0.09$, $0.25$ and $0.18$, i.e., at most half of the
threshold.  \uni also shows an average loss of only at most $0.26$ for the
threshold $0.5$.  The maximum value of the loss (values not shown in the
table) actually reaches the limit \eps for the three datasets and the
two variants.  This indicates that both \inter and \uni strive to offer
the best compression at the cost of maximum neighborhood loss that is
bounded by $\eps$.

\begin{table}[ht]
	\centering
	\begin{tabular}{c cc cc cc}
		\toprule
		\multirow{2}{*}{$\eps$} & \multicolumn{6}{c}{\bf Average neighborhood loss $anl$} \\
		\cline{2-7}
		& \inter EE & \uni EE & \inter GF & \uni GF & \inter AP & \uni AP \\
		\midrule
		0.10 & $0.00 \pm 0.00$ & $0.00 \pm 0.00$ & $0.00 \pm 0.01$ & $0.00 \pm 0.01$ & $0.01 \pm 0.02$ & $0.00 \pm 0.01$ \\
		0.20 & $0.00 \pm 0.03$ & $0.00 \pm 0.02$ & $0.00 \pm 0.02$ & $0.00 \pm 0.02$ & $0.02 \pm 0.05$ & $0.01 \pm 0.03$ \\
		0.25 & $0.02 \pm 0.07$ & $0.02 \pm 0.06$ & $0.01 \pm 0.03$ & $0.01 \pm 0.04$ & $0.05 \pm 0.08$ & $0.04 \pm 0.07$ \\
		0.30 & $0.02 \pm 0.07$ & $0.03 \pm 0.07$ & $0.01 \pm 0.04$ & $0.02 \pm 0.05$ & $0.06 \pm 0.09$ & $0.06 \pm 0.08$ \\
		0.40 & $0.07 \pm 0.12$ & $0.09 \pm 0.13$ & $0.02 \pm 0.06$ & $0.09 \pm 0.12$ & $0.12 \pm 0.13$ & $0.14 \pm 0.12$ \\
		0.50 & $0.18 \pm 0.20$ & $0.26 \pm 0.20$ & $0.09 \pm 0.15$ & $0.26 \pm 0.15$ & $0.25 \pm 0.18$ & $0.26 \pm 0.16$ \\
		0.60 & $0.22 \pm 0.22$ & $0.34 \pm 0.21$ & $0.16 \pm 0.18$ & $0.37 \pm 0.16$ & $0.32 \pm 0.19$ & $0.34 \pm 0.18$ \\
		0.70 & $0.31 \pm 0.26$ & $0.41 \pm 0.24$ & $0.27 \pm 0.22$ & $0.46 \pm 0.17$ & $0.42 \pm 0.21$ & $0.43 \pm 0.20$ \\
		0.75 & $0.37 \pm 0.29$ & $0.47 \pm 0.26$ & $0.37 \pm 0.25$ & $0.51 \pm 0.18$ & $0.49 \pm 0.22$ & $0.48 \pm 0.21$ \\
		0.80 & $0.41 \pm 0.31$ & $0.51 \pm 0.27$ & $0.46 \pm 0.25$ & $0.56 \pm 0.18$ & $0.53 \pm 0.23$ & $0.53 \pm 0.21$ \\
		0.90 & $0.47 \pm 0.35$ & $0.62 \pm 0.31$ & $0.65 \pm 0.25$ & $0.69 \pm 0.19$ & $0.63 \pm 0.26$ & $0.65 \pm 0.23$ \\
		\bottomrule
	\end{tabular}
	\caption{Average neighborhood loss $anl$ vs. $\eps$: Average neighborhood losses are much lower than the maximum allowed neighborhood loss threshold.}
	\label{tab:anl}
\end{table}

\begin{figure*}[t]
	\subfloat[\inter]
	{
		\includegraphics[width=0.45\linewidth]{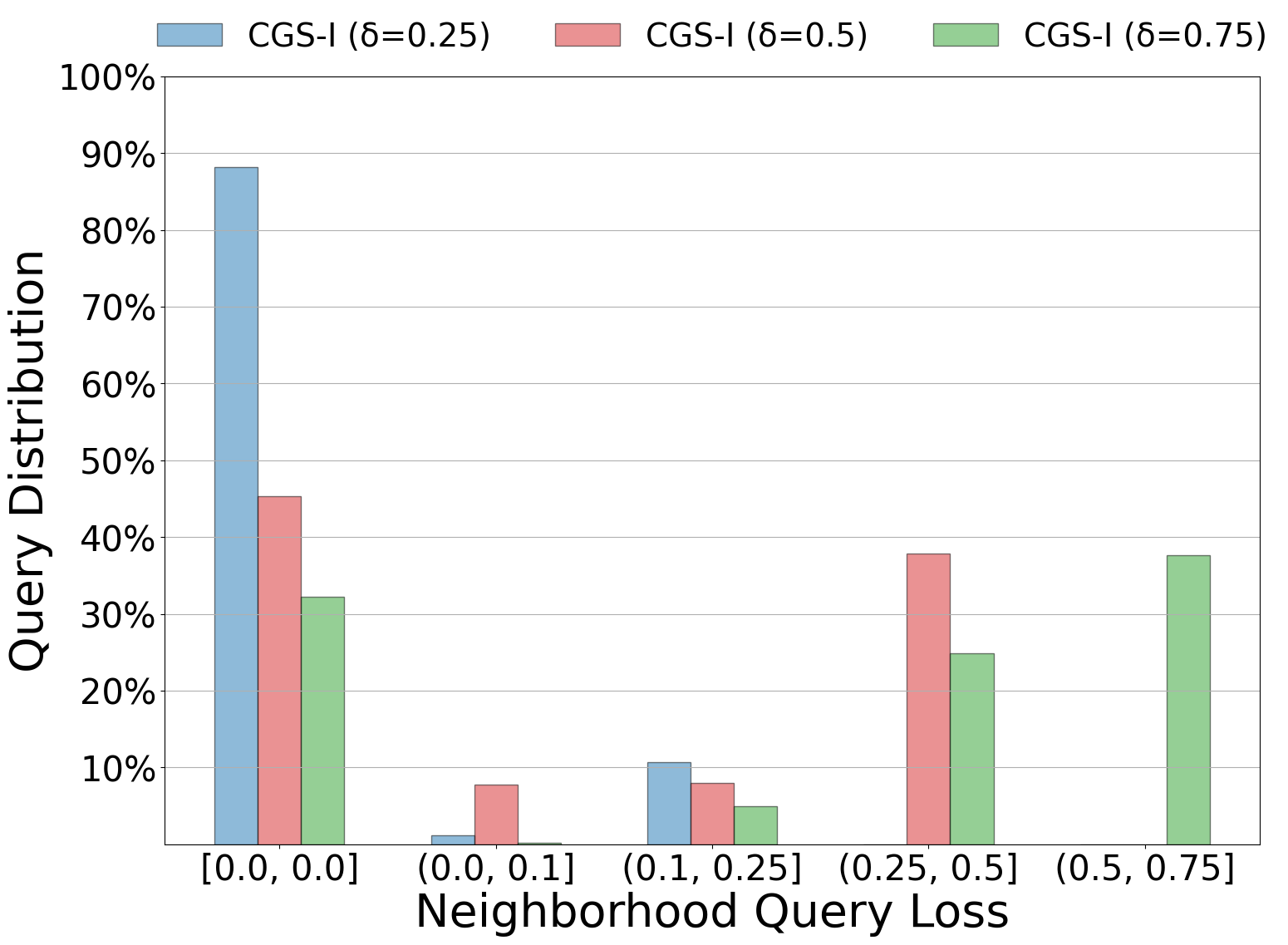}
		\Description[This figure shows the effect of Neighborhood query loss and the Query Distribution (in percentage) for \inter. The $x$-axis bins the different values of loss $\eps_u$ (Eq.~\ref{eq:nlt}) for the vertices, while the $y$-axis shows the percentage of queries that show the corresponding error.]{A high percentage of queries exhibit zero error (the first set of bars). \inter generally performs better or equivalent to \uni.}
		\label{fig:nq_histo_NOFAP}
	}
	\subfloat[\uni]
	{
		\includegraphics[width=0.45\linewidth]{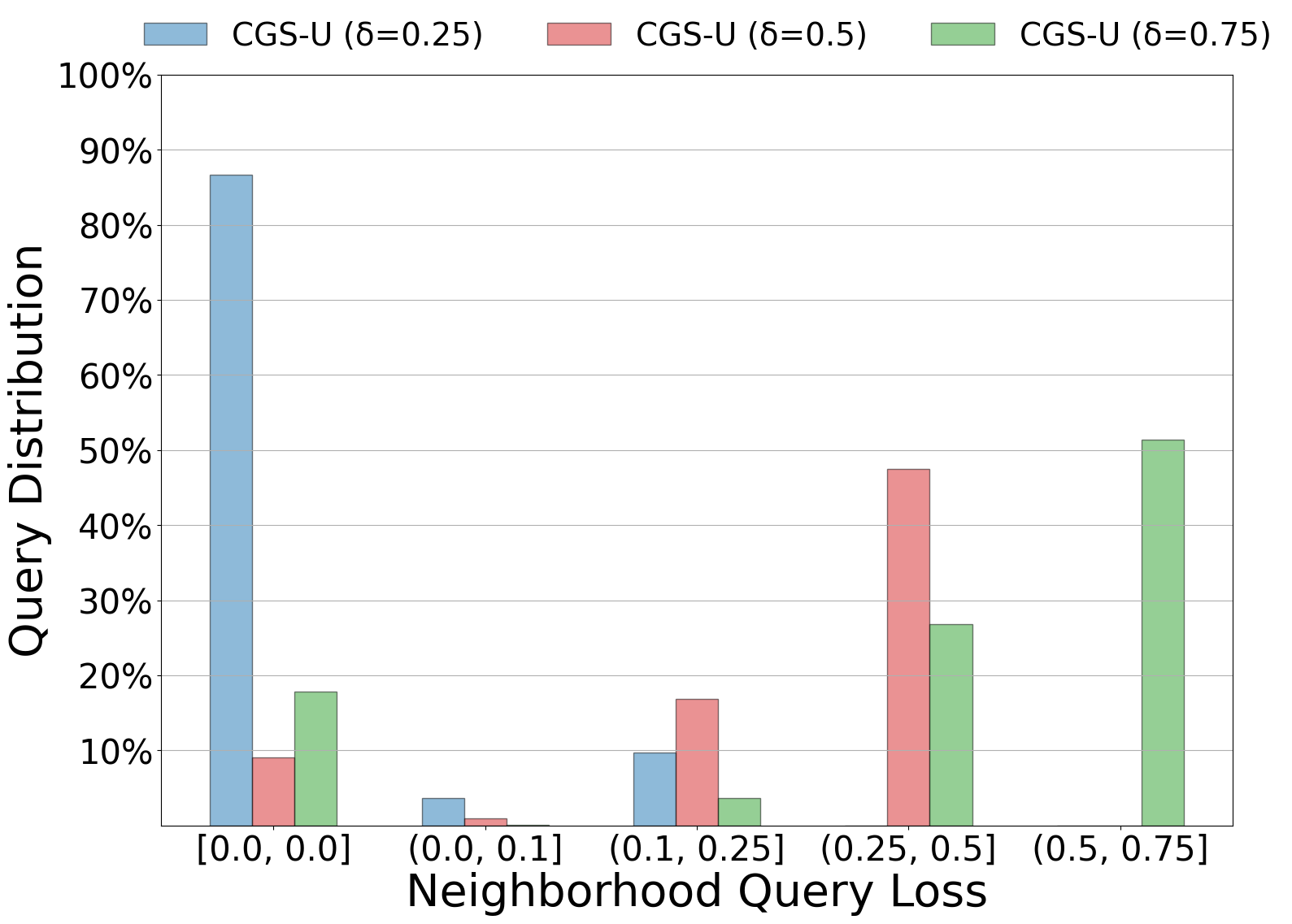}
		\Description[This figure shows the effect of Neighborhood query loss and the Query Distribution (in percentage) for \uni. The $x$-axis bins the different values of loss $\eps_u$ (Eq.~\ref{eq:nlt}) for the vertices, while the $y$-axis shows the percentage of queries that show the corresponding error.]{A high percentage of queries exhibit zero error (the first set of bars). \inter generally performs better or equivalent to \uni.}
		\label{fig:nq_histo_NOFAN}
	}
	\caption{Neighborhood query loss, EE: Most queries are answered with zero loss.}
	\label{fig:NQ Accuracy}
\end{figure*}

Fig.~\ref{fig:NQ Accuracy} shows the neighborhood query loss of \inter
and \uni for three different values of $\eps$.  The $x$-axis bins the different
values of loss $\eps_u$ (Eq.~\ref{eq:nlt}) for the vertices, while the $y$-axis
shows the percentage of queries that show the corresponding error.  A high
percentage of queries exhibit zero error (the first set of bars).  \inter generally performs better or equivalent to \uni.  When \eps
increases, the loss increases and hence, the number of queries with larger errors also increases.
For instance, for $\eps=0.25$, both \inter and \uni retrieve more than $85\%$
of neighborhood queries without any error, and 
$\sim10\%$ of queries with error
$< 0.25$.

\begin{figure*}[t]
	\includegraphics[width=0.45\linewidth]{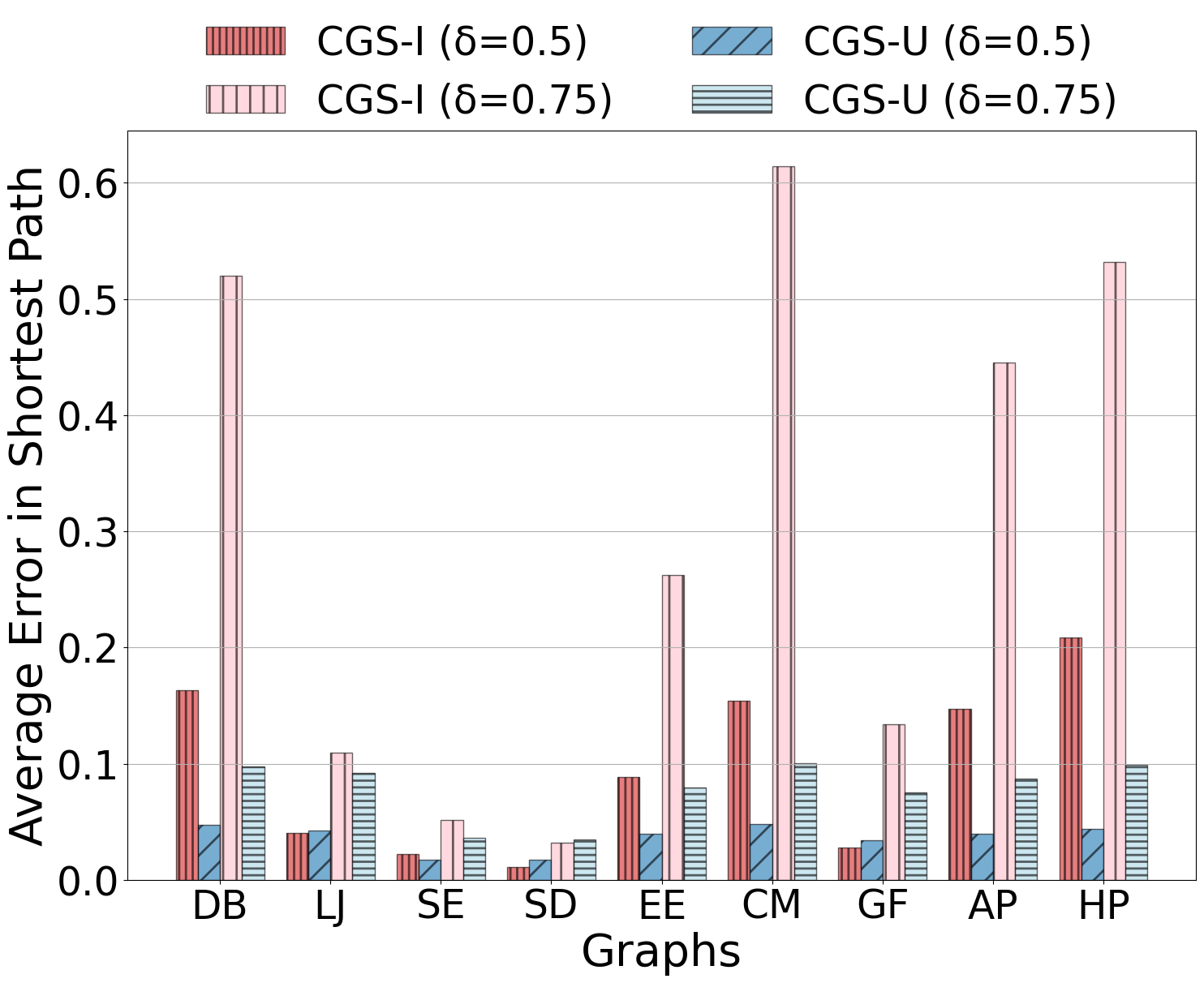}
	\Description[This figure demonstrates the average errors in shortest path query for \inter and \uni on all the real graphs.]{\inter performs better than \uni. For $\eps=0.5$, error in most datasets $< 0.2$.}
	\caption{Average error in shortest path query: Overall, \inter performs better than \uni.}
	\label{fig:sp_error}
\end{figure*}

\begin{figure*}[t]
	\subfloat[\inter]
	{
		\includegraphics[width=0.45\linewidth]{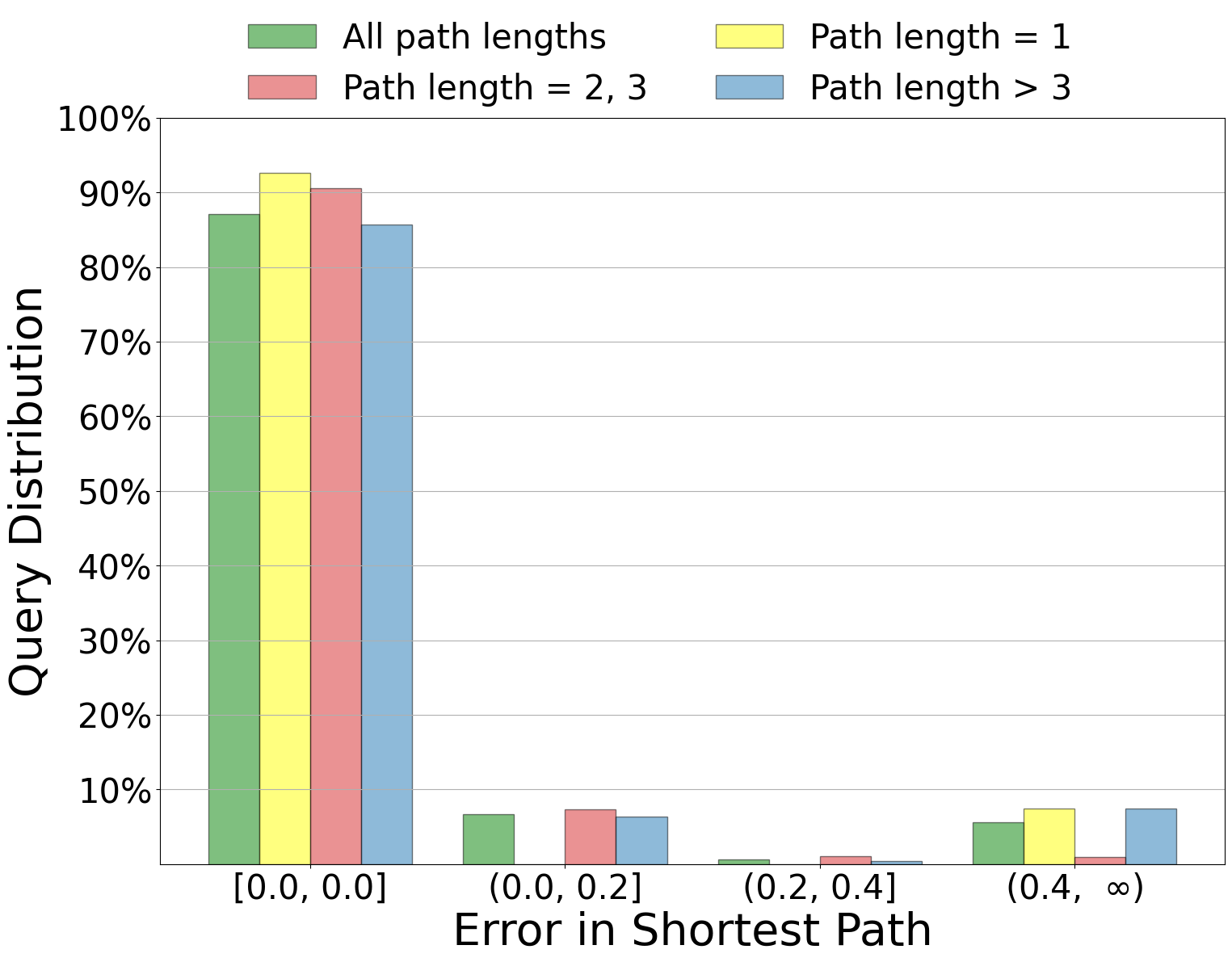}
		\Description[This figure shows the error in shortest path at $\eps=0.5$ for the EE dataset for the \inter version. Most queries return the exact shortest path length.]{A very high percentage of queries
(more than 80\% for both \inter and \uni) are evaluated exactly. Moreover, error in shorter path lengths are less.}
		\label{fig:sp_freq_EE_NOFAP_05}
	}
	\subfloat[\uni]
	{
		\includegraphics[width=0.45\linewidth]{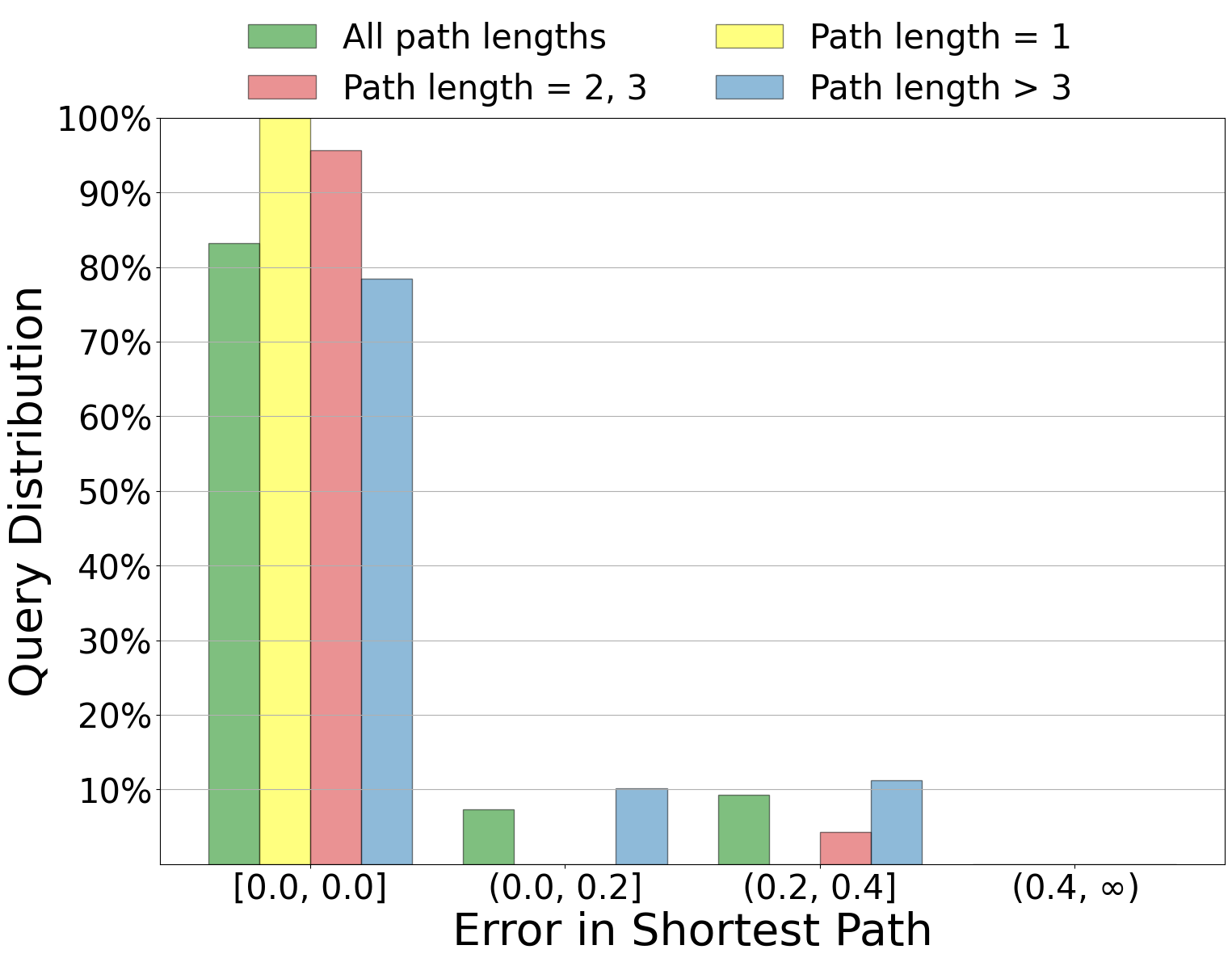}
		\Description[This figure shows the error in shortest path at $\eps=0.5$ for the EE dataset for the \uni version. Most queries return the exact shortest path length.]{A very high percentage of queries
(more than 80\% for both \inter and \uni) are evaluated exactly. Moreover, error in shorter path lengths are less.}
		\label{fig:sp_freq_EE_NOFAN_05}
	}
	\caption{Error in shortest path, $\eps=0.5$, EE: Most queries return the exact shortest path length.}
	\label{fig:sp query accuracy}
\end{figure*}

\begin{figure*}[t]
  	\includegraphics[width=0.45\linewidth]{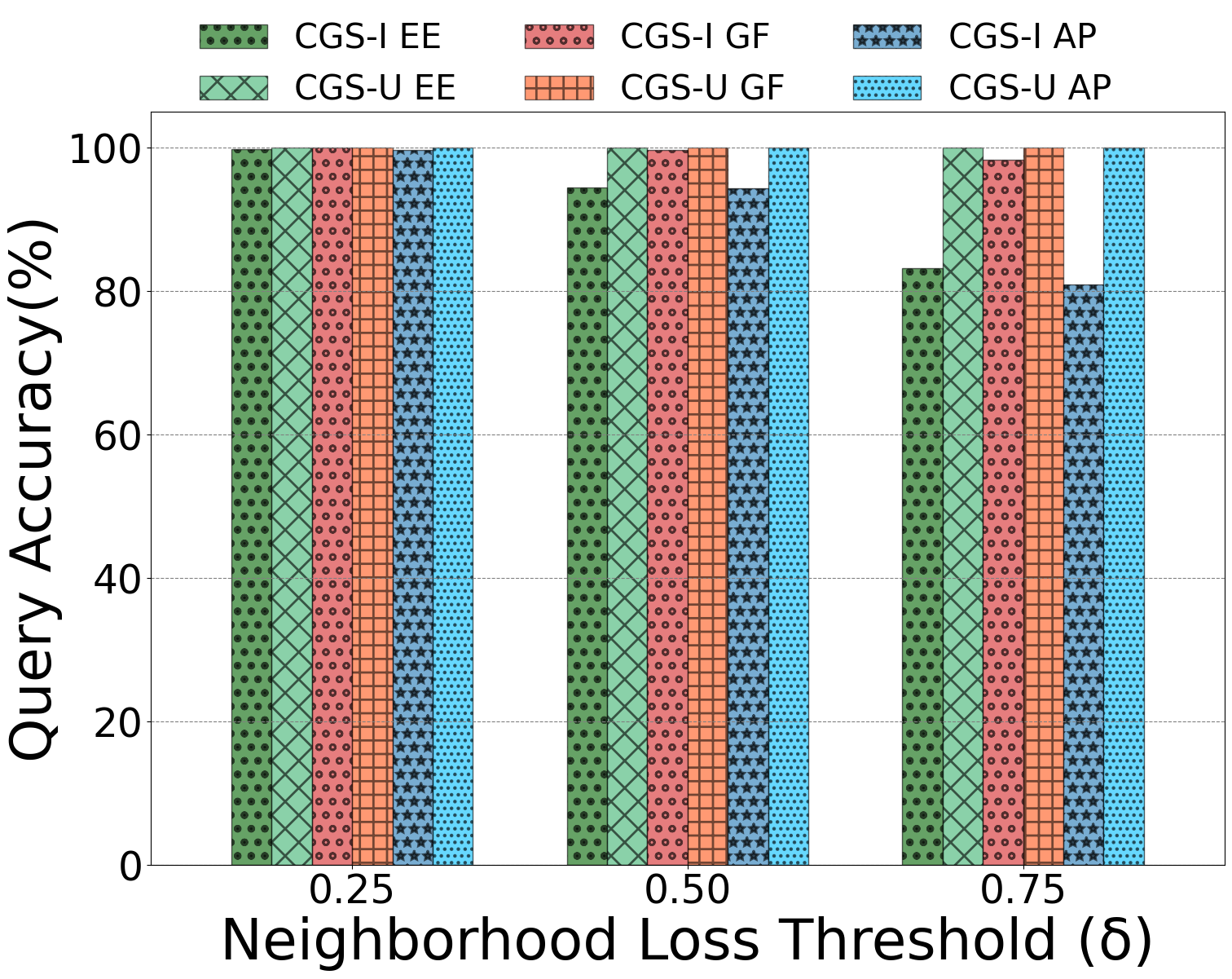}
	\Description[This figure shows the percentage of reachability queries that are reported correctly over a set of 100,000 random
reachability queries for the EE, GF and AP datasets. \inter is accurate at low error thresholds while \uni is always accurate.]{The accuracy of \uni is always $100\%$ for all values of $\eps$ and all datasets.The accuracy of \inter declines as \eps increases.}
	\caption{Reachability query accuracy: \inter is accurate at low error thresholds while \uni is always accurate.}
	\label{fig:rq_acc}
\end{figure*}

\begin{figure*}[t]
	\subfloat[Neighborhood query]
	{
		\includegraphics[width=0.45\linewidth]{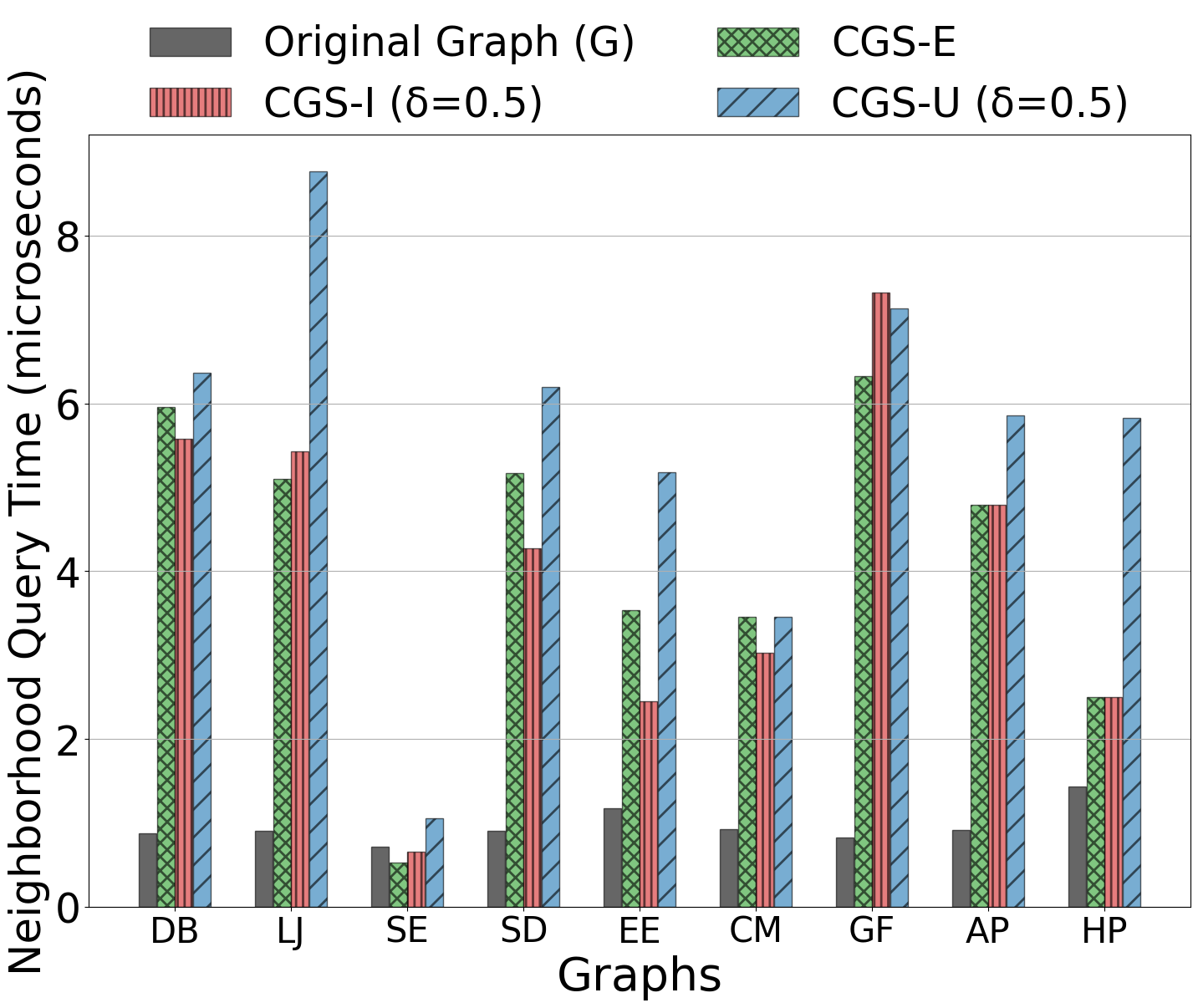}
		\Description[This figure demonstrates the query response time for the neighborhood query in microseconds for all the real graphs.]{This is done for all the three variants using local decompression i.e. when the queries are run directly on
the summary graph.}
		\label{fig:nq_time}
	}
	\subfloat[Shortest path query]
	{
		\includegraphics[width=0.45\linewidth]{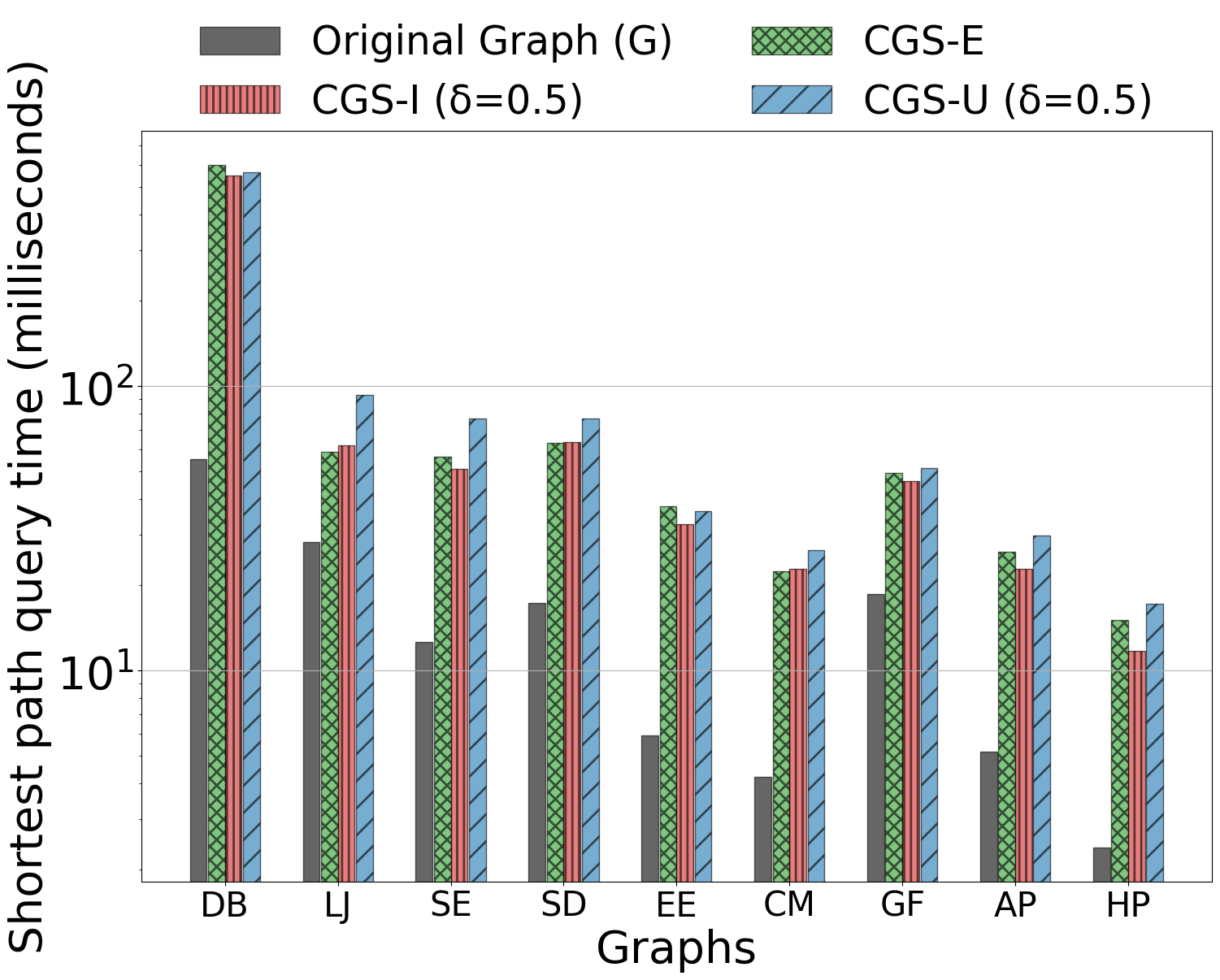}
		\Description[This figure demonstrates the query response time for the shortest path query for all the real graphs.]{This is done for all the three variants using local decompression i.e. when the queries are run directly on
the summary graph. While the neighborhood queries took microseconds to run, the shortest
path queries finish in milliseconds.}
		\label{fig:sp_time}
	}
	\caption{Query response times: Queries on the summary graphs are slower due to local decompression.}
	\label{fig:query performance}
\end{figure*}
	
\begin{figure*}[t]
	\includegraphics[width=0.45\linewidth]{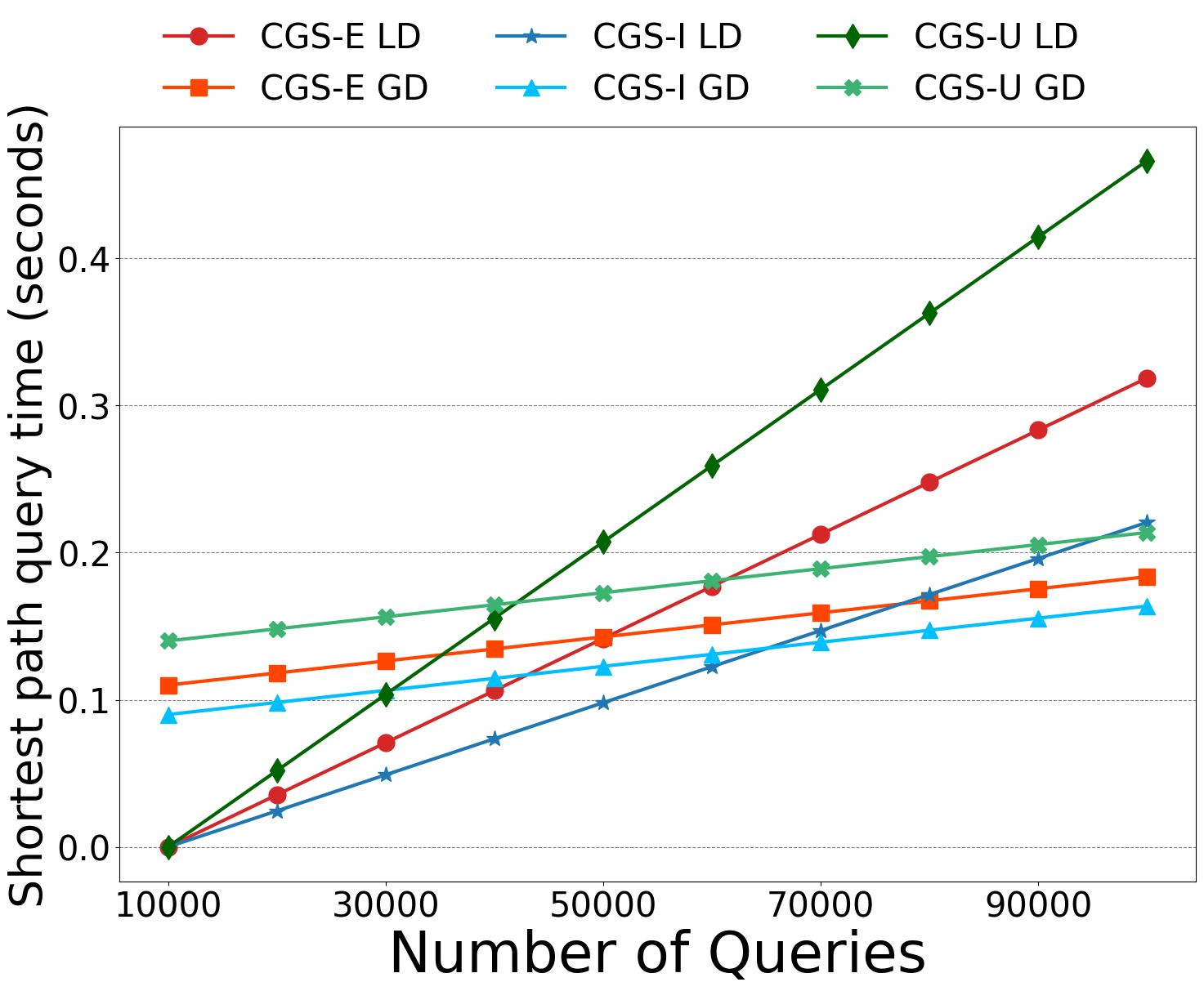}
	\Description[This figure demonstrates the time required for varying number of shortest path queries in seconds for the Local Decompression(LD) vs. Global Decompression(GD) for all the three variants. When the number of queries is large, global decompression is better than local decompression.]{For \iu, the savings in global decompression starts only when 50,000 queries are run. The cutoff points for \uni and \inter are around 40,000 and 65,000, respectively.}
	\caption{Local (LD) vs. global decompression (GD), EE: When the number of queries is large, global decompression is better than local decompression.}
	\label{fig:sp_time_GD_LD}
\end{figure*}

\subsubsection{\bf Shortest Path Queries}

To evaluate shortest path queries, 100,000 random pairs of nodes are chosen for
each graph dataset. For each chosen pair $(u,v)\in G$, the error in shortest
path is computed as $error(u,v)=|\dist_{G}(u,v) -
\dist_{G_r}(u,v)|/\dist_{G}(u,v)$, where $\dist_{G}$ is the (actual) shortest path
distance in $G$, and $\dist_{G_r}$ is the (approximate) shortest path in the reconstructed graph $G_r$.
Fig.~\ref{fig:sp_error} shows the average errors for \inter and \uni on all the
graphs.  For $\eps=0.5$,
error in most datasets $< 0.2$.
Fig.~\ref{fig:sp query accuracy} shows that a very high percentage of queries
(more than 80\% for both \inter and \uni) are evaluated exactly.  Moreover,
error in shorter path lengths are less. This shows the effectiveness of query
processing of \cn.

\subsubsection{\bf Reachability queries}

Fig.~\ref{fig:rq_acc} shows the percentage of reachability queries that are reported correctly over a set of 100,000 random
reachability queries. These queries are constructed by selecting 100,000 random and unique node pairs $(u, v)$ from the graph, ensuring $u \ne v$. For each pair, reachability is evaluated on both the original and the reconstructed graph. The results are then compared to obtain the accuracy of reachability preservation in the summary graph. The reachability query accuracy measures how often the summary graph correctly preserves the reachability compared to the original graph for node pairs that are actually reachable. The accuracy of \uni is always $100\%$ for all values of
$\eps$ and datasets. This follows from Table~\ref{tab:query properties}.
On the other hand, the accuracy of \inter declines as \eps increases, since it
leads to more disconnected components, as explained in Sec.~\ref{sec:query
characteristics}. Nevertheless, \inter offers fairly high accuracy for all
values of $\eps$.

\subsubsection{\bf Query Response Times}

Fig.~\ref{fig:query performance} shows the response times of the different
queries using local decompression, i.e., when the queries are run directly on
the summary graph.  (Times for reachability queries are not shown as they are similar to those for the shortest
path queries.) While the neighborhood queries take microseconds to run, the shortest
path queries finish in milliseconds. The queries run faster on the original graph, since they do not require local decompression.

\subsubsection{\bf Local versus Global Decompression}

The next set of experiments seeks to answer the question: which decompression
method to use -- local or global. While local decompression does not require reconstructing the entire graph; if, however, the graph is reconstructed, the queries run faster, since local decompression operations are not needed. There is, thus, a trade-off.  Fig.~\ref{fig:sp_time_GD_LD} shows
the times required for varying number of shortest path queries.  For \iu, the
savings in global decompression starts only when 50,000 queries are run.
Assuming, the time taken for running a query on the locally decompressed graph is $t_l$, that in the reconstructed graph is $t_g$ and the time to reconstruct the graph is $t_r$, we get $50000 \times t_l \approx t_r + 50000 \times t_g$. The cutoff points for \uni and \inter are around 40,000 and 65,000, respectively.

\begin{figure*}[t]
\minipage{0.45\textwidth}
	\includegraphics[width=\linewidth]{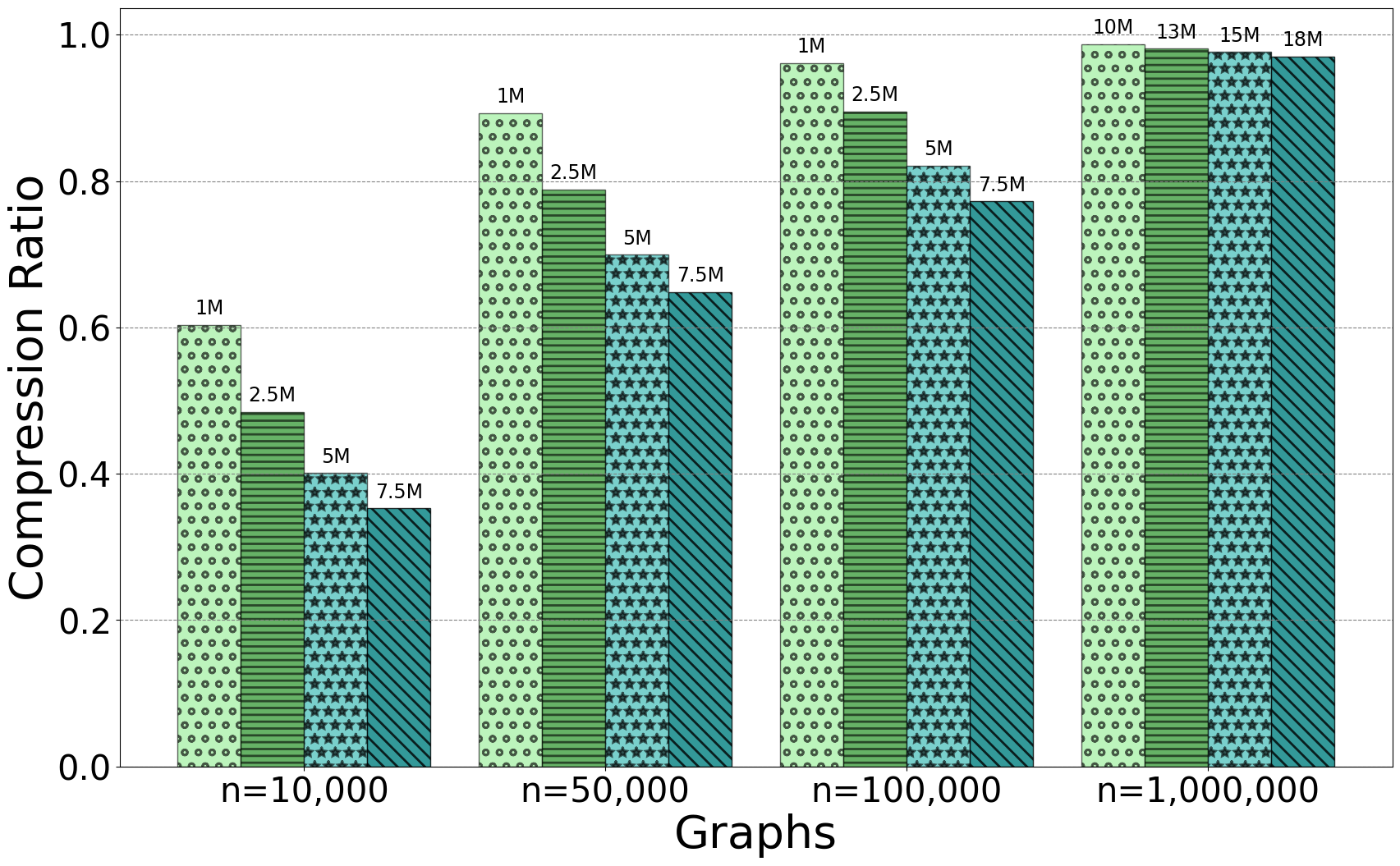}
	\Description[This figure shows the scalability performance of \iu in terms of compression ratio on BA graphs of varying sizes of nodes. Compression is better for denser graphs where the number of edges $m$ is shown above each bar, where $M$ denotes one million).]{For a given number of nodes such as $10^5$, the compression improves as the number of edges increases. As mentioned earlier, \cn offers better compression for denser graphs.}
	\caption{Scalability evaluation of CGS-E in terms of compression ratio on synthetic BA graphs: Compression is better for denser graphs (The number of edges $m$ is shown above each bar, where $M$ denotes one million).}
	\label{fig:AB_cr_scalability}
\endminipage
\hfill
	{\ }
\minipage{0.45\textwidth}
	\includegraphics[width=\linewidth]{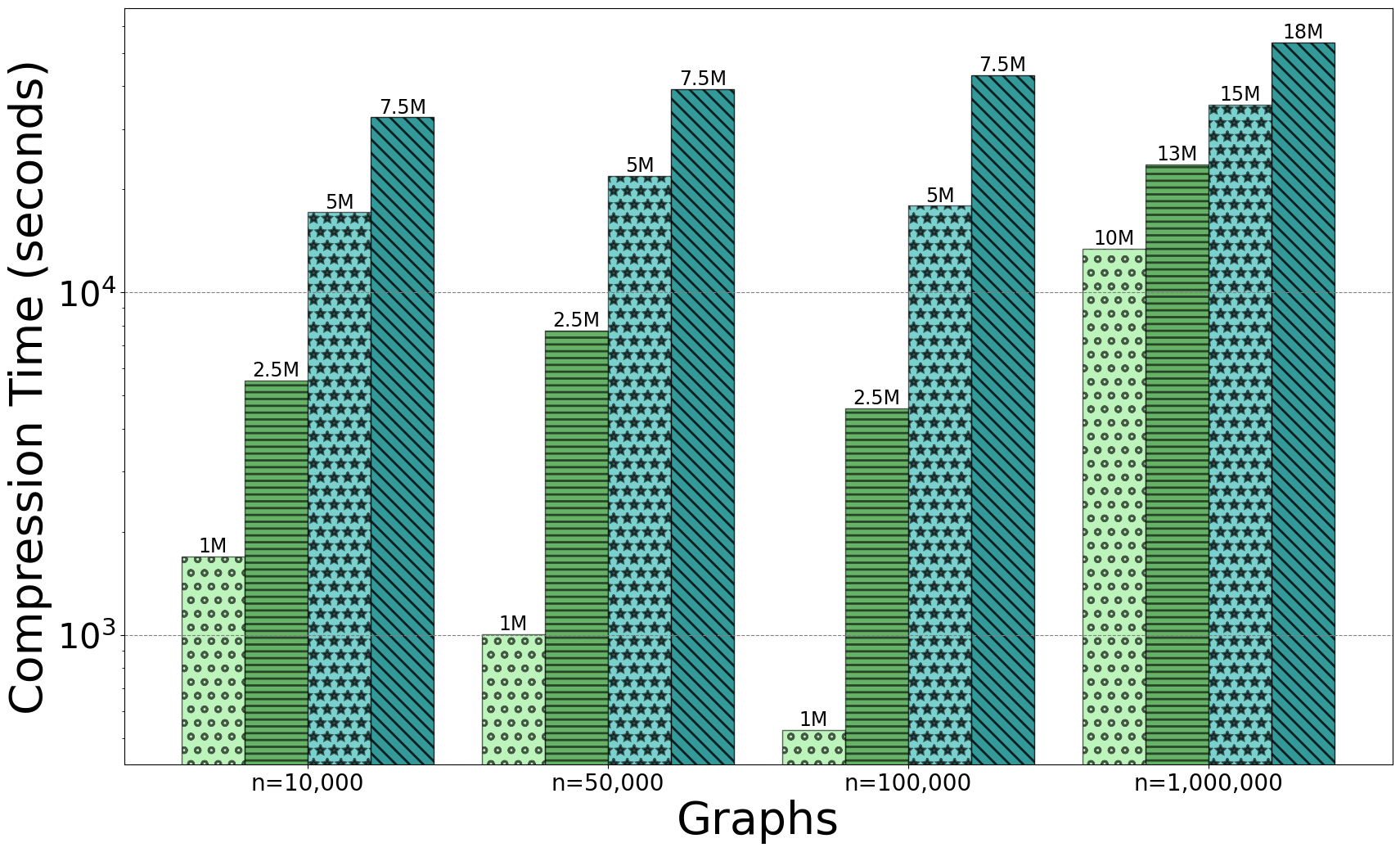}
	\Description[This figure demonstrates the compression time(in seconds) results of \iu on the BA graphs. As the graphs become larger and denser, the compression time increases.]{The compression time increases owing to larger number of eligible pairs of nodes with positive compression gain.}
	\caption{Scalability evaluation of \iu in terms of compression time on synthetic BA graphs: More time is required to compress the denser graphs.}
	\label{fig:AB_ct_scalability}
	\endminipage\hfill
\end{figure*}

\begin{figure*}[t]
\minipage{0.45\textwidth}
	\includegraphics[width=\linewidth]{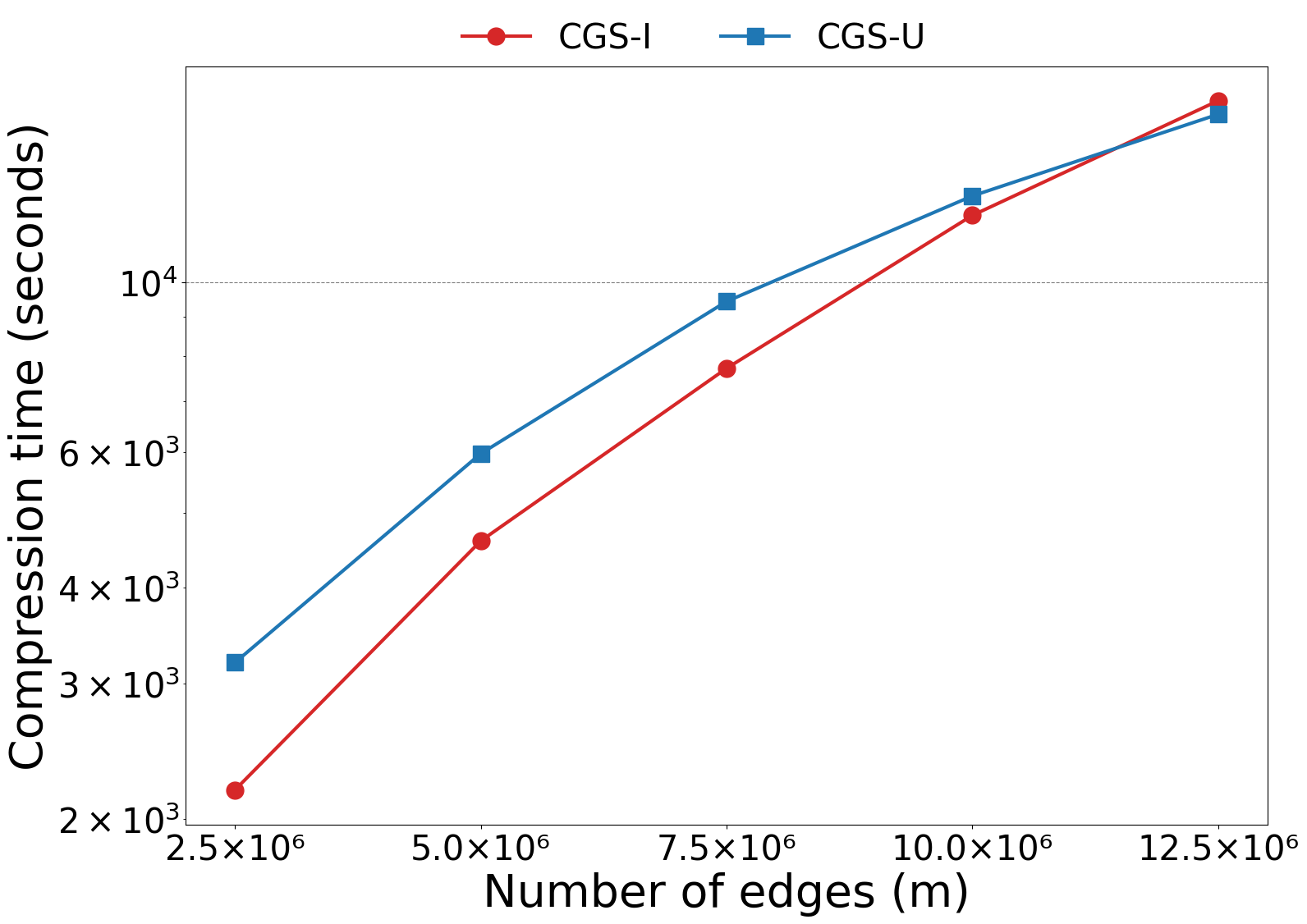}
	\Description[This figure demonstrates the scalability evaluation of lossy schemes - \uni and \inter in terms of compression time (in seconds) on synthetic BA graphs. Denser graphs take more time to compress.]{As the number of edges increases, the compression time increases.  Overall, the cost of \cn framework on million scale graph varies from few minutes to few hours.}
	\cutfigabove
\caption{Scalability evaluation of lossy schemes in terms of compression time on synthetic BA graphs: Denser graphs take more time to compress.}
	\label{fig:AB_lossy_ct_scalability}
\endminipage
\hfill
	{\ }
\minipage{0.45\textwidth}
	\includegraphics[width=\linewidth]{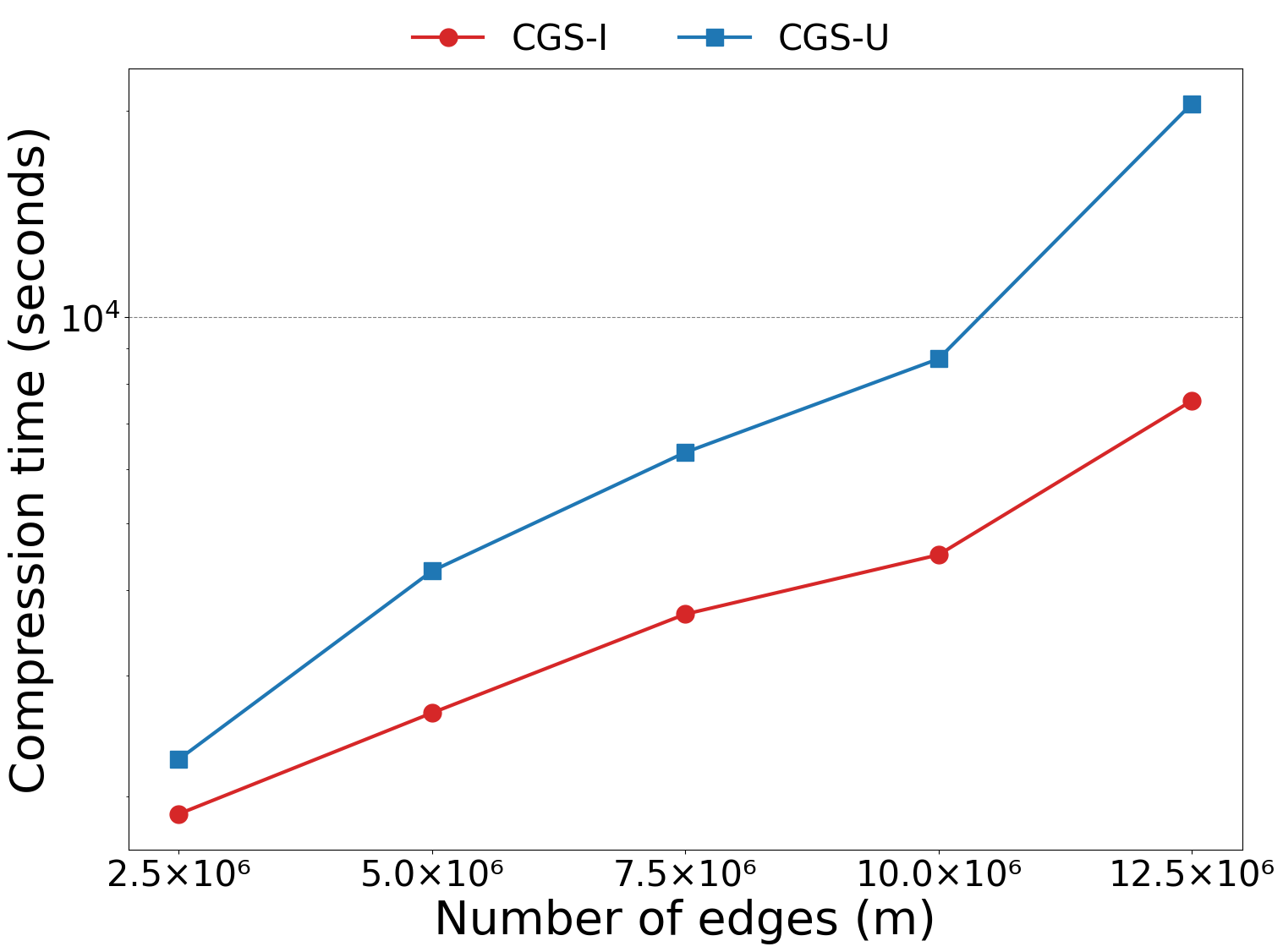}
	\Description[This figure demonstrates the scalability evaluation of lossy schemes - \uni and \inter in terms of compression time (in seconds) on synthetic ER graphs. Denser graphs take more time to compress.]{As the number of edges increases, the compression time increases.  Overall, the cost of \cn framework on million scale graph varies from few minutes to few hours.}
	\caption{Scalability evaluation of lossy schemes in terms of compression time on synthetic ER graphs: Denser graphs take more time to compress.}
	\label{fig:RG_lossy_ct_scalability}
	\endminipage\hfill
\end{figure*}

\begin{figure}[t]
	\subfloat[Lossless schemes]
	{
		\includegraphics[width=0.49\linewidth]{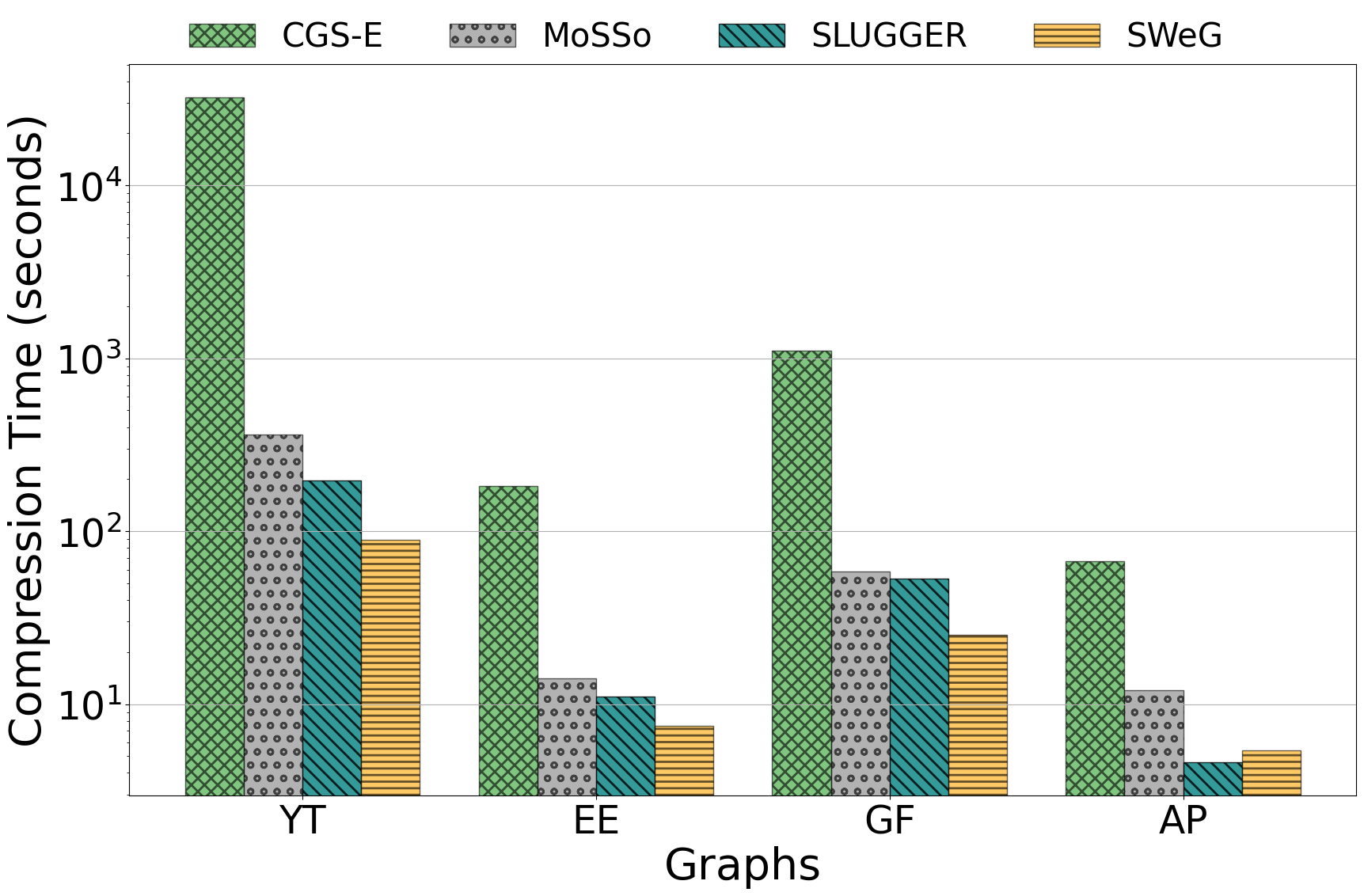}
		\Description[This figure shows the compression time(in seconds) for lossless schemes on 4 real datasets i.e. YT, EE, GF and AP. The results on other datasets are similar. The compression times of the \cn variants are higher than its competitors i.e. MoSSo, SLUGGER and SWeG.] {Though the compression times are higher, the absolute times are practical. Since summarization is typically an offline and infrequent activity, the higher running time of \cn does not reduce its practical applicability.} 
	}
	\subfloat[Lossy schemes]
	{
		\includegraphics[width=0.49\linewidth]{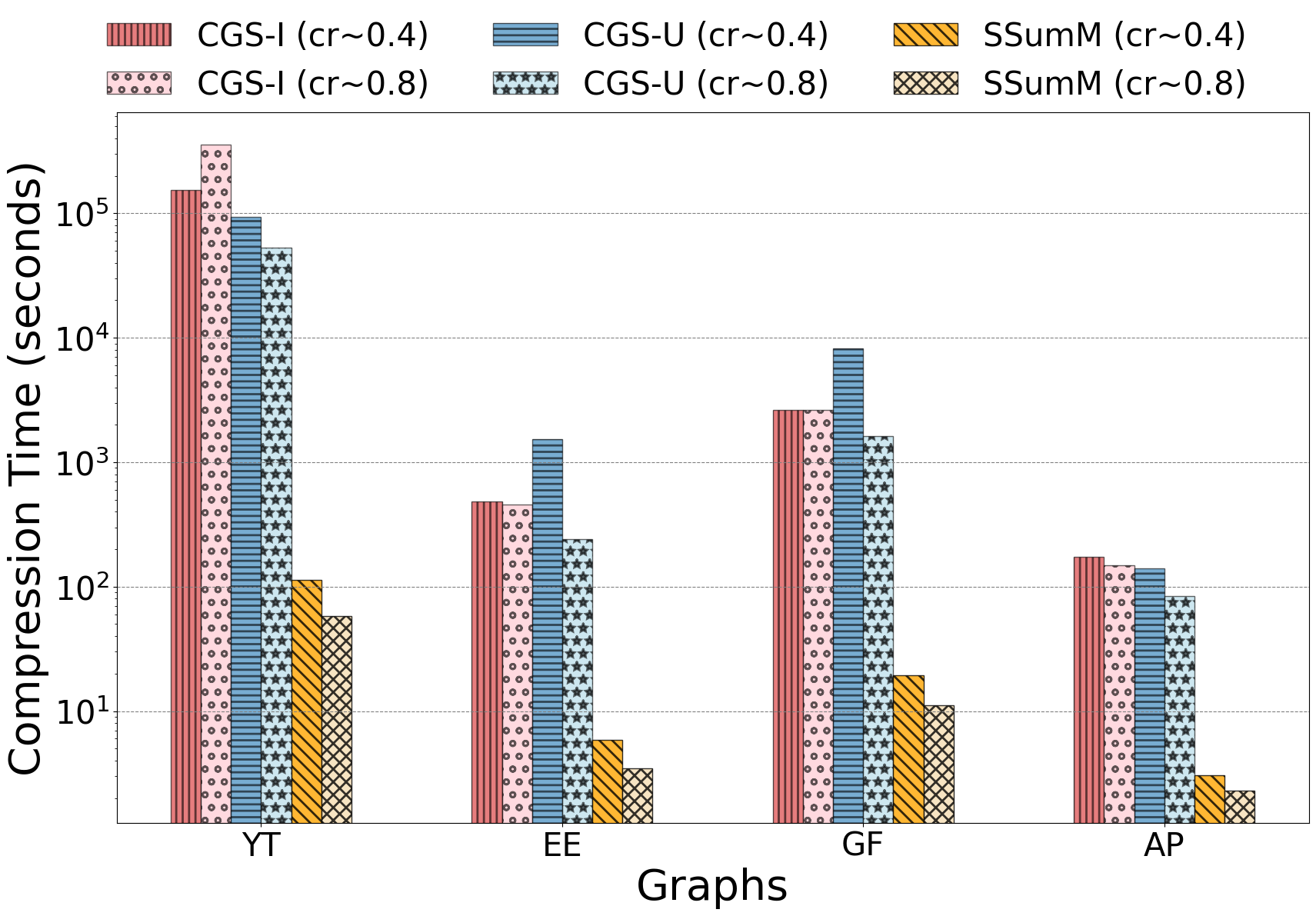}
		\Description[This figure shows the compression time(in seconds) for lossy schemes on 4 real datasets i.e. YT, EE, GF and AP. The results on other datasets are similar. The compression times of the \cn variants are higher than its competitor i.e. SSumM.] {Though the compression times are higher, the absolute times are practical. Since summarization is typically an offline and infrequent activity, the higher running time of \cn does not reduce its practical applicability. For the lossy schemes, higher compression, i.e., lower compression ratio typically comes at a cost of higher compression time.} }
	\caption{Scalability evaluation in terms of compression time on real graphs: \cn algorithms are slower but practical for real graphs.}
	\label{fig:ct results}
\end{figure}

\subsection{Scalability Results}
\label{sec:exp scalability}

This section addresses the research question RQ6. To assess the scalability of the \cn framework, we evaluate on both synthetic as well as real world graphs. For evaluation on synthetic datasets, we consider the BA graphs and the ER graphs of different sizes.  Fig.~\ref{fig:AB_cr_scalability} shows the scalability performance of \iu in terms of compression ratio on BA graphs of varying sizes. For a given number of nodes such as $10^5$, the compression improves as the number of edges increases. On the other hand, for a given number of edges such as $10^6$, the compression decreases as the number of nodes increases. This agrees with our earlier observation in Sec.~\ref{sec:exp compression ratio}, that \cn offers better compression for denser graphs.

Fig.~\ref{fig:AB_ct_scalability} shows the compression time results of \iu on the same BA graphs. As the graphs become larger and denser, the compression time increases owing to larger number of eligible pairs of nodes with positive compression gain. 
Fig.~\ref{fig:AB_lossy_ct_scalability}  and Fig.~\ref{fig:RG_lossy_ct_scalability} shows the compression time results for the lossy variants on BA and ER graphs, respectively. The number of nodes is fixed to $10,000$ and number of edges varies from $2.5 \times 10^{6}$ to $12.5 \times 10^{6}$. Overall, the cost of \cn framework on million scale graph varies from few minutes to few hours.

Fig.~\ref{fig:ct results} shows the compression time results for lossless schemes and lossy schemes on 4 representative real datasets. The results on other datasets are similar. While the compression times of the \cn variants are higher than its competitors, the absolute times are practical. 
Given that \cn offers high configurability, along with superior
compression and support for arbitrary queries with fairly high accuracy, its high compression time comes as a trade-off. Moreover, since summarization is typically an
offline and infrequent activity, the higher running time of \cn does not reduce its practical applicability. 
For the lossy schemes, higher compression, i.e., lower compression ratio typically comes at a cost of higher compression time.

\begin{figure*}[t]
	\includegraphics[width=0.45\linewidth]{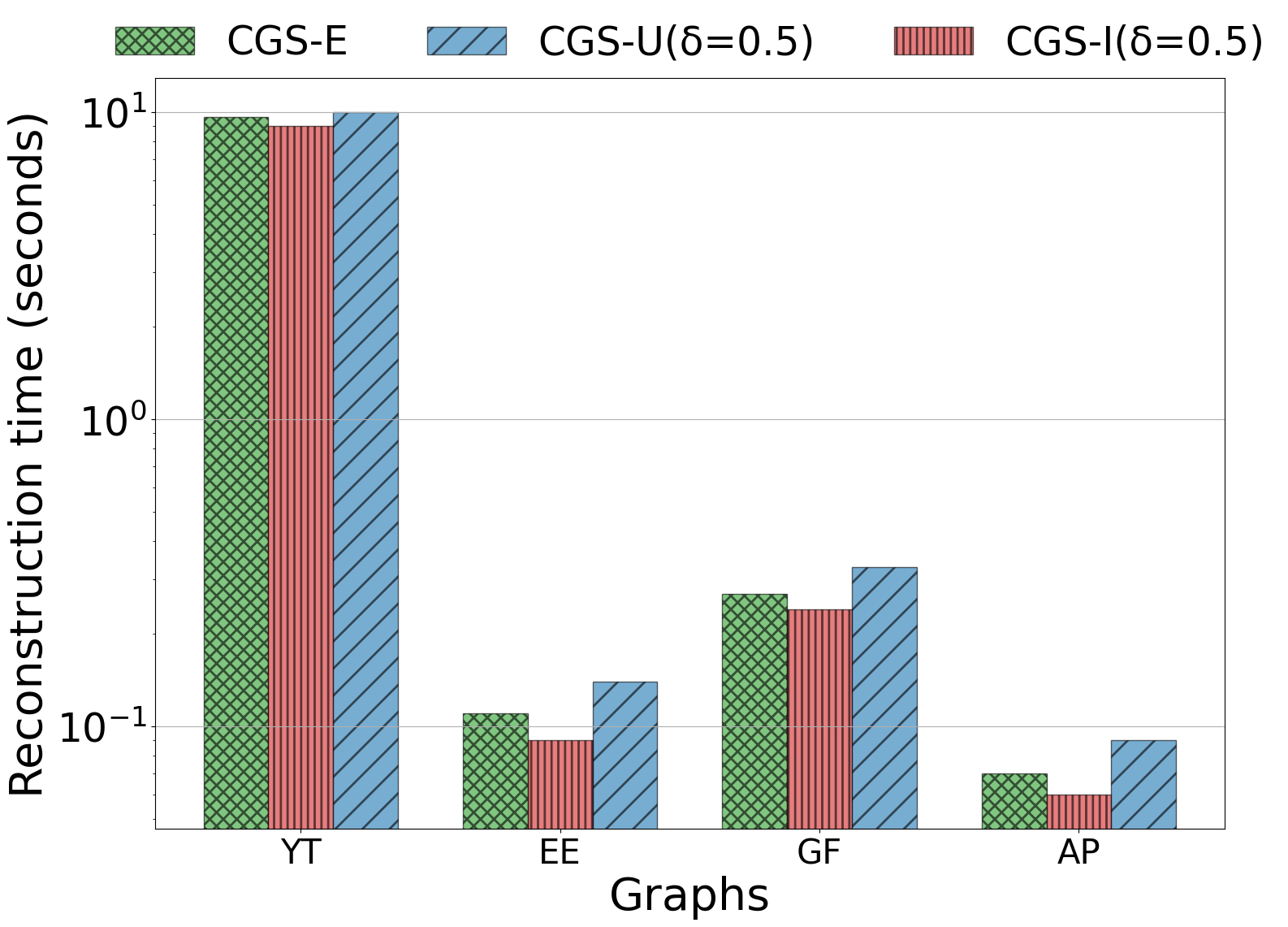}
	\Description[This figure shows the reconstruction time(in seconds) for the YT, EE, GF and AP graphs for all the three variants. Among the three \cn variants, \uni has the highest reconstruction time, followed by \iu, and \inter has the lowest.]{This behavior is justified in Sec.~\ref{sec:exp effect of neighborhood loss tolerance threshold}.}
	\caption{Reconstruction time results: \cn reconstruction times are practical.}
	\label{fig:reconstruction_time}
\cutfigbelow
\end{figure*}

Fig.~\ref{fig:reconstruction_time} shows the reconstruction time for the same graphs. Among the three \cn variants, \uni has the highest reconstruction time, followed by \iu, and \inter has the lowest. This behavior is justified in Sec.~\ref{sec:exp effect of neighborhood loss tolerance threshold}. 

\subsection{Summary of the Experiments}
\label{sec:exp summary}

We conclude the section by listing the key takeaways. \iu offers superior summarization on both real and synthetic graphs than all the state-of-the-art lossless schemes. For denser graphs, the compression is usually more. Among the lossy variants, \inter usually 
offers smaller summaries (by up to $46\%$) and lower reconstruction errors (by up to $44\%$) than SSumM \cite{lee2020ssumm} on real graphs. However, on synthetic BA and ER graphs, \uni offers the best compression.

Queries for \iu, by nature, are exact.  For neighborhood queries, while \inter is usually better than \uni, for shortest path queries, \uni is usually better than \inter.  For reachability queries, \uni is 100\% accurate, while \inter has a fairly high accuracy.
Though compression times of \cn are higher than its competitors, they are practical. Evaluation on several synthetic and real world graph datasets show that the \cn framework is scalable.  Further,  \inter and \uni cannot substitute each other because of their distinct summarization and query characteristics, as stated in Sec.~\ref{sec:formulation}. Finally, as observed in Sec.~\ref{sec:exp effect of neighborhood loss tolerance threshold}, among the three variants of \cn, there is no scheme that universally offers the most compressed summaries on all graphs. Therefore, each variant uniquely contributes towards the goal of configurable graph summarization.

\section{Conclusions}
\label{sec:conc}

In this paper, we proposed a configurable graph summarization framework \cn, based upon the common neighborhood of nodes. The framework offers three summarization variants. While \inter and \uni avoid false positive and false negative edges respectively, \iu is a lossless scheme. The lossy variants of \cn, \inter and \uni, allow a parameterized control over the neighborhood loss, by using a maximum loss tolerance threshold.
This enables the lossy schemes to reconstruct the input graph with bounded neighborhood loss. Moreover, \cn  can answer neighborhood queries with bounded quality guarantees, and other graph queries such as, shortest path queries and reachability queries with fairly high accuracy. Empirical evaluation on several real-world graphs confirm the efficacy and efficiency of \cn. In future, we plan to extend this framework to handle  dynamic graphs and streaming graphs.

\bibliographystyle{ACM-Reference-Format}
\bibliography{papers}

\end{document}